\documentclass[12pt,letterpaper]{article}
\pdfoutput=1
\makeatletter
\renewcommand\paragraph{\@startsection{paragraph}{4}{\z@}%
                                      {1ex \@plus1ex \@minus.2ex}%
                                      {-1em}%
                                      {\normalfont\normalsize\bfseries}}
\makeatother

\usepackage{authblk}

\usepackage[margin=1in]{geometry}

\usepackage{mathpazo}

\usepackage[shortlabels]{enumitem}

\usepackage{microtype}

\usepackage{tikz}
\usepackage{tikz-cd}
\usetikzlibrary{backgrounds,fit,decorations.pathreplacing,calc}
\usetikzlibrary{positioning,shapes.misc}
\usepackage{soul}

\usepackage{float}

\usepackage{subcaption}

\usepackage[colorlinks = true]{hyperref}
\definecolor{darkred}  {rgb}{0.5,0,0}
\definecolor{darkblue} {rgb}{0,0,0.5}
\definecolor{darkgreen}{rgb}{0,0.5,0}
\hypersetup{
  urlcolor   = blue,         
  linkcolor  = darkblue,     
  citecolor  = darkgreen,    
  filecolor  = darkred       
}

\usepackage{amsmath,amssymb,amsfonts,amsthm,amstext}
\usepackage{bbm}

\usepackage{etoolbox}

\usepackage{mathtools}
\mathtoolsset{centercolon}
\makeatletter
\protected\def\tikz@nonactivecolon{\ifmmode\mathrel{\mathop\ordinarycolon}\else:\fi}
\makeatother

\usepackage{tcolorbox}

\usepackage{algorithm}
\usepackage{algpseudocode}

\usepackage[noabbrev,capitalise]{cleveref}

\crefname{lemma}{Lemma}{Lemmas}
\crefname{equation}{equation}{equations}
\crefname{proposition}{Proposition}{Propositions}
\crefname{definition}{Definition}{Definitions}
\crefname{theorem}{Theorem}{Theorems}
\crefname{conjecture}{Conjecture}{Conjectures}
\crefname{corollary}{Corollary}{Corollaries}
\crefname{claim}{Claim}{Claims}
\crefname{section}{Section}{Sections}
\crefname{appendix}{Appendix}{Appendices}
\crefname{figure}{Fig.}{Figs.}
\crefname{table}{Table}{Tables}

\usepackage{tocloft}

\usepackage{multirow}
\usepackage{diagbox}

\newcommand{\ket}[1]{|#1\rangle}
\newcommand{\bra}[1]{\langle#1|}
\newcommand{\braket}[2]{\langle#1|#2\rangle}

\newcommand{\x}{\otimes}

\newcommand{\ct}{^{\dagger}}

\DeclarePairedDelimiter{\set}{\lbrace}{\rbrace}
\DeclarePairedDelimiter{\abs}{\lvert}{\rvert}
\DeclarePairedDelimiter{\norm}{\lVert}{\rVert}

\DeclarePairedDelimiter{\ip}{\langle}{\rangle}

\DeclareMathOperator{\spn}{span}

\DeclareMathOperator{\Tr}{Tr}
\DeclareMathOperator{\tTr}{\tilde{Tr}}

\DeclareMathOperator{\supp}{supp}
\newcommand{\incl}{\hookrightarrow}

\newcommand{\C}{\mathbb{C}}
\newcommand{\R}{\mathbb{R}}
\newcommand{\N}{\mathbb{N}}
\newcommand{\Z}{\mathbb{Z}}
\newcommand{\calH}{\mathcal{H}}
\newcommand{\calX}{\mathcal{X}}
\newcommand{\calY}{\mathcal{Y}}
\newcommand{\calA}{\mathcal{A}}
\newcommand{\calB}{\mathcal{B}}

\newcommand{\calN}{\mathcal{N}}
\newcommand{\calU}{\mathcal{U}}
\newcommand{\calL}{\mathcal{L}}
\newcommand{\calF}{\mathcal{F}}

\newcommand{\mbK}{\mathbb{K}}
\newcommand{\K}{\mathbb{K}}
\newcommand{\Q}{\mathbb{Q}}

\newcommand{\1}{\mathbb{1}}
\newcommand{\iso}{\cong}

\newcommand{\coRE}{\ensuremath{\mathsf{coRE}}}
\newcommand{\RE}{\ensuremath{\mathsf{RE}}}
\newcommand{\MIP}{\ensuremath{\mathsf{MIP}}}
\newcommand{\MM}{\ensuremath{\mathsf{MM}}}
\newcommand{\pMM}{\pmb{\ensuremath{\mathsf{MM}}}}

\newcommand{\PerfectStrat}{\ensuremath{\mathsf{PerfectStrategy}}}
\newcommand{\GapPerfectStrat}{\ensuremath{\mathsf{GappedPerfectStrategy}}}
\newcommand{\Membership}{\ensuremath{\mathsf{Membership}}}
\newcommand{\Intersection}{\ensuremath{\mathsf{Intersection}}}

\newcommand{\tP}{\tilde{P}}
\newcommand{\tQ}{\tilde{Q}}
\newcommand{\tM}{\tilde{M}}
\newcommand{\tN}{\tilde{N}}

\newcommand{\fC}{\mathfrak{C}}

\newcommand{\oast}{\circledast}
\newcommand{\pr}[2]{P(#1|#2)}

\newcommand{\ep}{\epsilon}

\newtheorem{theorem}{Theorem}[section]
\newtheorem{lemma}[theorem]{Lemma}
\newtheorem{proposition}[theorem]{Proposition}
\newtheorem{definition}[theorem]{Definition}
\newtheorem{corollary}[theorem]{Corollary}

\newtheorem*{conjecture*}{Conjecture}

\theoremstyle{definition}

\newtheorem*{perfectStrat}{Problem $\PerfectStrat_t$}
\newtheorem*{gapPerfectStrat}{Problem $\GapPerfectStrat_t$}
\newtheorem*{membership}{Problem $\Membership_{t,\K}$}
\newtheorem*{constMembership}{Problem $\Membership(n_A,n_B,m_A,m_B)_{t,\K}$}
\newtheorem*{intersection}{Problem $\Intersection(n_A,n_B,m_A,m_B)_{t}$}

\begin{document}

\title{The membership problem for constant-sized quantum correlations is undecidable}
\renewcommand\Affilfont{\itshape\small}
\author[1,4]{Honghao Fu \thanks{honghao.fu@concordia.ca}}
\author[1,2]{Carl A. Miller \thanks{camiller@umd.edu}}
\author[3]{William Slofstra \thanks{weslofst@uwaterloo.ca}}

\affil[1]{Joint Institute for Quantum Information and Computer Science, University of Maryland, College Park, MD, 20742, USA}
\affil[2]{National Institute of Standards and Technology,
100 Bureau Dr., Gaithersburg, MD 20899, USA }
\affil[3]{Institute for Quantum Computing and Department of Pure Mathematics, University of Waterloo, Waterloo, Canada}
\affil[4]{Concordia Institute for Information Systems Engineering, Concordia University, Montreal, Canada}

\date{}

\maketitle

\begin{abstract}
When two spatially separated parties make measurements on an unknown entangled
quantum state, what correlations can they achieve?  How difficult is it to
determine whether a given correlation is a quantum correlation?  These
questions are central to problems in quantum communication and computation.
Previous work has shown that the general membership problem for quantum
correlations is computationally undecidable.  In the current work we show
something stronger: there is a family of constant-sized correlations --- that
is, correlations for which the number of measurements and number of measurement
outcomes are fixed --- such that solving the quantum membership problem for
this family is
computationally impossible.  Thus, the
undecidability that arises in understanding Bell experiments is 
not dependent on varying the number of measurements in the experiment. This
places strong constraints on the types of descriptions that can be given for
quantum correlation sets.  Our proof is based on a combination of techniques
from quantum self-testing and undecidability results 
for linear system nonlocal games.
\end{abstract}
\newpage
\tableofcontents

\section{Introduction}
\label{sec:intro}

Suppose two spatially separated parties, say Alice and Bob, are each able to
perform different measurements on their local system. If Alice can perform
$n_A$ different measurements, each with $m_A$ outcomes, and Bob can perform
$n_B$ different measurements, each with $m_B$ outcomes, then from the point of
view of an outside observer, their behaviour is captured by the collection
\begin{equation*}
    P = \{P(a,b | x,y) : 0 \leq a < m_A,\ 0 \leq b < m_B,\ 0 \leq x < n_A,\ 0 \leq y < n_B\}
\end{equation*}
where $P(a,b|x,y)$ is the probability that Alice measures outcome $a$ and Bob
measures outcome $b$, given that Alice performs measurement $x$ and Bob
performs measurement $y$. The collection $P$ is called a \emph{correlation
(matrix)} or \emph{behaviour} \cite{Tsirelson93}. Colloquially, the size
of a correlation is given by the tuple $(n_A,n_B,m_A,m_B)$. 

It is natural to ask which correlations can occur in nature. Suppose
measurement $x$ on Alice's system always gives outcome $c_x$, and measurement
$y$ on Bob's system always gives outcome $d_y$. Then the corresponding
correlation is $P(a,b|x,y) = \delta_{a, c_x} \delta_{b,d_y}$, where $\delta$ is
the Kronecker delta. Correlations of this form are called \emph{deterministic
correlations}. The convex hull of the set of deterministic correlations is
denoted by $C_c(n_A,n_B,m_A,m_B)$, or $C_c$ when the tuple $(n_A,n_B, m_A, m_B)$
is clear. Correlations in $C_c$ are called \emph{classical correlations}. All
deterministic correlations obviously occur in nature, and if Alice and Bob have
access to shared randomness, they can also achieve all correlations in $C_c$.
It is a fundamental fact of quantum mechanics, first observed theoretically by
John Bell and now verified in many experiments, that Alice and Bob can achieve
correlations outside of $C_c$ by using quantum entanglement \cite{Be64}.

Bell's theorem leads to the question of which correlations can be achieved in
quantum mechanics. To study this question, Tsirelson introduced the set of
quantum correlations \cite{Tsirelson93}. There are actually several ways to define the
set of quantum correlations, depending on whether we assume that all Hilbert
spaces are finite-dimensional, and whether we use the tensor-product axiom or
commuting-operator axiom for joint systems. This leads to several different
choices for the set of quantum correlations: the finite-dimensional quantum
correlations $C_q$, the quantum-spatial correlations $C_{qs}$, the
quantum-approximate correlations $C_{qa}$, and the commuting-operator
correlations $C_{qc}$. We use the same convention as for classical correlations,
in that $C_t$ refers to $C_t(n_A,n_B,m_A,m_B)$ when the tuple
$(n_A,n_B,m_A,m_B)$ is clear. Tsirelson suggested that all four sets should be
equal, but we now know that (for large enough $n_A, n_B, m_A, m_B$) all four sets are different, and hence give a
strictly increasing sequence
\begin{equation*}
    C_c \subsetneq C_q \subsetneq C_{qs} \subsetneq C_{qa} \subsetneq C_{qc}
\end{equation*}
\cite{slofstra2017, coladan2018, re, dykema2019, coladan2020}. The last inequality $C_{qa}
\subsetneq C_{qc}$ is a very exciting consequence of the recent proof \cite{re} that
$\MIP^* = \RE$ by Ji, Natarajan, Vidick, Wright, and Yuen, and following
\cite{fritz2012,JNPPSW11}, this inequality gives a negative resolution to the Connes
embedding problem. 

As the convex hull of a finite set, $C_c$ is a polytope in $\R^{N}$, where
$N = n_A n_B m_A m_B$. The sets $C_t$, $t \in \{q,qs,qa,qc\}$, are also convex
subsets of $\R^N$ (in addition, $C_{qa}$ and $C_{qc}$ are closed), but it
follows from a result of Tsirelson \cite{Tsirelson85} that these sets are not
polytopes. Following up on this point in \cite[Problem 2.10]{Tsirelson93}, Tsirelson asks whether the sets of quantum
correlations might still have nice geometric descriptions, specifically by
analytic or even polynomial inequalities. This question is significant for two
reasons: practical, in that the quantum correlation set captures what is
possible with quantum entanglement, and thus a description of this set tells us
what is theoretically achievable in experiments and quantum technologies; and
conceptual, in that a nice description of the set of quantum correlations could
improve our conceptual understanding of quantum entanglement, similarly to how
the description of $C_c$ as a polytope is
tied to our understanding of classical correlations as mixtures of deterministic correlations.

Due to the significance of this question, describing the set of quantum
correlations has been a central question in the field. On the geometric side,
Tsirelson's original results show that when $m_A = m_B=2$, a certain linear slice of the quantum
correlation set is the elliptope, a convex set described by quadratic
inequalities (\cite{Tsirelson85}, see also
\cite{Landau88,WernerWolf2001,Masanes2003,Pitowsky2008} for subsequent work on the
special case that $n_A = n_B = 2$, and
\cite{TVC2019} for a description as the elliptope).
The convex geometry of
$C_q(2,2,2,2)$ is studied in detail in \cite{GKWVWCLS2018}. The case of
$C_q(2,2,2,2)$ benefits from a dimension reduction argument: by Jordan's
lemma, any correlation in $C_q(2,2,2,2)$ can be expressed
as a convex combination of correlations from two-qubit systems.  In general, we
might ask whether there is a bound on the dimension of Hilbert spaces needed to realize
correlations in $C_q(n_A,n_B,m_A,m_B)$. There are several different proofs that, as
the number of questions and outcomes increases, there are correlations which
require Hilbert spaces of arbitrarily high dimension \cite{Tsirelson93,brunner2008testing,slofstra2011}. If we fix $n_A,
n_B,m_A,m_B$ such that $C_q(n_A,n_B,m_A,m_B)$ is not closed (which again,
happens if $n_A, n_B, m_A,m_B$ are large enough), then $C_{q}(n_A,n_B,m_A,m_B)$
contains correlations which require Hilbert spaces of arbitrarily large dimension.
The first author gives an explicit family of correlations $\set{P_d \mid d \geq 1}$
with $n_A,n_B,m_A,m_B$ fixed requiring maximally entangled states of dimension
$d$ in \cite{fu2019}.  Thus the methods used to study $C_q(2,2,2,2)$ do not work for more
measurements or outcomes.
Using a different dimension reduction argument, Russell describes another linear slice of $C_q$, the
synchronous correlations, in $C_q(3,3,2,2)$, but again this description does
not extend to other numbers of measurements and outcomes \cite{Russ2020}.

In another line, a number of authors have considered whether it's possible to
give a conceptual, rather than geometric, description of the quantum
correlation sets. The first result in this line comes from Tsirelson's original
definition of quantum correlations, where he observes that quantum correlations
belong to the set of nonsignalling correlations, which are those correlations
$P$ for which the sums
\begin{equation*}
    \sum_{b} P(a,b|x,y) \text{ and } \sum_{a} P(a,b|x,y)
\end{equation*}
are independent of $y$ and $x$ respectively. This condition captures the fact that,
when spatially separated, Alice and Bob cannot communicate with each other.
Since the set of nonsignalling correlations is strictly larger than the
commuting-operator correlations $C_{qc}$, the fact that Alice and Bob cannot
communicate does not identify the set of quantum correlations among all
correlations. But it is natural to ask whether there might not be additional
principles which would suffice to identify the set of correlations. Some 
examples of conditions which further restrict the set of nonsignalling correlations
and which are satisfied by quantum correlations can be found in
\cite{BBLMTU,Pawlowski2009,NM2010,FSABCLA,SGAN}, but so far these do not give a
complete description of the set of quantum correlations.

Based on the apparent difficulty of describing the set of quantum correlations,
there has also been a line of work studying the computational complexity
of problems related to these sets.  The main line of inquiry, initiated in \cite{Cleve:2004}, has been to
consider the difficulty of determining the quantum and commuting-operator
values of a nonlocal game.   For example, one can consider the problem
of determining whether a given nonlocal game has a perfect strategy.
\begin{perfectStrat}
	Given a tuple of natural numbers $(n_A, n_B, m_A, m_B)$ and a nonlocal game $G$ with $n_A$ and $n_B$ questions
	and $m_A$ and $m_B$ answers, does $G$ have a perfect strategy in $C_t$?
\end{perfectStrat}
From the point of view of convex geometry, the
quantum (resp. commuting-operator) value of a nonlocal game is the maximum a
certain linear functional on the set $C_{qa}$ (resp.  $C_{qc}$). Asking whether such a nonlocal game
has a perfect strategy corresponds to asking whether this maximum is equal to $1$. Leading
up to \cite{re}, there was a series of deep works showing that even the
approximate version of this optimization problem is indeed very difficult
\cite{ito2012, ruv2013, ji2017, lowdegree, neexp}. These
results have implications in computational complexity theory, as they imply
lower bounds on the complexity class $\MIP^*$ of multiprover proofs with
entangled provers. 
In the exact (rather than approximate) case, previous results by the
last author of the current paper imply that the problems $\PerfectStrat_t$
are undecidable for $t \in \{q, qs, qa, qc\}$  \cite{slofstra2017,slofstra2019, ji2019}. 
To understand the difficulty of approximating the quantum and commuting-operator
values, we can look at a gapped variant of $\PerfectStrat_t$:
\begin{gapPerfectStrat}
	Given a tuple of natural numbers $(n_A, n_B,m_A, m_B)$ and a nonlocal game $G$ with $n_A$ and $n_B$ questions
	and $m_A$ and $m_B$ answers, decide whether $G$ has a perfect strategy in $C_t$, or the quantum value of $G$ is
	$\leq 1/2$, given that one of the two is the case.
\end{gapPerfectStrat}
The result by Ji, Natarajan, Vidick, Wright, and Yuen mentioned above shows that
$\GapPerfectStrat_{t}$ is also undecidable for $t \in \{q,qs,qa\}$ \cite{re}.

Rather than looking at nonlocal games, a more straightforward way to study the difficulty of describing
quantum correlation sets is to look at the membership problem for these sets. 
Specifically, we can look at the decision problems for $t \in \{q,qs,qa,qc\}$ and subfields $\mbK \subseteq \R$.
\begin{membership}
	Given a tuple $(n_A, n_B, m_A, m_B)$, and a
        correlation $P \in \mbK^{n_A n_B m_A m_B}$, is $P \in C_t(n_A,n_B,m_A,m_B)$?
\end{membership}
 The point of
restricting to correlations in $\K^{n_A n_B m_A m_B}$ rather than $\R^{n_A n_B m_A m_B}$
is that it is not possible to describe all real numbers in a finite fashion. We
are primarily interested in subfields of $\R$ such as $\Q$, where it is practical to work
with elements of the field on a computer. For our results we actually need to
take a larger field than $\Q$, so in what follows we'll set $\mbK =
\overline{\Q} \cap \R$ unless otherwise noted, where $\overline{\Q}$ is the
algebraic closure of the rationals.\footnote{Since $\overline{\Q}$ is
computable, it is possible to work with $\overline{\Q}$ and $\overline{\Q} \cap
\R$ on a computer, and indeed support for this is included in Mathematica and
other computer algebra packages.}

The problems $\Membership_{t,\K}$
are a very general way of studying descriptions of the sets $C_t$ for $t \in
\{q,qs,qa,qc\}$, since we don't restrict to any particular form of description,
but instead just look at a basic functionality that we would hope to have from any
nice
description, namely a way of being able to distinguish elements inside the set from those
outside. The decision problems $\Membership_{t,\K}$ are not
equivalent to the problems $\PerfectStrat_t$ or $\GapPerfectStrat_t$,
since nonlocal games do not necessarily have unique perfect strategies in $C_t$.
Nonetheless, the two families of decision problems are closely related. Indeed, the methods
used in \cite{slofstra2019} to show the undecidability of
$\PerfectStrat_{qc}$ are adapted in \cite{coudron2019} to show the
undecidability of $\Membership_{qc,\K}$ \cite{coudron2019}. The
methods of \cite{slofstra2017} can be adapted to show the undecidability of
$\Membership_{t,\K}$ for $t \in \{q,qs,qa\}$ in similar fashion
(although some work is needed for the case $t=q$).  The undecidability of
$\GapPerfectStrat_t$ can be used (in a blackbox fashion,
without referring to the proof methods) to get the stronger result that
$\Membership_{t,\Q}$ is undecidable for $t \in \{q,qs,qa\}$ \cite{re}.

The above undecidability
results put very strong restrictions on what descriptions of the
quantum correlation sets are possible. For instance, they imply that there is no Turing
machine which takes tuples $(n_A,n_B,m_A, m_B)$
as inputs, and outputs a description of $C_t(n_A, n_B, m_A, m_B)$ in terms of a
finite list of polynomial inequalities, since such a Turing machine would allow
us to decide $\Membership_{t,\K}$ (such a description is sometimes called a uniform polynomial description). 
Similarly, these results also
imply that there can be no finite set of principles, independent of $(n_A,n_B,m_A,m_B)$,
such that we can decide algorithmically whether a correlation satisfies every principle, and
such that a correlation satisfies all the principles if and only if it belongs
to $C_t(n_A, n_B, m_A, m_B)$. 

However, we note that the reasoning in the last two paragraphs depends
crucially on the fact that the parameters $(n_A, n_B, m_A, m_B)$ can vary. The
papers \cite{slofstra2019, slofstra2017, re} all involve games with unbounded
alphabet size. Hence these results leave open the possibility that every set
$C_t(n_A,n_B,m_A,m_B)$ has a nice description, but that it is just not possible
to have a Turing machine which outputs these descriptions as a function of
$(n_A,n_B,m_A,m_B)$ (in other words, they might have a non-uniform description). Thus it is natural to ask what happens to the complexity of
$\Membership_{t,\K}$ when $(n_A, n_B, m_A, m_B)$ is held constant.  This question motivates our main result.

\begin{constMembership}
	Given a correlation $P \in \K^{n_An_Bm_Am_B}$, is $P \in C_t(n_A,n_B,m_A,m_B)$?
\end{constMembership}

\begin{theorem}\label{thm:mainintro}
    There is an integer $\alpha$ such that the decision problem
    $\Membership(n_A,n_B,m_A,m_B)_{t,\K}$ is undecidable for $t \in
    \{qa,qc\}$ and $n_A, n_B, m_A, m_B > \alpha$.
\end{theorem}

This result shows that the undecidability of membership in $C_{qa}$ and
$C_{qc}$ is not only a consequence of varying the size of the correlation, but
is in fact embedded into the shape of a single set $C_t(n_A, n_B, m_A, m_B)$
for large enough $(n_A, n_B, m_A, m_B)$. As a practical consequence, the result
shows that there is no description of the set $C_t(n_A,n_B,m_A,m_B)$ (e.g. by
polynomial inequalities) that would allow us to decide membership in that set.

As mentioned above, in this theorem $\mbK$ is the intersection $\overline{\Q}
\cap \R$. However, the proof of this theorem does not rely on writing down very
complicated elements of $\overline{\Q}$.  In fact, $\mbK$ could be replaced with
$\mbK_0 \cap \R$, where $\mbK_0$ is the subfield of $\overline{\Q}$ generated by
roots of unity. In this way, the theorem is similar to the undecidability
results for $(\Membership_{t,\mbK})$ that follow from
\cite{slofstra2017,slofstra2019,coudron2019}.  However, in those results, if the
correlations are defined in terms of observables (as in, for instance, the definition of quantum correlations in \cite{Tsirelson85}) instead of the standard approach using measurements, then
it is possible to take $\mbK = \Q$. In our case, even if we work with
correlations defined in terms of observables, we still need to use roots of
unity. We also note that the correlations constructed in the proof of \cref{thm:mainintro}
are synchronous (see \cref{def:synch_cor}), so \cref{thm:mainintro} holds for the subsets of synchronous correlations.

It is interesting to also consider upper bounds on the complexity of the problem
$\Membership(n_A,n_B,m_A,m_B)_{t,\K}$.  When $t = qc$, this problem
is contained in $\coRE$, and the proof of \cref{thm:mainintro} actually shows that this
problem is $\coRE$-complete (for large enough $n_A, n_B, m_A, m_B$). When $t=q$ or $t=qs$,
this problem is contained in $\RE$, but when $t=qa$, the best known upper bound
on this decision problem is $\Pi^0_2$. In this case, \cref{thm:mainintro}
only shows that $\Membership(n_A,n_B,m_A,m_B)_{qa,\mbK}$ is $\coRE$-hard, so
this lower bound is not necessarily tight. Recently, Mousavi, Nezhadi, and Yuen
have shown that $\PerfectStrat_{qa}$ is $\Pi^0_2$-complete \cite{mousavi2021},
and it seems reasonable to conjecture that $\Membership(n_A,n_B,m_A,m_B)_{qa,\K}$
is also $\Pi^0_2$-complete for large enough $n_A, n_B, m_A, m_B$. We leave this for future research. 

\subsection{Paper overview}
\label{sec:overview}

We summarize the technical content of this paper. The starting point for the
proof of Theorem~\ref{thm:mainintro} is the fact that the halting problem for
Minsky machines is undecidable \cite{minsky}. Minsky machines, which we
review in Section \ref{sec:kms}, are a model of universal computation
similar to Turing machines.  To relate Minsky machines to correlations, we go
through group theory: specifically the Kharlampovich-Myasnikov-Sapir (KMS)
groups \cite{sim_group}, also described in Section \ref{sec:kms}. To construct
correlations from these groups, we use the machinery of \cite{slofstra2017},
described in \cref{S:embeddings}. Section~\ref{sec:prelim} and
Section~\ref{sec:quant_cor} contain some basic background on group theory and
quantum correlations, respectively.

For the proof of \cref{thm:mainintro}, we pick a Minsky machine $\MM$ with an
undecidable halting problem. For each $n \geq 1$, we then write down a finite
set of correlations $F_n$ such that $F_n \cap C_{qa} \neq \emptyset$ if $\MM$
does not accept $n$, and $F_n \cap C_{qc} = \emptyset$ otherwise. The
correlations in $F_n$ have two parts.  The first part is constructed from $\MM$ using the
method of \cite{slofstra2017}, and is independent of $n$. The second
part encodes the input $n$ using the method of \cite{fu2019} to keep
the number of measurements and outcomes fixed. This is described in
\cref{sec:dihedral}. In this way, the number of measurements and measurement
outcomes for correlations in $F_n$ will depend only on $\MM$, not on $n$.
The proof of \cref{thm:mainintro} is completed in \cref{sec:main_cqa}.

\subsection{Acknowledgements}

The authors thank Henry Yuen for helpful conversations about the topics of this paper, and the anonymous referee for detailed comments on the manuscript.
CAM thanks Johannes Bausch for a conversation about \cite{bausch2020undecidability} which helped to inspire this project.
WS is supported by NSERC DG 2018-03968 and an Alfred P. Sloan Research Fellowship.
This paper is partly a contribution of the U.~S.~National Institute
of Standards and Technology, and is not subject to copyright in the United States.

\section{Notation and group theory background}
\label{sec:prelim}
In this section we give a brief description of some of the notation and group
theory concepts we'll use throughout the paper. For basic notation, we denote
the set $\set{0,1,\ldots n-1}$ by $[n]$ and the set $\set{ x \in \R \mid x \geq c}$ by $\R_{\geq c}$. 
We index vectors in $\C^{n}$ starting
from $0$, so $\C^{n} = \C^{[n]}$. The $n$-th root of unity is denoted by
$\omega_n := e^{i2\pi/n}$.  For a Hilbert space $\calH$, we let $\calL(\calH)$
be the set of all bounded linear operators acting on $\calH$, and
$\calU(\calH)$ be the group of unitaries acting on $\calH$. We let
$\|\cdot\|_{op}$ denote the operator norm on $\calL(\calH)$. For finite-dimensional
Hilbert spaces, we also work with the \emph{normalized Hilbert-Schmidt norm}, which
is defined by 
\begin{align*}
        \norm{M} = \sqrt{\frac{\Tr(M\ct M)}{d}}.
\end{align*}
for $M \in \calL(\C^d)$. Note that we don't use any subscript to distinguish
this from other norms, as this will be our default norm. We also let $\tTr(M)$
be the \emph{normalized trace} $\Tr(M)/d$ of $M$. 
When working with a group $G$, we use $e$ for the identity, and let $[g,h]$ be
the commutator $g^{-1} h^{-1} g h$ of $g,h \in G$. We let $g^h$ denote the
conjugate $h^{-1} g h$ of $g$ by $h$.  This notation matches with
\cite{sim_group}.

If $S$ is a set, we let $\mathcal{F}(S)$ be the free group generated by $S$. If
$R$ is a subset of $\mathcal{F}(S)$, then we let $\langle S : R \rangle$ be the
quotient of $\mathcal{F}(S)$ by the normal subgroup generated by $R$. The pair
$S,R$ is called a \emph{presentation} of $\langle S : R \rangle$, and as usual
we use $\langle S :R \rangle$ to refer to the presentation and the group
defined by the presentation interchangeably. We also will write $\ip{S : r_i =
t_i, i \in I}$ to mean the presentation $\ip{S : \{ r_i t_i^{-1}, i \in I\}}$.
If $S$ and $R$ are finite, then $\langle S : R \rangle$ is said to be
\emph{finitely-presented}. If $G = \langle S_G : R_G \rangle$, $S$ is disjoint
from $S_G$, and $R$ is a subset of $\mathcal{F}(S_G \cup S)$, then we sometimes
denote the presentation $\langle S_G \cup S : R_G \cup R \rangle$ by $\langle
G, S : R \rangle$. 
An example of a finitely-presented group that we'll use is the \emph{dihedral
group} 
\begin{align*}
    D_n = \ip{ t_1, t_2 : t_1^2 = t_2^2 = (t_1t_2)^n = e }.
\end{align*}
This group has order $2n$, and the elements are 
$(t_1t_2)^j$ and $t_2(t_1t_2)^j$ for $j \in [n]$.

The \emph{free product} of a group $G$ with a group $H$ is denoted by $G \ast
H$.  Note that if $G = \ip{ S_G : R_G }$ and $H =  \ip{S_H : R_H}$, then $G
\ast H = \ip{ S_G \cup S_H : R_G \cup R_H} = \langle G, S_H : R_H \rangle$,
where the unions of $S_G$ and $S_H$ are disjoint.  A more general notion of the
free product of groups is \emph{the free product of groups with amalgamation}.
Let $G_1$ and $G_2$ be two groups with subgroups $H_1$ and $H_2$ respectively
such that there exists an isomorphism $\phi: H_1 \rightarrow H_2$.  Then the
free product of $G_1$ and $G_2$ with amalgamation is defined by $G_1 \ast_\phi
G_2 := G_1 \ast G_2 / \ip{ h_1 \phi(h_1)^{-1} \;|\; h_1 \in H_1 }$.

Another way to construct new groups from a given group is by
Higman-Neumann-Neumann extension (\emph{HNN-extension}) \cite{HNN}. If $H$ is
a subgroup of $G$ and $\phi:H\to H$ is an injective homomorphism, then the
HNN-extension of $G$ is $\overline{G} = \ip{ G, t : t^{-1} h t = \phi(h), h \in H}$.
By \cite[Theorem $11.70$]{rotman2012}, the natural homomorphism sending $g \in G$
to its image in $\overline{G}$ is injective, meaning that we can regard $G$
as a subgroup of $\overline{G}$. We shall introduce other important properties
of the free product with amalgamation and the HNN-extension later when they are
needed. For more background on these concepts, we refer to \cite{rotman2012}. 

When $\phi:H \to H$ is an isomorphism of order $n$, we similarly define the
$\Z_n$-HNN extension of $G$ by $\hat{G} = \ip{G, t : t^n = e, t^{-1} h t =
\phi(h) \text{ for } h \in H}$. As in the case of the ordinary HNN-extension,
$G$ is embedded in $\hat{G}$:  
\begin{lemma}\label{lem:ZnHNN}
    Let $G$ be a group, $H$ a subgroup of $G$, $\phi : H \to H$ an isomorphism
    of order $n$, and $\hat{G} := \ip{G, t : t^n = e, t^{-1} h t = \phi(h)
    \text{ for } h \in H}$ the $\Z_n$-HNN extension. Then the inclusion
    \begin{equation*}
        G \to \hat{G} : g \mapsto g
    \end{equation*}
    is injective, and $t$ has order $n$.
\end{lemma}
\begin{proof}
    Let $G^{* n}$ denote the free product of $G$ with itself $n$ times, where we
    index the factors by elements of $\Z_n$. Let $i_k : G \incl G^{* n}$ be the
    inclusion of the $k$th factor, $k \in \Z_n$, and let $\psi : G^{* n} \to G^{* n}$
    be the cyclic shift, so $\psi(i_k(g)) = \psi(i_{k+1}(g))$ for $k \in \Z_n$.

    Let $N$ be the normal subgroup of $G^{*n}$ generated by
    $i_k(h^{-1})i_{k+1}(\phi(h))$ for all $h \in H$ and $0 \leq k \leq n-2$.
    Since we only include these relations for $k \leq n-2$, we can check (for
    instance, by looking at presentations) that $G^{*n} / N$ is an iterated
    amalgamated product of $G$ with itself $n$ times. But observe that $N$ also
    contains
    \begin{align*}
        \left(i_0(\phi(h))^{-1})i_1(\phi^2(h))\right) & \left(i_1(\phi^2(h)^{-1}) i_2(\phi^3(h))\right) \cdots
            \left(i_{n-2}(\phi^{n-1}(h)^{-1}) i_{n-1}(\phi^{n}(h))\right)  \\
             & = i_0(\phi(h)^{-1}) i_{n-1}(h),
    \end{align*}
    and hence contains $i_{n-1}(h^{-1}) i_0(\phi(h))$. Consequently
    $\psi(N) = N$, so $\psi$ induces an automorphism $\tilde{\psi}$ of $G^{* n} / N$.  
    We claim that $\hat{G}$ is isomorphic to $K := G^{* n} / N
    \rtimes_{\tilde{\psi}} \Z_n$. Indeed, suppose $t$ is the generator of $\Z_n$,
    so that 
    \begin{equation*}
        t^k \cdot x = \tilde{\psi}^k(x) \cdot t^k \text{ for all } x \in G^{*n} / N.
    \end{equation*}
    Then 
    \begin{equation*}
        t^{-1} i_0(h)N t = \tilde{\psi}^{-1}(i_0(h) N)  = i_{n-1}(h) N = i_0(\phi(h)) N 
    \end{equation*}
    for all $h \in H$, so there is a homomorphism $\alpha : \hat{G} \to K$ sending $g
    \mapsto i_0(g)N$, $g \in G$, and $t \mapsto t$. 

    Going the other way, there is a homomorphism $\beta : G^{*n} \to \hat{G}$ sending
    $i_k(g) \mapsto t^k g t^{-k}$. This homomorphism sends 
    \begin{equation*}
        i_k(h^{-1}) i_{k+1}(\phi(h)) \mapsto t^k h^{-1} t^{-k} \cdot t^{k+1} \phi(h) t^{-(k+1)}
            = t^{k+1} (t^{-1} h^{-1} t) \phi(h) t^{-(k+1)} = e,
    \end{equation*}
    so $\beta$ descends to a homomorphism $\tilde{\beta} : G^{*n}/ N \to \hat{G}$. If $g \in G$
    and $a \in \Z_n$ then 
    \begin{equation*}
        t^{a} \tilde{\beta}(i_k(g) N) t^{-a} = t^{a+k} g t^{-(a+k)} = \tilde{\beta}(i_{a+k}(g) N)
            = \tilde{\beta}(\tilde{\psi}^a(i_k(g) N)).
    \end{equation*}
    We conclude that $\tilde{\beta}(\tilde{\psi}^a(x)) = t^a \tilde{\beta}(x) t^{-a}$ for all
    $x \in G^{* n} / N$, and hence there is a homomorphism $K \to \hat{G}$ sending $x \in G^{*n} / N$
    to $\tilde{\beta}(x)$ and $t \mapsto t$. Since, in particular, this homomorphism sends $i_0(g) N$
    to $g$, it is an inverse to $\alpha$, proving the claim that $\hat{G}$ and $K$ are isomorphic.

    Since $G^{*n} / N$ is an iterated amalgamated free product, the homomorphism $G \to G^{*n} / N$
    sending $g \mapsto i_0(g) N$ is injective. Since $G^{*n} / N$ is a subgroup of $K$, 
    $g \mapsto i_0(g) N$ is still injective when considered as a homomorphism $G \to K$. Composing
    with the isomorphism $K \iso \hat{G}$, we get that the homomorphism $G \to \hat{G} : g \mapsto g$
    is injective. Finally, as the generator of $\Z_n$, $t$ has order $n$.
\end{proof}

A \emph{unitary representation} $\rho$ of a group $G$ on the Hilbert space
$\calH$ is a homomorphism $\rho: G \to \calU(\calH)$.  For any set $X$, we let
$\ell^2 X$ denote the Hilbert space with Hilbert basis $\{\ket{x} : x \in X\}$.
The \emph{left regular representation} $L : G \to \calU(\ell^2 G)$ of a group
$G$ is defined by $L(g)\ket{h} = \ket{gh}$, and \emph{right regular
representation} $R: G \to \calU(\ell^2 G)$ is defined by $R(g)\ket{h} =
\ket{hg^{-1}}$ for all $g, h \in G$. Note that $L(g)$ and $R(g')$ commute for
all $g, g' \in G$.

To construct correlations in $C_{qa}$ (which, recall from the introduction, are
limits of finite-dimensional correlations), we use finite-dimensional
approximate representations of groups. The norm we use for these approximate
representations is the normalized Hilbert-Schmidt norm. 
\begin{definition}[Definition $5$ of \cite{slofstra2017}]
    Let $G=\ip{ S : R }$ be a finitely-presented group,
    and let $\calH$ be a finite-dimensional Hilbert space.
    A finite-dimensional \textbf{$\pmb{\ep}$-approximate representation} of $G$ is a homomorphism
    $\phi: \mathcal{F}(S) \rightarrow \calU(\calH)$ such that
    $\norm{ \phi(r) - \1 } \leq \ep$ for all $r \in R$.
\end{definition}
It is also possible to talk about approximate representations in, for instance, tracial von Neumann algebras,
but we will use approximate representations to mean finite-dimensional approximate representations, since that's what we use in this paper.
An element $g \in G = \ip{S:R}$ represented by a word $w \in \calF(S)$ is
\emph{nontrivial in approximate representations} of $G$ if there exists some
$\delta > 0$ such that, for all $\ep > 0$, there is an $\ep$-approximate
representation $\phi: \calF(S) \to \calU(\calH)$ such that $\norm{ \phi(w) - \1
} \geq \delta$ (this does not depend on the choice of word $w$). Otherwise we say that $g$ is \emph{trivial in approximate
representations}. If $F$ is a finite subset of elements which are non-trivial
in approximate representations, then we can find $\ep$-representations where
all the elements of $F$ are bounded away from the identity, and in fact we can
do this by having the trace of all the elements of $F$ be close to zero:
\begin{proposition}
	\label{prop:tensor_trick}
    Let $G = \ip{S : R}$ and $W$ be a finite subset of $\calF(S)$.
     Then, for every $\ep, \zeta > 0$, there is an $\ep$-approximate
    representation $\phi$ such that for all $w \in W$,
    \begin{enumerate}
    \item if $w$ is trivial in approximate representations of $G$, then
    $1 - \zeta \leq \tTr(\phi(w)) \leq 1$, and 
    \item if $w$ is nontrivial in approximate representations of $G$, then
     $0 \leq \tTr(\phi(w)) \leq \zeta$.    
     \end{enumerate}
     \end{proposition}
The proof of the proposition above is very similar to the proof of \cite[Lemma
$12$]{slofstra2017}, so we omit it here.
The next two well-known lemmas are useful when working with approximate representations.
\begin{lemma}
	\label{lm:norms}
	Let $G =\ip{ S: R}$ be a finitely presented group.
	Suppose $f \in \C[\calF(S)]$.
	Then there is a constant $c$ (depending on $f$) such that for any $\ep$-approximate representation 
	$\phi : \calF(S) \to \calU(\calH)$, $\norm{\phi(f)}_{op} \leq c$.
	Furthermore, if
	 $f = 0$ in $\C[G]$,
	then there is a constant $c'$ (also depending on $f$) such that for any $\ep$-approximate representation 
	$\phi : \calF(S) \to \calU(\calH)$,
 $\norm{ \phi(f) } \leq c' \ep$.
\end{lemma}
\begin{proof}
	Suppose $f = \sum_{i \in [k]} a_i u_i$ where $a_i \in \C$ and $u_i \in \calF(S)$. Then
	\begin{align*}
		\norm{ \phi(f) }_{op} \leq \sum_{i\in [k]} 
		\abs{a_i} \norm{ \phi(u_i)}_{op} 
		 =  \sum_{i \in [k]} \abs{a_i}.
	\end{align*}
	If $f = 0$ in $\C[G]$, 
	we can write 
	\begin{align*}
		f = \sum_{i \in [k]} b_i x_i (e - r_i) y_i,
	\end{align*}
	where $b_i \in \C$, $x_i, y_i \in \calF(S)$ and $r_i \in R$.
	Then
	\begin{align*}
		\norm{ \phi(f) } \leq \sum_{i\in [k]} 
		\abs{b_i} \norm{ \phi(x_i) ( \1 - \phi(r_i))\phi(y_i)}
		=\sum_{i\in [k]} \abs{b_i} \norm{\1 - \phi(r_i)} 
		\leq \sum_{i \in [k]} \abs{b_i} \ep,
	\end{align*}
	by the unitary invariance of the normalized Hilbert-Schmidt norm.
	The lemma follows from taking $c = \sum_{i \in [k]} \abs{a_i}$ and $c' = \sum_{i \in [k]} \abs{b_i}$.
\end{proof}
\begin{lemma}
	\label{lm:round_prj}
	There is a nondecreasing function $\Delta: \R_{\geq 1} \times \N \to \R_{\geq 1}$
	such that 
	if $\set{ P_i \;|\; i \in [n] } \subset \calL(\C^d)$ is a set of matrices such that
	\begin{align*}
		\norm{P_i}_{op} \leq c, \quad \norm{ P_i^2 - P_i } \leq \ep, \quad
		\norm{P_i^\ast - P_i} \leq \ep,
		 \quad \norm{ P_i P_j } \leq \ep,  \text{ and } \quad \norm{\sum_{k \in [n]}P_k - \1}\leq \ep
	\end{align*}
	for all $i, j \in [n]$, $i \neq j$ and some $c \in \R$,
	then
	there is a projective measurement $\set{ \Pi_i \;|\; i \in [n] } \subset \calL(\C^d)$ 
	such that
	$\norm{ \Pi_i -P_i } \leq \Delta(c,n)\ep$ for all $i \in [n]$.
\end{lemma}
\begin{proof}
When $c =1$ and $P_i$ is positive for all $i \in [n]$, 
this is shown in Lemma 3.5 of \cite{kim2018}, with function $\Delta(1,n) = \Delta_{pos} (n)$, where $\Delta_{pos}(n)$ is defined recursively
by $\Delta_{pos}(n+1) = (40n +3 )\Delta_{pos} (n)$ and $\Delta_{pos}(1) = 2\sqrt{2}$. 
To reduce to this case, suppose that $\set{ P_i \mid i \in [n]  }$ satisfy the conditions of the lemma.
If $\set{ Q_i \mid i \in [n]} \subseteq \calL(\C^d)$ are self-adjoint operators such that 
$\norm{Q_i}_{op} \leq c'$ and $\norm{P_i - Q_i} \leq \delta$ for all $i \in [n]$, then
\begin{equation*}
\begin{aligned}
	\norm{Q_i^2 - Q_i} &\leq \norm{Q_i(Q_i - P_i)} + \norm{ (Q_i - P_i)P_i} + \norm{Q_i - P_i} + \norm{P_i^2 - P_i} \\
		&\leq \norm{Q_i}_{op} \norm{Q_i - P_i} +  \norm{P_i}_{op} \norm{Q_i - P_i}+ \norm{Q_i - P_i} + \norm{P_i^2 - P_i} \\
		&\leq (c+c' + 1) \delta + \ep.
\end{aligned}
\end{equation*}
Similarly, $\norm{Q_i Q_j} \leq (c+c')\delta + \ep$ and 
	$\norm{ \sum_{i \in [n]}Q_i - \1 } \leq n\delta + \ep$.

If we take $Q_i = (P_i + P_i^*)/2$, then $Q_i$ is self-adjoint, $\norm{Q_i}_{op} \leq c$ and $\norm{Q_i - P_i}\leq \ep/2$.
Let $\chi_{[1/2, \infty)}$ be the indicator function of the interval $[1/2, \infty)$.
Since $\abs{ \chi_{[1/2, \infty)}(t) -t } \leq 2\abs{t^2 -t}$ for all $t \in \R$,
\begin{align*}
	\norm{ \chi_{[1/2, \infty)} (Q_i) - Q_i } \leq2 \norm{Q_i^2 - Q_i}  \leq ( 2c + 3) \ep.
\end{align*}
Hence the self-adjoint projections $Q_i' = \chi_{[1/2, \infty)} (Q_i)$ satisfy the conditions
\begin{align*}
	&\norm{Q_i' Q_j'} \leq  (c+1)(2c+3) \ep + \norm{Q_iQ_j} \leq (2c^2 + 7c + 4) \ep \text{ and }\\
	&\norm{\sum_{i\in [n]}Q_i' - \1} \leq (2c+3)n\ep+ \norm{\sum_{i\in[n]} Q_i -\1} \leq  (2c+7/2)n \ep + \ep.
\end{align*}
Applying Lemma 3.5 of \cite{kim2018} to $\set{Q_i' \mid i \in [n]}$ yields $\set{ \Pi_i \mid i \in [n]}$
such that 
\begin{align*}
	\norm{\Pi_i - Q_i'} \leq \Delta_{pos}(n) (2c^2 + 7c + 5)n \ep.
\end{align*}
Since 
\begin{align*}
\norm{\Pi_i - P_i} &\leq \norm{\Pi_i - Q_i'} + \norm{Q_i' - Q_i} + \norm{Q_i - P_i} \\
 &\leq 
\Delta_{pos}(n) (2c^2 + 7c + 5)n\ep + (2c+3)\ep + 1/2\ep,
\end{align*}
the lemma is true with $\Delta(c,n) = \Delta_{pos}(n) (2c^2 + 7c + 5)n + 2c + 4$.
\end{proof}
There are variants of \cref{lm:round_prj} that reduce the dependence on $n$ (see, e.g., \cite{dlS21}).
In this paper, $c$ and $n$ are fixed, so $\Delta(c,n)$ is a constant.

By the definition of the normalized Hilbert-Schmidt norm, the set of elements
of $G$ that are trivial in finite-dimensional approximate representations forms
a normal subgroup of $G$, denoted by $N^{fa}$.  For a group $G$, we define 
\begin{align*}
	G^{fa} := G / N^{fa}.
\end{align*}
If $\phi : G \to H$ is a homomorphism between finitely-presented groups and
$x \in G$ is trivial in approximate representations of $G$, then $\phi(x)$
is trivial in approximate representations of $H$, so there is an induced
homomorphism $G^{fa} \to H^{fa}$.
\begin{definition}[Definition $14$ of \cite{slofstra2017}]
    For finitely-presented groups $G$ and $H$,
    a homomorphism $\phi:G \rightarrow H$ is an $\pmb{fa}$\textbf{-embedding} if the induced map:
    $G^{fa} \rightarrow H^{fa}$ is injective.
\end{definition}
In other words $\phi : G \to H$ is an fa-embedding if whenever $x$ is 
non-trivial in approximate representations of $G$, then $\phi(x)$ is 
non-trivial in approximate representations of $H$. 

A finitely-presented group $G$ is said to be \emph{hyperlinear} if every
non-trivial element is non-trivial in approximate representations. Although
we've defined hyperlinearity only for finitely-presented groups, whether
a group is hyperlinear is independent of the presentation. To show that groups
are hyperlinear, we use the stronger properties of solvability, amenability,
and soficity. Recall that a group $G$ is \emph{solvable} if it has subgroups
$G_0 =\set{e}$, $G_1, \ldots, G_{k-1}$ and $G_k = G$ such that $G_{j-1}$ is
normal in $G_j$ and $G_j / G_{j-1}$ is an abelian group, for $1 \leq j \leq k$.
For the purposes of our paper, the definitions of amenable and sofic groups are
irrelevant; we just need the following well-known  properties of these classes
of groups (see \cite[Proposition $2.4.1$]{capraro2015}):
\begin{enumerate}
	\item Solvable groups are amenable, amenable groups are sofic, and sofic groups
        are hyperlinear. 
    \item If $H$ is an amenable subgroup of a sofic group $G$, and $\alpha: H
        \to H$ is an injective homomorphism, then the HNN-extension of $G$ by $\alpha$
        is sofic.
    \item If $H_1$ and $H_2$ are amenable subgroups of sofic groups $G_1$ and
        $G_2$, and $\alpha: H_1 \to H_2$ is an isomorphism, then the free
        product of $G_1$ and $G_2$ with amalgamation, $G_1 \ast_\alpha G_2$, is
        sofic.
    \item If $N$ is a normal subgroup of $G$ such that $N$ is sofic and $G/N$
        is amenable, then $G$ is sofic. 
\end{enumerate}

To these properties we can add:
\begin{lemma}\label{lem:ZnHNN2}
    If $H$ is an amenable subgroup of a sofic group $G$, and $\phi: H \to H$ is
    an isomorphism of order $n$, then the $\Z_n$-HNN-extension of $G$ by $\phi$
    is sofic. 
\end{lemma}
\begin{proof}
    We continue with the notation from the proof of Lemma \ref{lem:ZnHNN}.  As
    $G^{*n} / N$ is an iterated almagamated free product over the amenable
    group $H$, $G^{*n} / N$ is hyperlinear. Since $K / (G^{*n} / N) \iso \Z_n$
    and $\Z_n$ is amenable, $K$ is sofic. 
\end{proof}

\section{Quantum correlations}
\label{sec:quant_cor}

In this section, we now introduce our main object of study. Consider a scenario
with two parties or players, Alice and Bob, and a referee. The referee chooses
questions to send to Alice and Bob from finite sets $\calX$ and $\calY$ respectively,
and they return answers from finite sets $\calA$ and $\calB$. As mentioned in the
introduction, this can also be thought of as a scenario in which Alice and Bob
perform measurements labelled by the elements of $\calX$ and $\calY$, and receive
outcomes from $\calA$ and $\calB$ respectively. Alice and Bob's behaviour in this
scenario can be described by the function
\begin{equation*}
    P : \calA \times \calB \times \calX \times \calY \to \R_{\geq 0} : 
        (k,\ell,i,j) \mapsto \pr{k,\ell}{i,j},
\end{equation*} 
where $\pr{k,\ell}{i,j}$ is the probability of answers $(k,\ell) \in \calA \times \calB$
with questions $(i,j) \in \calX \times \calY$. We call a tuple $(\calX,\calY, \calA, \calB)$ of finite
sets a \emph{nonlocal scenario}, and a function $P : \calA \times \calB \times \calX \times
\calY \to \R_{\geq 0}$ such that
\begin{equation*}
    \sum_{k \in \calA,\ell \in \calB} \pr{ k,\ell}{i,j} = 1
\end{equation*}
 for all $i \in \calX, j \in \calY$ a \emph{bipartite correlation} for the
scenario $(\calX,\calY, \calA, \calB)$.

With quantum correlations, we want to capture what Alice and Bob can do in a
nonlocal scenario when they cannot communicate. Even though Alice and Bob
cannot communicate, the rules of quantum mechanics do allow them to share
entanglement. We can visualize this scenario as in Figure \ref{fig}.
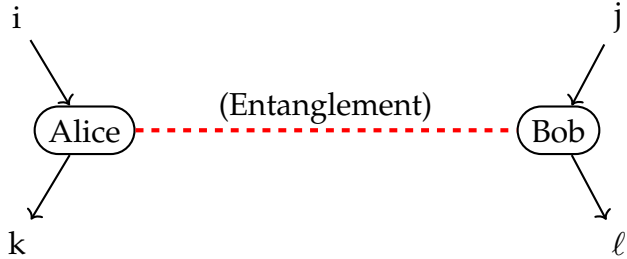
\begin{figure}[H]
\center
        \begin{tikzpicture}[thick]
        \tikzstyle{operator} = [draw,rounded rectangle,fill=white,minimum size=1.5em] 
        \tikzstyle{phase} = [fill,shape=circle,minimum size=5pt,inner sep=0pt]
        \tikzstyle{surround} = [fill=blue!10,thick,draw=black,rounded corners=2mm]
        %
    
        \node at (0,0) (a) {k};
        \node at (0,3) (x) {i};
        
        \node at (8,3) (y) {j};
        \node at (8,0) (b) {$\ell$};
        \node[operator] (qa) at (0.9,1.5){Alice} edge [<-] (x);
        \node[operator] (qb) at (7.2,1.5){Bob} edge [<-] (y);
        \node[text width=2cm] at (3.7,1.8) {(Entanglement)};
        \draw[->] (qa) -- (a);
        \draw[->] (qb) -- (b);
        \draw[red,dashed,ultra thick] (qa) -- (qb);
        \end{tikzpicture}
	\caption{A nonlocal test between Alice and Bob}
	\label{fig}
\end{figure}

Recall that a vector state in a Hilbert space $\calH$ is a unit vector, and a
projective measurement (with $m$ outcomes) is a collection $\{P^{(i)} \mid i
\in [m]\}$ of orthogonal
projections on $\calH$ 
such that
\begin{equation*}
    \sum_{i \in [m]} P^{(i)} = \1.
\end{equation*}
There are two ways in quantum mechanics to handle the restriction that Alice
and Bob cannot communicate. In the first, we require that their joint Hilbert
space be a tensor product of finite-dimensional Hilbert spaces:
\begin{definition}
    \label{def:qs_strat}
    A bipartite correlation $P$ for the scenario $(\calX,\calY, \calA, \calB)$ is a
    \textbf{quantum correlation} if there are
    \begin{enumerate}[(a)]
        \item finite-dimensional Hilbert spaces $\calH_A$ and $\calH_B$,
        \item a vector state $\ket{\psi} \in \calH_A \otimes \calH_B$,
        \item a collection of projective measurements $\{M^{(k)}_i \mid k \in \calA \}$
            on $\calH_A$ for every $i \in \calX$, and
        \item a collection of projective measurements $\{N^{(\ell)}_j \mid \ell \in \calB \}$
            on $\calH_B$ for every $j \in \calY$,
    \end{enumerate}
    such that
    \begin{equation*}
        P(k,\ell|i,j) = \bra{\psi} M^{(k)}_i \otimes N^{(\ell)}_j \ket{\psi}
    \end{equation*}
    for all $i \in X$, $j \in Y$, $k \in A$, $\ell \in B$.

    The set $C_q(\calX,\calY, \calA, \calB)$ is the set of all quantum correlations for scenario
    $(\calX,\calY, \calA, \calB)$, and the set $C_{qa}(\calX, \calY, \calA, \calB)$ is the closure of
    $C_{q}(\calX,\calY, \calA, \calB)$ in $\R^{\calA \times \calB \times \calX \times \calY}$. If $n_A$, $n_B$, $m_A$,
    and $m_B$ are positive integers, and $t \in \{q,qa\}$, then we set
    $C_t(n_A,n_B,m_A,m_B) := C_t([n_A],[n_B],[m_A],[m_B])$.
\end{definition}
If $|\calX| = n_A$, $|\calY| = n_B$, $|\calA| = m_A$, and $|\calB| = m_B$, then $\R^{\calA \times \calB
\times \calX \times \calY} \iso \R^{m_A m_B n_A n_B}$, and this is used to define the
closure of $C_q(\calX,\calY,\calA,\calB)$. In fact, this isomorphism identifies $C_t(\calX,\calY,\calA,\calB)$
with $C_t(n_A,n_B,m_A,m_B)$ for $t \in \{q,qa\}$. So although it's convenient
to be able to use arbitrary labels for questions and answers (and we'll use
non-integer labels in this paper), we could just work with the sets
$C_t(n_A,n_B,m_A,m_B)$ if we wanted. 

Some additional terminology we'll use: a collection of Hilbert spaces, vector
state, and projective measurements as in parts (a)-(d) is called a
\emph{quantum strategy}. Also, although we don't study this set in this paper,
the set $C_{qs}(\calX,\calY,\calA,\calB)$ of \emph{quantum-spatial correlations} is defined
similarly to the set of quantum correlations, but the Hilbert spaces $\calH_A$
and $\calH_B$ in the strategy are allowed to be infinite-dimensional. The
closure of $C_{qs}$ is also equal to $C_{qa}$ \cite{werner08}. 

Moving on, a more general way to handle the restriction that Alice and Bob
cannot communicate is to drop the requirement that Alice and Bob's projective
measurements act on different Hilbert spaces, and instead just require that
their projective measurements commute:
\begin{definition}
    \label{def:qc_strat}
    A bipartite correlation $P$ for a scenario $(\calX,\calY,\calA,\calB)$
    \textbf{commuting-operator correlation} if there is 
    \begin{enumerate}[(a)]
        \item a Hilbert space $\calH$,
        \item a vector state $\ket{\psi} \in \calH$,
        \item a collection of projective measurements $\{M^{(k)}_i \mid k \in \calA \}$
            on $\calH$ for every $i \in \calX$, and
        \item a collection of projective measurements $\{N^{(\ell)}_j \mid \ell
            \in \calB \}$ on $\calH$ for every $j \in \calY$,
    \end{enumerate}
    such that
    \begin{equation*}
        M_{i}^{(k)}N_{j}^{(\ell)} = N_{j}^{(\ell)} M_{i}^{(k)}
    \end{equation*}
    and 
    \begin{equation*}
        P(k,\ell|i,j) = \bra{\psi} M^{(k)}_i \cdot N^{(\ell)}_j \ket{\psi}
    \end{equation*}
    for all $i \in \calX$, $j \in \calY$, $k \in \calA$, $\ell \in \calB$.

    The set $C_{qc}(\calX,\calY,\calA,\calB)$ is the set of all commuting-operator correlations
    for the scenario $(\calX,\calY,\calA,\calB)$, and if $n_A$, $n_B$, $m_A$, and $m_B$ are
    positive integers, then $C_{qc}(n_A,n_B,m_A,m_B) := C_{qc}([n_A],[n_B],[m_A],[m_B])$.
\end{definition}
As with quantum correlations, we refer to a Hilbert space, vector state, and
projective measurements as in (a)-(d) as a \emph{commuting-operator strategy}.
Note that the Hilbert space $\calH$ in a commuting-operator strategy does not
have to be finite-dimensional.

We define one more subtype of correlation that we'll use:
\begin{definition}
\label{def:synch_cor}
A bipartite correlation $P$ for scenario $(\calX,\calY,\calA,\calB)$
is \textbf{synchronous} if $\calX=\calY$, $\calA=\calB$, and
\begin{align*}
    \sum_{ k \in \calA } \pr{k,k}{i,i} = 1
\end{align*}
for all $i \in \calX$.
\end{definition}
Equivalently, a correlation with $\calX=\calY$ and $\calA=\calB$ is synchronous if
$P(k,\ell|i,i) = 0$ for all $k \neq \ell$ and $i$, or in other words
if Alice and Bob always return the same answer when given the same question.
The following fact about synchronous correlations is well-known: 
\begin{proposition}[Theorem 5.5(i) in \cite{paulsen2016}]\label{prop:eqv_test}
    Let $P$ be a synchronous correlation for $(\calX,\calX,\calA,\calA)$, and let $\calH$,
    $\ket{\psi}$, $\{M_i^{(k)} \mid k \in \calA\}$, $\{N_j^{(\ell)} \mid \ell \in \calA\}$
    be a commuting operator strategy for $P$. Then
    \begin{equation*}
        M_i^{(k)} \ket{\psi} = N_i^{(k)} \ket{\psi}
    \end{equation*}
    for all $i \in \calX$, $k \in \calA$.
\end{proposition}

\cref{prop:eqv_test} is an immediate consequence of the following lemma, which 
is contained in the proof of \cite[Theorem 5.5(i)]{paulsen2016}.
\begin{lemma}
\label{lem:eqv_test}
    Let $\ket{\psi} \in \calH$ be a quantum state, and $\set{M_j \mid j\in [n] }$
    and $\set{N_j \mid j\in [n] }$ be two projective measurements on $\calH$ for
    some $n \geq 2$, such that $M_j N_k = N_k M_j$ for all $j,k \in [n]$. 
    If $\bra{\psi}M_j  N_k \ket{\psi} = 0$ for all $j \neq k \in [n]$, then
    \begin{align*}
        &M_j\ket{\psi} = N_j\ket{\psi}
    \end{align*}
    for each $j \in [n]$.
\end{lemma}
The next two conclusions of \cref{lem:eqv_test} will help us work with correlations.
\begin{lemma}
\label{prop:comm_test}
Let $\ket{\psi} \in \calH$ be a quantum state,  and let
$\set{M_0^{(k)} \mid k \in [m_A]}$ and $\set{M_1^{(k)} \mid k \in [m_A]}$ be two projective measurements on $\calH$,
both of which commute with the projective measurement $\set{N^{(l,l')} \mid l,l' \in [m_A]}$ on $\calH$.
If 
\begin{align*}
	\bra{\psi} M_0^{(k)} N^{(l,l')} \ket{\psi} = \bra{\psi} M_1^{(k')} N^{(l,l')} \ket{\psi} = 0
\end{align*}
for any $k \neq l$ and $k' \neq l'$,
then
\begin{align*}
    &M_0^{(k)} M_1^{(k')} \ket{\psi} = M_1^{(k')}M_0^{(k)} \ket{\psi} 
\end{align*}
for any $k, k' \in [m_A]$.
\end{lemma}
\begin{proof}
	The condition implies that the two measurement pairs
	\begin{align*}
		\set{M_0^{(k)} \mid k \in [m_A]}, \set{ \sum_{l' \in [m_A]} N^{(k,l')} \mid k \in [m_A]} 
	\end{align*}
	and
	\begin{align*}
		 \set{M_1^{(k')} \mid k' \in [m_A]}, \set{ \sum_{l \in [m_A]} N^{(l,k')} \mid k' \in [m_A]}
	\end{align*}
	both satisfy the condition of \cref{lem:eqv_test} with respect to $\ket{\psi}$.
    	We conclude that 
    	\begin{align*}
        		&M_0^{(k)} \ket{\psi} = \sum_{l' \in [m_A]} N^{(k,l')} \ket{\psi} \text{ and } \\
		&M_1^{(k')} \ket{\psi} = \sum_{l \in [m_A]} N^{(l, k')} \ket{\psi}
        \end{align*}
        for each $k, k' \in [m_A]$.
    	Then,
    	\begin{align*}
        		M_0^{(k)} M_1^{(k')} \ket{\psi} &= M_0^{(k)} \sum_{l \in [m_A]} N^{(l, k')} \ket{\psi} 
			= \sum_{l \in [m_A]} N^{(l, k')} M_0^{(k)} \ket{\psi} \\
                         &= \sum_{l \in [m_A]} N^{(l, k')} \sum_{l' \in [m_A]} N^{(k, l')} \ket{\psi} 
                            = N^{(k,k')} \ket{\psi} 
                            = \sum_{l' \in [m_A]} N^{(l', k)} \sum_{l \in [m_A]} N^{(l, k')}  \ket{\psi} \\
                            &= M_1^{(k')} \sum_{l' \in [m_A]} N^{(l', k)} \ket{\psi}
                            = M_1^{(k')} M_0^{(k)}\ket{\psi},
    \end{align*}
    for each $k, k' \in [m_A]$.
\end{proof}
\begin{lemma}[Substitution Lemma]
    \label{lm:sub}
    Let $\ket{\psi} \in \calH$ be a quantum state.
    Suppose there exist sequences of unitaries 
    $\set{V}$, $\set{V_i \;|\; i \in [k]}$ and $\set{ M_j  \;|\; j \in [n]}$ on $\calH$ commuting with
    another sequence of unitaries
    $\set{ N_j \;|\; j \in [n]}$ on $\calH$, such that
    \begin{align*}
        &M_j \ket{\psi} = N_j \ket{\psi}
    \end{align*}
    for each $j \in [n]$, and 
    \begin{align*}
        &V \ket{\psi} = \prod_{i \in [k]} V_i \ket{\psi}.
    \end{align*}
    Then
    \begin{align*}
        V\prod_{j \in [n]} M_j \ket{\psi} 
        = \left(\prod_{i \in [k]} V_i\right) 
        \left(\prod_{j \in [n]} M_j \right)\ket{\psi}.
    \end{align*}
\end{lemma}
\begin{proof}
    We prove this lemma by induction on $n$.
    The $n = 0$ case follows the condition that $V \ket{\psi} = \prod_{i \in [k]} V_i \ket{\psi}$.
    Assume the conclusion holds for $n = m$ and consider the case $n = m+1$. Then
    \begin{align*}
        V\prod_{j \in [m+1]} M_j \ket{\psi} 
        &= V \left(\prod_{j \in [m]} M_j \right) M_{m} \ket{\psi} 
        = V \left(\prod_{j \in [m]} M_j \right) N_{m} \ket{\psi} \\
        &= N_{m} V \left(\prod_{j \in [m]} M_j \right) \ket{\psi} 
        = N_{m}\left(\prod_{i \in [k]} V_i\right) 
        \left(\prod_{j \in [m]} M_j \right)\ket{\psi} \\
        &=\left(\prod_{i \in [k]} V_i\right) 
        \left(\prod_{j \in [m]} M_j \right) N_{m} \ket{\psi} 
        = \left(\prod_{i \in [k]} V_i\right) 
        \left(\prod_{j \in [m+1]} M_j \right) \ket{\psi},
    \end{align*}
   which completes the proof.
\end{proof}

When working with correlations, it sometimes simplifies arguments if we restrict to strategies with the following property:
\begin{definition}
\label{def:good_strat}
A commuting operator strategy 
\begin{align*}
\ket{\psi} \in \calH, \set{M_i^{(k)} \mid k \in \calA}, i \in \calX,
 \set{N_j^{(\ell)} \mid \ell \in \calB}, j \in \calY
 \end{align*}  is \textbf{good} if 
 \begin{enumerate}[(a)]
 \item for all $i \in \calX$ and $k \in \calA$,
if $\bra{\psi}M_i^{(k)} \ket{\psi} = 0$, then $M_i^{(k)}=0$, and
\item  for all $j \in \calY$ and $\ell \in \calB$, if $\bra{\psi} N_j^{(\ell)} \ket{\psi} = 0$,
then $ N_j^{(\ell)} = 0$.
\end{enumerate}
\end{definition}
\begin{proposition}
	\label{prop:good_strat}
	If $P \in C_{qc}(\calX, \calY, \calA, \calB)$,
	then $P$ has a good commuting-operator strategy.
\end{proposition}
\begin{proof}
	Suppose $(\ket{\psi}, \set{M_i^{(k)} }, \set{N_j^{(\ell)}})$ is a strategy for $P$.
	If $P(k. \ell \mid i, j) \neq 0$ then 
	\begin{align*}
	\bra{\psi} M_i^{(k)} \ket{\psi} = \sum_{\ell' \in \calB} P(k, \ell' \mid i, j ) \neq 0.
	\end{align*}
	If  $\bra{\psi} M_i^{(k)} \ket{\psi} = 0$, then 
	 $\bra{\psi} M_i^{(k)} N_j^{(\ell)} \ket{\psi} = P(k,\ell \mid i,j) = 0$ for all $j \in \calY$ and $\ell \in \calB$.
	 Since $\sum_{k' \in \calA} M_i^{(k')} = \1$, there must be some $k'$ such that $\bra{\psi} M_i^{(k')} \ket{\psi} \neq 0$.
	 If we replace $M_i^{(k')}$ with $M_i^{(k')} + M_i^{(k)}$ and $M_i^{(k)}$ with $0$, we get another commuting-operator strategy for
	 $P$, and
	 doing this for all $k \in \calA$, $i \in \calX$ with $\bra{\psi} M_i^{(k)} \ket{\psi} = 0$, and similarly for all $\ell \in \calB$, $j \in \calY$
	 with $\bra{\psi} N_j^{(\ell)} \ket{\psi} = 0$,
	  gives a good strategy for $P$.
	 
\end{proof}

\section{Quantum correlations and group theory}\label{S:embeddings}

In this section, we introduce the notion of a perfect correlation associated
with a binary linear system. We then recall the notion of a solution group
associated to a linear system, and use the solution group to show that every
linear system has a perfect correlation. Finally we recall how to embed an
arbitrary finitely presented group in a solution group.

\subsection{Solution groups and correlations}
    
\begin{definition}\label{def:perfectcorrelations}
	Let $Ax = 0$ be an $m \times n$ binary linear system, so $A$ is an $m \times
    n$ matrix over $\Z_2$, and $0 \in \Z_2^n$. Suppose that each row of $A$ has
    $\kappa$ non-zero entries. For each $i \in [m]$, let
    \begin{equation*}
        I_i = \set{ j \in [n] \mid A_{ij} = 1 } 
    \end{equation*}
    and let $\phi_i : I_i \to [\kappa]$ be the unique order-preserving bijection
    (so the smallest element of $I_i$ maps to $0$, the largest maps to $\kappa -1$,
    and so on).  
    Let $\calX_{var} := \set{ x_i \mid i \in [n]}$, and let $\calX := [m] \cup \calX_{var}$.
    Let $S := \{v \in \Z_2^{\kappa} : \sum_{i \in [\kappa]} v_i=0\}$. 
    A correlation $P$ for the scenario $(\calX,\calX,\Z_2^{\kappa},\Z_2^{\kappa})$ is
    a \textbf{perfect correlation for $Ax=0$} if $P(a,b|x,y) = 0$ whenever
    \begin{enumerate}[(1)]
        \item\label{item:X} $x \in [m]$ and $a \not\in S$ or $y \in [m]$ and $b \not\in S$;
        \item\label{item:Xvar} $x \in \calX_{var}$ and $(a_0, \ldots a_{\kappa-2}) \neq (0, 0 \ldots, 0)$, or $y \in \calX_{var}$ and $(b_0, \ldots b_{\kappa-2}) \neq (0, 0 \ldots, 0)$;
        \item\label{item:equation} $x,y \in [m]$ and $a_{\phi_x(k)} \neq b_{\phi_y(k)}$ for some
            $k \in I_x \cap I_y$;
        \item\label{item:varcheckB} $x \in [m]$, $y = x_i \in \calX_{var}$ for some $i \in I_x$, and
            $a_{\phi_x(i)} \neq b_{\kappa-1} $ in $\Z_2$;
        \item\label{item:varcheckA} $x = x_i \in \calX_{var}$ for some $i \in I_y$, $y \in [m]$, and
            $a_{\kappa-1} \neq b_{\phi_y(i)}$ in $\Z_2$; or
        \item\label{item:synch} $x = y \in \calX_{var}$, and $a_{\kappa-1} \neq b_{\kappa-1}$. 
    \end{enumerate}
\end{definition}
In perfect correlations as defined here, the questions that Alice and Bob
receive are either variables in $\calX_{var}$, or indices from $[m]$. When Alice
or Bob gets an index $x \in [m]$, condition \ref{item:X} requires them to
return an element $a \in S$. This element should be thought of as an assignment to the
$x$th equation, where variable $x_k$, $k \in I_x$, receives value $a_{\phi_x(k)}$.
When Alice or Bob gets a variable $x_i$ in $\calX_{var}$, condition \ref{item:Xvar}
forces them to return an element $a \in \Z_2^{\kappa}$ with $a_j = 0$ for $0\leq j \leq \kappa-2$ and $a_{\kappa-1} \in \{0,1\}$. This
should be thought of as an assignment to $x_i$ from $\Z_2$. The remaining conditions
state that if Alice and Bob are asked about the same variable (either as part
of an equation or directly) then their assignments to that variable must agree. 
Note that conditions \ref{item:equation} and \ref{item:synch} imply that every 
perfect correlation for $Ax=0$ is synchronous.

\Cref{def:perfectcorrelations} is stated for linear systems with a 
constant number $\kappa$ of non-zero entries in each row. This allows us
to use answer sets $\Z_2^{\kappa}$. It is possible to define perfect
correlations for systems with a varying number of entries in each row, by
either using answer sets which vary with the question, or by using larger
answer sets $\Z_2^{n}$ as in \cite{kim2018}. However, all the linear systems we
work with have a constant number of non-zero entries in each row, and making
this assumption in \cref{def:perfectcorrelations} simplifies later analysis.

The point of perfect correlations is that any strategy for a perfect correlation yields
a representation of a certain group associated with $Ax =0$.
\begin{definition}[Definition $17$ of \cite{slofstra2017}]
    \label{def:sol_grp}
    Let $Ax = 0$ be an $m \times n$ linear system over $\Z_2$,
    so $A$ is an $m \times n$ matrix with entries in $\Z_2$ and
    $0 \in \Z_2^n$.
    For $j \in [m]$, define $I_j = \set{ k \in [n] \:|\; A_{jk} = 1}$.
    The \textbf{homogeneous solution group} of $Ax = 0$ is
    \begin{align*}
        \Gamma(A) := \ip{x_0,x_1,\ldots x_{n-1}: &
        x_j^2  = e \text{ for all } j \in [n], \\
        &\prod_{k \in I_i} x_k = e \text{ for all } i \in [m],\\
        &[x_j, x_k] = e \text{ if } j,k \in I_i \text{ for some } i}.
    \end{align*}
\end{definition}

\begin{proposition}\label{prop:correlationrep}
    Let $Ax=0$ be a binary linear system with $\kappa$ non-zero entries in each row.
    Suppose $P$ is a perfect correlation for $Ax=0$, and that $\calH$, $\ket{\psi}$,
    $\{M_x^{(a)} \mid a \in \Z_2^{\kappa}\}$, $x \in \calX$, $\{N_y^{(b)} \mid 
    b \in \Z_2^{\kappa}\}$, $y \in \calX$ is a good commuting-operator strategy for
    $P$. Let $\calH_0 := \overline{\calA \cdot \ket{\psi}}$, the closure of
    $\calA \ket{\psi}$ in $\calH$, where $\calA$ is the algebra generated by
    $M_x^{(a)}$ and $N_x^{(a)}$ for $x \in \calX$, $a \in \Z_2^{\kappa}$.
    For $i \in [n]$, let
    \begin{equation*}
        M(x_i) := M_{x_i}^{(0,\ldots,0,0)} - M_{x_i}^{(0,\ldots,0,1)} \text{ and }
        N(x_i) := N_{x_i}^{(0,\ldots,0,0)} - N_{x_i}^{(0,\ldots,0,1)}.
    \end{equation*}
      Then
    \begin{enumerate}[(1)]
        \item $M(x_i)$ and $N(x_i)$ are binary observables such that $M(x_i)\ket{\psi} = N(x_i) \ket{\psi}$ for all $i \in [n]$, 
        \item there are unitary representations $\Phi_M$ and $\Phi_N$ of $\Gamma(A)$ on
            $\calH_0$ sending $x_i \in \Gamma(A)$ to $M(x_i)|_{\calH_0} $ and $N(x_i)|_{\calH_0}$ respectively, and
         \item if $\Phi_M(r) \ket{\psi} = \ket{\psi}$ for some $r \in \Gamma(A)$, then $\Phi_M(r) = \1_{\calH_0}$.
    \end{enumerate}
\end{proposition}
\cref{prop:correlationrep} is similar to \cite[Lemma
8]{cleve2017}. We include a full proof for completeness. 

\begin{proof}[Proof of \cref{prop:correlationrep}]
    Since perfect correlations are synchronous, $M^{(a)}_x \ket{\psi} =
    N^{(a)}_x \ket{\psi}$ for all $x \in \calX$, $a \in \Z_2^{\kappa}$ by
    \cref{prop:eqv_test}. As a result, we have $M(x_i) \ket{\psi} = N(x_i) \ket{\psi}$
    for all $i \in [n]$. For the second part of the proposition, 
    define projections
    \begin{equation*}
        M_{i,k}^{(c)} := \sum_{\substack{a \in \Z_2^{\kappa} \\ a_{\phi_i(k)} = c}}
            M_{i}^{(a)}, i \in [m], k \in I_i, c \in \Z_2.
    \end{equation*}

    Observe that by condition \ref{item:Xvar} of \cref{def:perfectcorrelations}, 
    if $i \in [n]$ and $b \in \Z_2^{\kappa}$ with $(b_0, \ldots, b_{\kappa-2}) \neq (0,\ldots,0)$, then
    \begin{equation*}
        \bra{\psi} N_{x_i}^{(b)} \ket{\psi} = \sum_{a \in \Z_2^{\kappa}}
            \bra{\psi} M_1^{(a)} N_{x_i}^{(b)}  \ket{\psi} = \sum_{a \in \Z_2^{\kappa}}
                P(a,b|1,x_i) = 0,
    \end{equation*}
    and hence $N_{x_i}^{(b)} \ket{\psi} = 0$. 
    Since the strategy is a good strategy, by \cref{def:good_strat}, $N_{x_i}^{(b)} = 0$
    for $(b_0, \ldots, b_{\kappa-2}) \neq (0,\ldots,0)$.
    Thus $\set{N_{x_i}^{(0)}, N_{x_i}^{(1)}}$ is a complete measurement,
    where for convenience we write 
    $N_{x_i}^{(c)}$ for $N_{x_i}^{(0,\ldots,0,c)}$
    for all $i \in [n]$ and $c \in \Z_2$.
    Then $N(x_i) = N_{x_i}^{(0)} - N_{x_i}^{(1)}$ is a binary observable on $\calH$.
    Similarly, condition \ref{item:X}
    of \cref{def:perfectcorrelations} and \cref{def:good_strat} imply that $M_{i}^{(a)}= 0$
    for all $i \in [m]$ and $a \not\in S$.

    By condition \ref{item:varcheckB} of \cref{def:perfectcorrelations}, 
    \begin{equation*}
        \bra{\psi} M_{i,k}^{(0)} \cdot N_{x_k}^{(1)} \ket{\psi}
            = \sum_{\substack{a \in \Z_2^{\kappa} \\ a_{\phi_i(k)} = 0 }} \bra{\psi} M_i^{(a)} N_{x_k}^{(1)} \ket{\psi}
            = \sum_{\substack{a,b \in \Z_2^{\kappa} \\ a_{\phi_i(k)} = 0 }} P(a,(0,\ldots,0,1)|i,x_k) = 0,
    \end{equation*}
    and $\bra{\psi} M_{i,k}^{(1)} \cdot N_{x_k}^{(0)} \ket{\psi} = 0$
    similarly. Applying \cref{lem:eqv_test} to the projective measurements
    $\{M_{i,k}^{(0)}, M_{i,k}^{(1)}\}$ and $\{N_{x_k}^{(0)},
    N_{x_k}^{(1)}\}$, we get that $M_{i,k}^{(c)} \ket{\psi}
   = N_{x_k}^{(c)} \ket{\psi}$
    for all $i \in [m]$, $k \in I_i$, and $c \in \Z_2$.

    Now let $M_{i,k} := M_{i,k}^{(0)} - M_{i,k}^{(1)}$ for some $i \in [m]$.
    Since the projections $M_{i}^{(a)}$, $a \in \Z_2^{\kappa}$ commute,
    $[M_{i,k},M_{i,l}] = 0$ for all $i \in [m]$ and $k,l \in I_i$. If
    $k_0,\ldots,k_{\ell-1}$ is a sequence in $I_i$, and $\sigma$ is a
    permutation of $[\ell]$, then
    \begin{align*}
        N(x_{k_0}) \cdots N(x_{k_{\ell-1}}) \ket{\psi}
        = M_{i,k_{\ell-1}} M_{i,k_{\ell-2}} \cdots M_{i,k_0} \ket{\psi}
        & = M_{i,k_{\sigma(\ell-1)}} \cdots M_{i,k_{\sigma(0)}} \ket{\psi} \\
        & = N(x_{\sigma(k_0)}) \cdots N(x_{\sigma({k_{\ell-1}})}) \ket{\psi},
    \end{align*}
    so the operators $N(x_k), k \in I_i$ commute on $\ket{\psi}$. When we
    take the product across all of $I_i$, 
    \begin{align*}
        \prod_{k \in I_i} N(x_k) \ket{\psi} & = \prod_{k \in I_i} M_{i,k} \ket{\psi}
        = \prod_{k \in I_i} \left(\sum_{a \in \Z_2^{\kappa}} (-1)^{a_{\phi_i(k)}} M_{i}^{(a)}\right) \ket{\psi} 
        = \sum_{a \in \Z_2^{\kappa}} (-1)^{\sum_{k \in [\kappa]} a_k} M_i^{(a)} \ket{\psi} \\
        & = \sum_{a \in S} M_i^{(a)} \ket{\psi} = \sum_{a \in \Z_2^{\kappa}} M_i^{(a)}
        \ket{\psi} = \ket{\psi},
    \end{align*}
    where we use that $M_{i}^{(a)} = 0$ for $a \not\in S$ and $\sum_k a_k = 0$
    for $a \in S$. 
    Finally, $N(x_k)^2 = N_{x_k}^{(0)} + N_{x_k}^{(1)} = \1_{\calH}$.
    
    To finish the proof, let $\calA_0$ (resp. $\calA_1$) be the algebra
    generated by $M_{x}^{(a)}$ (resp. $N_{x}^{(a)}$) for $x \in \calX$ and $a
    \in \Z_{2}^{\kappa}$. A standard result about synchronous correlations,
    following immediately from the fact that $M^{(a)}_x \ket{\psi} = N^{(a)}_x
    \ket{\psi}$, is that $\calA \ket{\psi} = \calA_0 \ket{\psi} = \calA_1 \ket{\psi}$.
    If $R \in \calA_1$ satisfies $R \ket{\psi} = 0$, then $R T \ket{\psi}
    = T R \ket{\psi} = 0$ for all $T \in \calA_0$, and thus $R v = 0$ for all
    $v \in \calH_0 = \overline{\calA \ket{\psi}} = \overline{\calA_0 \ket{\psi}}$. 
    Define $\Phi_N : \calF(x_1,\ldots,x_n) \to U(\calH_0) : x_i \mapsto N(x_i)|_{\calH_0}$, and
    suppose $r \in \calF(x_1,\ldots,x_n)$ is a defining relation for $\Gamma(A)$
    from \cref{def:sol_grp}. We've shown that $\Phi_N(r) \ket{\psi} = \ket{\psi}$,
    and hence $\1 - \Phi_N(r)$ is $0$ on $\calH_0$. It follows that $\Phi_N$ induces a
    representation of $\Gamma(A)$ on $\calH_0$ sending $x_i \mapsto N(x_i)|_{\calH_0}$.
    Switching $M$ and $N$ in the above argument, we see that there is also a
    representation of $\Gamma(A)$ on $\calH_0$ sending $x_i \mapsto M(x_i)|_{\calH_0}$.
    The same argument shows that if $\Phi_M(r) \ket{\psi} = \ket{\psi}$ for $r \in \Gamma(A)$,
    then $\Phi_M(r) = \1_{\calH_0}$.
\end{proof}
We can also construct perfect correlations for $Ax = 0$ from representations of the solution group $\Gamma(A)$.
We do this for particular representations in \cref{sec:main_cqa}.

\subsection{Embedding groups in solution groups}\label{subsec:embeddings}

As mentioned in \cref{sec:overview}, for our proof we embed Kharlampovich-Myasnikov-Sapir (KMS) groups into solution groups using the results of
\cite{slofstra2017}. 
Recall the following technical definition:
\begin{definition}[Definition $32$ of \cite{slofstra2017}]
    \label{def:ex_homo_lcg}
    Let $A$ be an $m \times n$ matrix over $\Z_2$, $C_0 \subseteq [n] \times [n] \times [n]$, $C_1 \subseteq [\ell] \times [n] \times [n]$ for some $\ell \geq 1$, and
    $L$ be an $\ell \times \ell$ lower-triangular matrix with non-negative integer entries. Let
    \begin{align*}
        E\Gamma(A,C_0,C_1,L):=\ip{\Gamma(A),y_0,\ldots,y_{\ell-1}:&
        x_i x_j x_i = x_k \text{ for all } (i,j,k) \in C_0, \\
            &y_i^{-1} x_j y_i = x_k \text{ for all } (i,j,k) \in C_1, \\
            & y_i^{-1} y_j y_i = y_j^{L_{ij}} \text{ for all } i > j \text{ with } L_{ij} > 0 }.
    \end{align*}
    We say a group $G$ is an \textbf{$m \times n \times \ell$ extended homogeneous-linear-plus-conjugacy group}
    if it has a presentation of this form.
\end{definition}

By \cite[Propositions $27$ and $33$]{slofstra2017}, extended
homogenous-linear-plus-conjugacy groups can be $fa$-embedded into solution groups. 
We use the following version of that result:
\begin{proposition}\label{prop:lpc_to_sol}
    Let $G = E \Gamma(A,C_0,C_1,L)$ be an $m \times n \times \ell$ extended homogeneous
    linear-plus-conjugacy group. Then there is an $m' \times n'$ matrix $A'$
    over $\Z_2$ for some $m' \geq m$ and $n' \geq n$, and a homomorphism
    \begin{equation*}
        \widetilde{\phi} : \calF(x_0,\ldots,x_{n-1},y_0,\ldots,y_{\ell-1})
            \to \calF(x_0,\ldots,x_{n'-1}), 
    \end{equation*}
    such that:
    \begin{enumerate}[(a)]
        \item Each row in $A'$ has only three non-zero entries.
        \item $\widetilde{\phi}(x_i) = x_i$ for all $i \in [n]$. 
        \item For all $i \in [\ell]$ there are $j, k \in [n']$ such that
            $\widetilde{\phi}(y_i) = x_j x_k$.
        \item If $r$ is a defining relation of $G$, then $\widetilde{\phi}(r)$
            is in the normal subgroup generated by the defining relations 
            of $\Gamma(A')$, so $\widetilde{\phi}$ induces a homomorphism $\phi
            : G \to \Gamma(A')$.
        \item There are integers $n_1,n_2,n_3$ such that if $\gamma$ is an
            $\epsilon$-representation of $G$, then there is an
            $O(\epsilon)$\nobreakdash-representation $\alpha$ of $\Gamma(A')$ with
            \begin{equation*}
                \alpha(\widetilde{\phi}(g)) = \gamma(g)^{\oplus n_1} \oplus \overline{\gamma(g)}^{\oplus n_2} \oplus \1_{n_3}
            \end{equation*} 
            for all $g \in \calF(x_0,\ldots,x_{n-1},y_0,\ldots,y_{\ell-1})$. As a 
            result, the homomorphism $\phi : G \to \Gamma(A')$ is an 
            $fa$-embedding.
    \end{enumerate}
    Furthermore, both $A'$ and $\widetilde{\phi}$ are constructible, in the sense that
    there's a Turing machine which, given $A$, $C_0$, $C_1$, and $L$, will
    output $A'$ and $\widetilde{\phi}(y_i)$ for all $i \in [\ell]$.
\end{proposition}
Although we don't need this fact, $A'$ and $\widetilde{\phi}$ can be
constructed in polynomial time. 
\begin{proof}
    Parts (b)-(e) follow from the proofs of Propositions 27 and 33 in \cite{slofstra2017}.
    Proposition 27 in \cite{slofstra2017} is actually about non-homogeneous
    solution groups $\Gamma(A,b)$, which have an additional central element $J$
    of order $2$ representing a scalar in $\Z_2$. However, the non-homogeneous
    solution group $\Gamma(A,b)$ associated with an $m \times n$ linear system
    $Ax=b$ can be regarded as a homogeneous solution group by adding variables
    $x_{n},\ldots,x_{2n}$, replacing $J$ with $x_{2n}$ wherever it occurs in 
    the presentation, and adding linear relations $x_{i} x_{n+i} x_{2n} = e$
    for all $0 \leq i < n$ to force $x_{2n}$ to be central. Alternatively, the
    proof of Proposition 27 can be adapted to the homogeneous case by
    replacing $J$ with $e$ wherever it occurs in group presentations, and
    ignoring assignments to $J$ in approximate representations. 

    For part (a), given an $m_0 \times n_0$ matrix $A_0$ over $\Z_2$, we can
    find an $m_1 \times n_1$ matrix $A_1$, where $m_1 \geq m_0$ and $n_1 \geq
    n_0$, such that there is an isomorphism $\Gamma(A_0) \to \Gamma(A_1)$
    sending $x_i \mapsto x_i$ for all $i \in [n_0]$, and $A_1$ has exactly
    three non-zero entries in each row. Indeed, suppose the $i$th row
    has non-zero entries in columns $j_1,\ldots,j_r$, where $r > 3$. 
    Adding variables $z_{1t}$,$z_{2t}$,$z_3$ and equations $z_3 = x_{j_1}
    x_{j_2}$, $z_{1t} x_{j_1}x_{j_t} = e$, and $z_{2t} x_{j_2} x_{j_t} = e$ for all
    $t=3,\ldots,r$, 
    and replacing the $i$th row of $A$ with the equation $z_3 x_{j_3} \cdots x_{j_r} = e$,
    we get an isomorphic solution group where equation $i$ has the number of
    non-zero entries reduced by one, and all the added equations have length
    exactly three (the equations $z_{ij} x_{j_i} x_{j_t}$ are needed to force
    $x_{j_i}$ and $x_{j_t}$ to commute). If row $i$ of $A_0$ has exactly one
    non-zero entry in column $j$, then we can add variables $z_1,z_2,z_3$ and
    replace the $i$th equation with equations $x_j z_1 z_2 = x_j z_1 z_3 = x_j
    z_2 z_3 = z_1 z_2 z_3 = e$ (which together force $x_j = z_1 = z_2 = z_3 =
    e$). If row $i$ has exactly two non-zero entries in columns $j$ and $k$,
    then we can add variables $z_1$, $z_2$, and replace the $i$th equation with
    equations $x_j z_1 z_2 = x_k z_1 z_2 = e$. Iterating these steps, we 
    eventually get $A_1$ as desired.

\end{proof}

\section{Minsky machines and Kharlampovich-Myasnikov-Sapir groups}
\label{sec:kms}

\subsection{Minsky machines}
A \emph{$k$-glass Minsky Machine} \cite{minsky}, denoted by $\MM$, consists
of $k$ glasses where each glass can hold arbitrarily many coins, a set of
states $[N]$, and a finite list of commands. Just like a Turing machine, a
configuration of $\MM$ describes which state the machine is in and how many
coins are in each of the glasses. A computation running on $\MM$ is a
sequence of commands, where each command maps one configuration to another.
Commands can leave a glass unchanged, add a coin to a glass, or remove a coin
from a non-empty glass, as well as change the state. Glasses are numbered starting
from $1$. In formal language, this
means that a \emph{configuration of $\MM$} is an element $(i; n_1, n_2,
\ldots n_{k}) \in [N] \times (\Z_{\geq 0})^{[k]}$. The state $0$ is
regarded as a final halt state, and $1$ is regarded as a start state. The
\emph{accept configuration} is $(0;0,0, \ldots 0)$ and the \emph{starting
configuration} with input $m$ is $(1;m,0, \ldots 0)$.  There are four types of
commands:
\begin{enumerate}
    \item Adding coins: When the state is $i$, add a coin to each of the glasses numbered $j_1,j_2 \ldots j_{\ell}$ where $\ell \leq k$, and go
    to state $j$.  This command is encoded as
    \begin{align*}
        i; \quad \to \quad j; Add(j_1, j_2 \ldots j_{\ell}).
    \end{align*}
    \item Removing coins: When the state is $i$,  if the glasses numbered $j_1,j_2 \ldots j_{\ell}$, $\ell \leq k$, are all nonempty, 
    then remove a coin from each of the glasses numbered $j_1,j_2 \ldots j_{\ell}$, and go to state $j$. This command is encoded as
    \begin{align*}
        i; n_{j_1} > 0, \ldots, n_{j_{\ell}} > 0 \quad \to \quad j; Sub(j_1,j_2,\ldots j_{\ell}).
    \end{align*}
    \item Empty check: When the state is $i$, if the glasses numbered $j_1,j_2 \ldots j_{\ell}$, $\ell \leq k$, are empty, go to state $j$.
    This command is encoded as
    \begin{align*}
         i; n_{j_1} = 0, n_{j_2} = 0, \ldots, n_{j_{\ell}} = 0 \quad \to\quad  j.
    \end{align*}
    \item Stop: When the state is $i$, change state to $0$.  This command is encoded as
    \begin{align*}
        i;\quad  \to\quad  0.
    \end{align*}
\end{enumerate}
In addition, the input state $i$ to each command must be non-zero, so there are
no commands leaving the halt state.  
A command can only be applied to configurations that match the description of the command. 
Some configurations may not have any applicable commands,
while some configurations can have more than one applicable command. If for
every configuration, there is at most one command that can be applied, the
Minsky machine is \emph{deterministic}. Otherwise, the Minsky machine is
\emph{non-deterministic}. A $k$-glass Minsky machine $\MM$ \emph{accepts} an
input $n \in \Z_{\geq 0}$ if there is a sequence of configurations
\begin{equation*}
    (1;n,0,\ldots,0) =: C_0 \to C_1 \to \ldots \to C_N := (0;0,\ldots,0)
\end{equation*}
from the input configuration to the accept configuration such that for
every $0 \leq i < N$, there is a command of $\MM$ that applies to $C_i$
and transforms it to $C_{i+1}$. In general, we let $\equiv_{\MM}$ be the
equivalence relation on configurations generated by the relations $C
\equiv_{\MM} C'$, where $C$ and $C'$ are configurations for which there is a command
of $\MM$ which applies to $C$ and transforms it to $C'$. Note that the
equivalence relation generated by these relations also includes the relations
$C' \equiv_{\MM} C$ whenever there is a command transforming $C$ to $C'$,
as well as the transitive closure of both types of relations.
So $C \equiv_{\MM} C'$ does not necessarily mean that $C$ can be transformed
to $C'$ by applying commands from $\MM$, but rather that $C$ can be
transformed to $C'$ by applying commands or the inverses of commands. 
If $\MM$ is deterministic, then it's easy to see that $C_1 \equiv_{\MM} C_2$
if and only if $C_1$ and $C_2$ can both be transformed to the same configuration 
$C_3$ by the operations of $\MM$. Since there are no commands with input
state $0$, a deterministic Minsky machine $\MM$ accepts an input $n$ if and
only if $(1;n,0,\ldots,0) \equiv_{\MM} (0;0,\ldots,0)$.

Recall that a subset $S$ of natural numbers is \emph{recursively enumerable}, or
$\RE$, if there is an algorithm such that the algorithm halts on input $s$ if
and only if $s \in S$. Minsky machines can recognize any recursively enumerable
set (or in other words, Minsky machines are Turing complete):
\begin{theorem}[Theorem 2.7, part (a) of \cite{sim_group}]
    $X$ is recursively enumerable if and only if there exists a $3$-glass
    deterministic Minsky machine $\MM$ such that $n \in X$ if and only if
    $\MM$ accepts $n$.
\end{theorem}
\begin{proof}
    In Theorem 2.7, part (a) of \cite{sim_group}, the theorem is stated 
    for $2$-glass Minsky machines with a different encoding of the input. However,
    the statement above for $3$-glass Minsky machines is used as part of the proof
    of the $2$-glass case. 
\end{proof}

If we change the input encoding, then $2$-glass Minsky machines are also
sufficiently powerful to recognize all RE sets. However, we use $3$-glass
Minsky machines as the starting point for our construction, since the simpler
input encoding is necessary to connect with other parts of the construction.

\begin{lemma}\label{lem:oneglassmore}
    Let $\MM$ be a $k$-glass deterministic Minsky machine. Then there is a
    $(k+1)$-glass deterministic Minsky machine $\MM'$ such that 
    \begin{enumerate}[(a)]
        \item $n$ is accepted by $\MM'$ if and only if $n$ is accepted by $\MM$, and 
        \item if $n$ is not accepted by $\MM$, then
                $(1;n,0,\ldots,0) \equiv_{\MM'} (1;m_1,\ldots,m_{k+1})$
            if and only if $(n,0\ldots,0) = (m_1,\ldots,m_{k+1})$.
    \end{enumerate}
\end{lemma}
\begin{proof}
    To construct $\MM'$, we add seven new states $0'$,$1',$2',$3'$,$4'$,$5'$ and $6'$
    to $\MM$. We keep the commands of $\MM$ the same, except that we replace
    states $0$ and $1$ with $0'$ and $1'$ wherever they occur in commands. We then
    add commands
    \begin{align*}
         1; n_2=0, n_3=0, \ldots=n_{k+1}=0  \quad &\to \quad 2' \\
         2'; n_1>0   \quad&\to \quad 3'; Sub(1) \\
         3' ; \quad &\to \quad 2' ; Add(2,k+1) \\
         2'; n_1=0   \quad&\to \quad 4' \\
         4';n_2 > 0 \quad & \to \quad 5'; Sub(2) \\
         5'; \quad & \to \quad 4'; Add(1) \\
        4'; n_2=0 \quad & \to \quad 1' \\
        0' ; n_1=0,n_2=0,\ldots,n_{k}=0 \quad & \to \quad 6' \\
        6'; n_{k+1} > 0 \quad & \to \quad 6'; Sub(k+1) \\
        6'; n_{k+1} =0 \quad & \to \quad 0.
    \end{align*}
    To explain these commands, suppose we start in configuration
    $(1;m_1,\ldots,m_{k+1})$ for $\MM'$. We are allowed to move onto state $2'$
    if and only if $m_2=\ldots=m_{k+1}=0$, or in other words we are in an input
    configuration $(1;n,0,\ldots,0)$. In states $2'$ and $3'$, we copy the
    $1$st glass to the $2$nd and $k+1$th glass. Once the configuration is
    $(2';0,n,0,\ldots,0,n)$, we can move on to state $4'$, where we copy
    back the $2$nd glass to the $1$st glass to get configuration
    $(4';n,0,\ldots,0,n)$. At this point, we can move on to state $1'$,
    where we can now start applying commands from $\MM$. Note that commands
    from $\MM$ do not affect the $k+1$th glass. If input $n$ is accepted by 
    $\MM$, then at some point $\MM'$ will end up in configuration
    $(0';0,\ldots,0,n)$. At this point the state changes to $6'$, and we
    empty out the $k+1$th glass and then move to the accept configuration
    $(0;0,\ldots,0)$. 

    Since $0$ is not an allowed input state for commands of $\MM$, the commands
    coming from $\MM$ in $\MM'$ do not share input states with any of the added
    commands above, so $\MM'$ is deterministic. The explanation of the commands
    above shows that if $\MM$ accepts $n$, then so does $\MM'$. If $\MM$
    does not accept $n$, then $\MM'$ will  never reach configuration
    $(0';0,\ldots,0,n)$, and thus cannot proceed to state $0$, so $\MM'$ does
    not accept $n$ either.

    Finally, suppose $c_1 \equiv_{\MM'} c_2$, where $c_1 := (1;n,0,\ldots,0)$
    and $c_2 := (1;m_1,\ldots,m_{k+1})$, and $n$ is not accepted by $\MM'$. Since
    $\MM'$ is deterministic, this means that $\MM'$ will take
    $(1;n,0,\ldots,0)$ and $(1;m_1,\ldots,m_{k+1})$ to a common configuration $c_3
    = (s;p_1,\ldots,p_{k+1})$. If $s=1$ then $c_1 = c_3 = c_2$. Otherwise
    $m_2=\ldots=m_{k+1}=0$. If $s \in \{1',2',3',4',5'\}$ then $\MM'$ will send $c_1$
    and $c_2$ to $(1';n,0,\ldots,0,n) = (1';m_1,0,\ldots,0;m_1)$. Since $\MM$
    does not accept $n$, we cannot have $s \in \{6',0\}$. Since commands from $\MM$ do not
    change the $k+1$th glass of $\MM'$, if $s \not\in\{1',2',3',4',5'\}$ then
    we must have $n=p_{k+1}=m_1$. In each case, we conclude that $(n,0,\ldots,0)
    = (m_1,m_2,\ldots,m_{k+1})$. 
    
\end{proof}

\subsection{Kharlampovich-Myasnikov-Sapir groups}
For a $k$-glass Minsky machine $\MM$, deterministic or non-deterministic, the
\emph{Kharlampovich-Myasnikov-Sapir (KMS) group} $G(\MM)$ is
a finitely presented group defined in \cite{sim_group}, based on an earlier
construction of Kharlampovich \cite{kharlampovich1982}. The definition of this group depends
on a parameter $p$, which we always take to be $2$. We let $S(\MM)$ and
$R(\MM)$ denote the generating set and relations for the presentation given
in \cite[Section 4.1]{sim_group}. 

The point of KMS groups is the following theorem:
\begin{theorem}[Properties 3.1 and 3.2 and Theorem 4.3 of \cite{sim_group}]
	\label{thm:kms_group}
    Let $\MM$ be a Minsky machine. Then $G(\MM)$ is solvable, and 
there is a computable function $w$ from configurations $c$ of $\MM$ to 
words $w(c)$ in the free group $\mathcal{F}(S(\MM))$, such that
\begin{align*}
	w(c) = w(c') \text{ in } G(\MM) \text{ if and only if } c \equiv_{\MM} c'.
\end{align*}
\end{theorem}
In particular, if $w(n) := w((1;n,0,\ldots,0))$ is the word for input configuration
and $w_{accept} := w((0;0,\ldots,0))$ is the word for the accept configuration,
then 
\begin{align*}
w(n) = w_{accept} \text{ in } G(\MM) \text{ if and only if } (1;n,0,\ldots,0) \equiv_{\MM}
(0;0,\ldots,0).
\end{align*}
If $\MM$ is deterministic, then we can replace this last condition with the condition
that $\MM$ accepts $n$.

For our purposes, we need some details of the definition of $G(\MM)$ and the
function $w$ (we also include some additional details for context).  We use the
notation from \cite{sim_group} for ease of reference.  Suppose $\MM$ has
$k$-glasses and state set $[N+1]$.  The generating set $S(\MM)$ is divided into
subsets $L_0$, $L_1$, and $L_2$, where
\begin{align*}
    L_0 & = \{ x(q_i A_{i_1} \cdots A_{i_m}) \mid i \in [N+1], 0 \leq m \leq k, 0 \leq i_1 < i_2 < \ldots < i_m \leq k \}, \\ 
    L_1 & = \{A_i \mid 0 \leq i \leq k\}, \text{ and } \\
    L_2 & = \{a_i, a_i', \tilde{a}_i, \tilde{a}_i' \mid 1 \leq i \leq k\}. 
\end{align*}

Intuitively, the generators $x(u)$ are used to keep track of state, the
generators $A_i$ represent the bottom of the $i$th glass (along with an additional $A_0$ for bookkeeping purposes), and the generators
$a_i$ are used to keep track of coins in the $i$th glass. 

The relations $R(\MM)$ include a number of relations (marked as (G1)-(G7) in
\cite[Section 4.1]{sim_group}) which are common to all $k$-glass Minsky
machines with state set $[N+1]$, and then a number of relations (marked as (G8)
in the same section) for the commands. For the purpose of discussion, we'll
call these \emph{common relations} and \emph{command relations}.
The common relations include relations specifying that the elements of $L_0$
and $L_1$ have order two, and that $[x,y]=e$ for all $x,y \in L_i$, $i=0,1,2$.
The order of the generators in $L_2$ is unspecified.  The common relations also
have the property that if $\MM'$ has state set $[N'+1]$ with $N \leq N'$, then
the common relations of $R(\MM)$ are common relations of $R(\MM')$. 

For the command relations, there is one relation for each command. To specify these
relations, the following notation is used:
 if $f \in \calF(S(\MM))$ and $1 \leq j \leq k$, let
\begin{align*}
    f \oast a_j := f^{-1} f^{a_j} (f^{-1})^{a_j^{-1}} f^{a_j'^{-1}}, 
\end{align*}
and 
\begin{align*}
    f \oast A_j := [f, A_j]
\end{align*}
(recall that $f^{a} := a^{-1} f a$). Also, let 
\begin{equation*}
    t_1 \oast t_2 \ldots \oast t_m := (\ldots (t_1 \oast t_2) \oast \ldots) \oast t_m \text{ and }t_1 \oast t_2^{\oast n} := t_1 \oast \underbrace{ t_2 \oast \ldots \oast t_2}_\text{n times}.
\end{equation*}
Then with this notation, the relation for $i; \to j; Add(j_1,\ldots,j_{\ell})$ is 
\begin{equation*}
    x(q_i A_0) = x(q_j A_0) \oast a_{j_1} \oast \cdots \oast a_{j_\ell},
\end{equation*}
the relation for $i;n_{j_1}>0,\ldots,n_{j_{\ell}} > 0 \to j; Sub(j_1,\ldots,j_{\ell})$ is 
\begin{equation*}
    x(q_i A_0) \oast a_{j_1} \oast \cdots \oast a_{j_\ell} = x(q_j A_0),
\end{equation*}
the relation for $i;n_{j_1}=0,n_{j_2}=0,\ldots,n_{j_\ell}=0 \to j$ is 
\begin{equation*}
    x(q_i A_0) \oast A_{j_1} \oast \cdots \oast A_{j_\ell} = x(q_j A_0) \oast
    A_{j_1} \oast \cdots \oast A_{j_\ell},
\end{equation*}
and the relation for $i; \to 0$ is 
\begin{equation*}
    x(q_i A_0) = x(q_0 A_0).
\end{equation*}

For $n \in \Z_{\geq 0}$, the word corresponding to the input configuration for $n$ is  
\begin{align*}
	w(n) := x(q_1A_0) \oast a_1^{\oast n} \oast A_1 \oast \ldots \oast A_k.
\end{align*}
In particular, $w(0) := x(q_1A_0) \oast A_1 \oast \ldots \oast A_k$.
For the accept configuration, the group element is
\begin{align*}
	w_{accept} := x(q_0A_0) \oast A_1 \oast \ldots \oast A_k.
\end{align*}
By \cite[Relations $(G5a)$ and $(G1)$]{sim_group}, we can actually reduce these two words to 
\begin{align*}
	w(0) = x(q_1A_0A_1\ldots A_k) \text{ and } w_{accept} = x(q_0A_0A_1\ldots A_k). 
\end{align*}
In particular, $w_{accept}^2 = w(0)^2 = [w(0),w_{accept}] = e$.

Suppose we start with a $k$-glass Minsky machine $\MM_0$, and add states
and commands to form a $k$-glass Minsky machine $\MM_1$. Then the generating
set $S(\MM_0) \subseteq S(\MM_1)$. As previously mentioned, all the common
relations for $R(\MM_0)$ belong to $R(\MM_1)$, and since all the commands
of $\MM_0$ are commands of $\MM_1$, the same is true for the command relations.
Hence, we also have $R(\MM_0) \subseteq R(\MM_1)$. This leads immediately to the
following lemma:

\begin{lemma}\label{lem:hom}
    Let $\MM_0$ and $\MM_1$ be $k$-glass Minsky machines with state sets
    $[N_0]$ and $[N_1]$ respectively, where $N_0 \leq N_1$. If every command
    of $\MM_0$ is a command of $\MM_1$, then there is a homomorphism
    $G(\MM_0) \to G(\MM_1)$ sending $x \mapsto x$ for all $x \in S(\MM_0)
    \subseteq S(\MM_1)$. In particular, this homomorphism sends the elements
    $w(0)$ and $w_{accept}$ for $G(\MM_0)$ to the same elements for $G(\MM_1)$.
\end{lemma}

We finish this section by stating how KMS groups connect to correlations:

\begin{lemma}[Lemma 42 of \cite{slofstra2017}]
	\label{lm:kms_group}
    Let $\MM$ be a Minsky machine, and let $S(\MM) = L_0 \cup L_1 \cup L_2$ be the
    partition described above. Then the KMS group $G(\MM)$ has a
    presentation as an $m \times n \times \ell$ extended homogeneous
    linear-plus-conjugacy group $E\Gamma(A,C_0,C_1,L)$, in which:
    \begin{enumerate}[(a)]
        \item the generators in $L_0$ and $L_1$ and the elements
            $x(q_1 A_0) \oast a_1$ and $w(0)w_{accept}$ all
             belong to the generating set $\{x_0,\ldots,x_{n-1}\}$, and
        \item the generators in $L_2$ belong to $\{y_0,\ldots,y_{\ell-1}\}$.
    \end{enumerate}
\end{lemma}
\begin{proof}
    By Lemma 42 of \cite{slofstra2017}, $G(\MM)$ has a presentation as an
    $m' \times n' \times \ell$ extended homogeneous linear-plus-conjugacy group,
    in which $x(q_1 A_0) \oast a_1$ belongs to the generating set
    $\{x_1,\ldots,x_{n'}\}$ (in the notation of \cite{slofstra2017}, $x(q_1 A_0)
    \oast a_1$ belongs to $\calN(L_0, S)$, which is abelian 
    \cite[Theorem 40, part (c)]{slofstra2017}). Looking at the proof of Lemma
    42, this presentation is constructed by starting with the generators for
    $G(\MM)$, and then adding additional generators so that all the relations
    for $G(\MM)$ can be rewritten as either linear or conjugacy relations. The
    generators from $L_0$ and $L_1$ end up in the generating set
    $\{x_1,\ldots,x_{n'}\}$ for the constructed presentation, and the generators
    $L_2$ end up in the generating set $\{y_1,\ldots,y_{\ell}\}$. 

    The element $w(0) w_{accept}$ also belongs to $\calN(L_0, S)$, and hence we can also
    add additional generators and relations to include this in the generating set
    of our presentation,
    following the produre detailed in the proof of Lemma 42. Alternatively, since
    $w(0) = x(q_1A_0A_1\ldots A_k)$ and $w_{accept} = x(q_0A_0A_1 \ldots A_k)$ are already
    generators in the presentation, and $w(0)$ and $w_{accept}$ commute, we can
    also just add one additional generator $x_{n'+1}$, along with the linear
    relation $x_{n'+1} w(0) w_{accept} = e$, to get a presentation with
    $w(0) w_{accept}$ in the generating set as required. 
\end{proof}

\subsection{An extension of the Kharlampovich-Myasnikov-Sapir group}
This section is devoted to proving the following proposition. 
\begin{proposition}\label{prop:exist}
    Let $p(n)$, $n \geq 1$ be an increasing sequence of prime numbers, where the
    function $p : \Z_{\geq 1} \to \Z_{\geq 1}$ is computable, let $X$ be a recursively
    enumerable set of positive integers, and let $r$ be a positive integer which is coprime to
    $p(n)$ for all $n \geq 1$. Then there exists an $m \times n' \times \ell$
    extended homogeneous linear-plus-conjugacy group $H = E
    \Gamma(A,C_0,C_1,L)$ and generators $x \in \{x_0,\ldots,x_{n'-1}\}$ and
    $u,t \in \{y_0,\ldots,y_{\ell-1}\}$ satisfying the following properties:
    \begin{enumerate}[(a)]
        \item $u^{-1} t u = t^r$ in $H$, 
        \item $H/\ip{ t^{p(n)} = e}$ is sofic for all $n \geq 1$, 
        \item $x = e \text{ in } H/\ip{ t^{p(n)} = e}$ if and only if $n \in X$,
            and 
        \item $t$ has order $p(n)$ in $H/\ip{ t^{p(n)} = e}$.
    \end{enumerate}
\end{proposition}
The function of the different generators $x$, $t$, and $u$ will be explained in the
next section.
We prove \cref{prop:exist} in a number of steps, starting with:
\begin{lemma}\label{lm:primeenc}
    Let $p(n)$, $n \geq 1$ be an increasing sequence of prime numbers, where the
    function $p : \Z_{\geq 1} \to \Z_{\geq 1}$ is computable, and let $X$ be a
    recursively enumerable set of positive integers. Then the set
	\begin{align*}
        P_{X} := \set{p(n) \;|\; n \in X}
	\end{align*}
    is recursively enumerable.
\end{lemma}
\begin{proof}
    Let $A_X$ be a Turing machine that accepts $x \in \N$ if and only if $x
    \in X$. Consider the Turing machine which does the following: Given $q \in
    \N$, it computes $p(n)$ for all $n \leq q$. If $q = p(k)$ 
    for some $k \leq q$, then it runs $A_X$ on $k$ and accept if $A_X$ accepts.
    If $q \neq p(k)$ for all $k \leq q$, then it runs forever without accepting. 
    Since $p(n)$ is an increasing sequence, if $q = p(k)$ for some $k
    \in \Z_{\geq 1}$ then $k \leq q$. As a result, this algorithm accepts $q$
    if and only if $q = p(k)$ for $k \in X$.
\end{proof}

Let $\pMM_0$ be a $3$-glass deterministic Minsky machine that accepts $q \in
\N$ if and only if $q \in P_{X}$ (where $P_X$ is the set from
\cref{lm:primeenc}), and let $\pMM$ be a $4$-glass deterministic Minksy machine satisfying
parts (a) and (b) of \cref{lem:oneglassmore} with respect to $\MM = \pMM_0$.
Let $G(\pMM)
=\ip{S(\pMM) : R(\pMM)}$ be the KMS group of $\pMM$. Although $w(p(n)) =
w_{accept}$ in $G(\pMM)$ if and only if $n \in X$, for the group $H$ in
\cref{prop:exist} we need elements $x$ and $t$ which do not depend on $n$, such
that $x=e$ in $H / \ip{ t^{p(n)}=e}$ if and only if $n \in X$. To get these
elements, we'll start with the group
\begin{align}\label{eq:G}
	\overline{G(\pMM)} := \ip{G(\pMM), t\ :\ [t, a_1] = [t, a_1'] = e, t^{-1} x(q_1A_0) t = x(q_1A_0) \oast a_1},
\end{align} 
where $a_1$, $a_1'$, and $x(q_1 A_0)$ are the generators of $G(\pMM)$ described above. We'll then
construct $H$ at the end of the subsection by adding $u$ to $\overline{G(\pMM)}$. 
For our construction, we also want to consider two other groups
\begin{align*}
	G_{p(n)}(\pMM) := &\ip{ G(\pMM) \ :\ x(q_1A_0) \oast a_1^{\oast p(n)}= x(q_1A_0)} \text{ and} \\
    	\overline{G_{p(n)}(\pMM)} := &\ip{ \overline{G(\pMM)} \ :\ x(q_1A_0) \oast a_1^{\oast p(n)}
    = x(q_1A_0), t^{p(n)} = e},
\end{align*}
defined for every $n \geq 1$. The group $\overline{G_{p(n)}(\pMM)}$ is the 
quotient of $\overline{G(\pMM)}$ by $\ip{t^{p(n)}=e}$. To show this, we need
to explain the definition of $\overline{G(\pMM)}$ a little more:
\begin{lemma}\label{lm:hnn-hom}
    Let $\MM$ be a Minsky machine, and let $K$ be the subgroup of $G(\MM)$
    generated by $x(q_1 A_0)$, $a_1$, and $a_1'$. Then there is a homomorphism
    $\alpha : K \to K$ sending $a_1 \mapsto a_1$, $a_1' \mapsto a_1'$, and
    $x(q_1 A_0) \mapsto x(q_1 A_0) \oast a_1$.
\end{lemma}
\begin{proof}
    As we've already noted, the generators of $K$ satisfy the relations
    $x(q_1A_0)^2 = [a_1, a_1'] = e$.  Let 
	\begin{align*}
		\psi: \calF( S(\MM) ) \to \ip{ b_1, b_2 : [b_1, b_2] = e} = \Z \times \Z
	\end{align*}
    be the homomorphism defined by $\psi(a_1) = b_1$, $\psi(a_1')=b_2$, and
    $\psi(s) = e$ for all $s \in S(\MM) \setminus \set{a_1,a_1'}$. Checking
    the relations in \cite{sim_group}, we see that $\psi(r)=e$ for all $r \in R(\MM)$. 
    Hence $\psi$ descends to a homomorphism $G(\MM) \to \Z \times \Z$, and this
    homomorphism restricts to a surjective homomorphism $\psi : K \to \Z \times \Z$. Let
    $\ip{ x(q_1 A_0) }^{K}$ be the normal subgroup generated by $x(q_1 A_0)$ in $K$. 
    Since $\psi(x(q_1 A_0)) = e$, this normal subgroup is contained in the kernel of
    $\psi$, and hence there is a surjective homomorphism $K / \ip{ x(q_1 A_0) }^{K} \to
    \Z \times \Z$. Since $[a_1,a_1'] = e$, we conclude that there is also a homomorphism
    $\Z \times \Z \to K$ sending $b_1 \mapsto a_1$ and $b_2 \mapsto a_1'$. Hence
    $K / \ip{ x(q_1 A_0) }^{K} \iso \Z \times \Z$, and thus $\ip{x(q_1 A_0)}^{K}$ is
    the kernel of $\psi$ in $K$. We conclude that $K \iso \ip{ x(q_1 A_0)}^{K} \rtimes (\Z \times \Z)$,
    and in particular every element of $K$ can be written uniquely as $g a_1^{n} (a_1')^{m}$
    for some $g \in \ip{ x(q_1 A_0) }^{K}$ and $n, m \in \Z$.

    By \cite[Lemma 4.1]{sim_group}, $\ip{ x(q_1 A_0) }^{K}$ is abelian. Hence the functions
    $\ip{ x(q_1 A_0) }^{K} \to \ip{ x(q_1 A_0) }^{K}$ sending $f \mapsto f^{-1}$, $f \mapsto
    f^{a_1}$, $f \mapsto (f^{-1})^{a_1}$, and $f \mapsto f^{(a_1')^{-1}}$ are all homomorphisms,
    and so we conclude that
    \begin{equation*}
        \ip{ x(q_1 A_0) }^{K} \to \ip{ x(q_1 A_0) }^{K} : f \mapsto f \oast a_1
    \end{equation*}
    is a homomorphism. Using the relation $[a_1,a_1'] = e$ again, we also see that
    \begin{equation*}
        f^{a_1^n(a_1')^m} \oast a_1 = (f \oast a_1)^{a_1^n(a_1')^m} 
    \end{equation*}
    for all $n, m \in \Z$. Using the fact that $K \iso \ip{x(q_1 A_0)}^{K} \rtimes (\Z \times \Z)$,
    we see that there is a homomorphism $\alpha : K \to K$ sending $a_1 \mapsto a_1$,
    $a_1' \mapsto a_1'$, and $f \in \ip{x(q_1 A_0)}^{K}$ to $f \oast a_1$ as desired.
\end{proof}
Let $K$ and $\alpha$ be the group and homomorphism from \cref{lm:hnn-hom}, with
$\MM = \pMM$.  The presentation of $\overline{G(\pMM)}$ states that $t^{-1} f t
= \alpha(f)$ for $f \in \{a_1,a_1',x(q_1 A_0)\}$, so this identity holds for
all $f \in K$. In other words, $\overline{G(\pMM)}$ is analogous to the HNN-extension
of $G(\pMM)$ by $\alpha$, although strictly speaking we do not know if
$\overline{G(\pMM)}$ is an HNN extension, since we do not know if $\alpha$ is
injective. This is enough to show:
\begin{corollary}
    \label{cor:iso}
    $\overline{G(\pMM)} / \ip{t^{p(n)}=e} \cong \overline{G_{p(n)}(\pMM)}$. 
\end{corollary}
\begin{proof}
    Let $K$ and $\alpha$ be the group and homomorphism from \cref{lm:hnn-hom},
    with $\MM = \pMM$. Since $t^{-1} f t = \alpha(f)$ for all $f \in K$,
    $t^{-k} f t^k = \alpha^k(f)$ for all $k \geq 0$. If
    $f \in \ip{x(q_1 A_0)}^{K}$, then the proof of Lemma \ref{lm:hnn-hom} shows
    that $\alpha(f) = f \oast a_1 \in  \ip{x(q_1 A_0)}^{K}$, so 
    $\alpha^k(f) = f \oast a_1^{\oast k}$ for all $k \geq 0$. Thus, in $\overline{G(\pMM)} /
    \ip{t^{p(n)}=e}$ we have
    \begin{equation*}
        x(q_1 A_0) \oast a_1^{\oast p(n)} = t^{-p(n)} x(q_1 A_0) t^{p(n)}
        = x(q_1 A_0).
    \end{equation*}
    So $\overline{G(\pMM)} / \ip{t^{p(n)}=e} = \ip{\overline{G(\pMM)} : t^{p(n)}=e,
             x(q_1 A_0) \oast a_1^{\oast p(n)} = x(q_1 A_0)} = \overline{G_{p(n)}(\pMM)}$.
\end{proof}

Although we can't show that $\overline{G(\pMM)}$ is an HNN-extension, we can
show:
\begin{lemma}
	\label{lm:hnn-ext}
	$\overline{G_{p(n)}(\pMM)}$ is a $\Z_{p(n)}$-HNN-extension of $G_{p(n)}(\pMM)$
    over the subgroup $K_{p(n)}$ generated by $x(q_1 A_0)$, $a_1$, and $a_1'$ in 
    $G_{p(n)}(\pMM)$. In addition, $K_{p(n)}$ is amenable.
\end{lemma}
\begin{proof}
    We continue with the notation from the proof of Lemma \ref{lm:hnn-hom}, with
    $\MM = \pMM$. Note that the map $\psi : \calF(S(\pMM)) \to \ip{b_1,b_2 : [b_1,b_2]=e}$
    sending $a_1 \mapsto b_1$, $a_1' \mapsto b_2$, and $s \mapsto e$ for all
    $s \in S(\pMM) \setminus \{a_1,a_1'\}$ sends $r \mapsto e$ for all $r$ in
    the normal subgroup of $\calF(S(\pMM))$ generated by $S(\pMM) \setminus \{a_1,a_1'\}$.
    If $f$ belongs to this subgroup, then $f \oast a_1$ also belongs to this 
    subgroup, and hence 
    \begin{equation*}
        \psi(x(q_1 A_0) \oast a_1^{\oast p(n)}) = e = \psi(x(q_1 A_0)).
    \end{equation*}
    We conclude that $\psi$ induces a homomorphism $G_{p(n)}(\pMM) \to \Z \times \Z$.
    We can then follow the proof of \cref{lm:hnn-hom} 
    exactly to show that $K_{p(n)} = \ip{ x(q_1 A_0)}^{K_{p(n)}} \rtimes (\Z \times \Z)$,
    and that there is a homomorphism $\widetilde{\alpha} : K_{p(n)} \to K_{p(n)}$
    such that $\widetilde{\alpha}(a_1) = a_1$, $\widetilde{\alpha}(a_1') = a_1'$,
    and $\widetilde{\alpha}(f) = f \oast a_1$ for all $f \in  \ip{ x(q_1 A_0)}^{K_{p(n)}}$.

    If $f \in  \ip{ x(q_1 A_0)}^{K_{p(n)}}$, then $\widetilde{\alpha}(f) =  f
    \oast a_1 \in  \ip{ x(q_1 A_0)}^{K_{p(n)}}$, so $\widetilde{\alpha}^k(f) = 
    f \oast a_1^{\oast k}$ for all $k \geq 0$. Hence
    \begin{equation*}
        \widetilde{\alpha}^{p(n)}(x(q_1 A_0)) = x(q_1 A_0) \oast a_1^{\oast p(n)}
        = x(q_1 A_0).
    \end{equation*}
    Since $\widetilde{\alpha}^{p(n)}(a_1) = a_1$ and $\widetilde{\alpha}^{p(n)}
    (a_1') = a_1'$ as well, we conclude that $\widetilde{\alpha}^{p(n)} = \1$
    on $K_{p(n)}$. In other words, $\widetilde{\alpha}$ is an automorphism of order
    $p(n)$. Looking at the presentations, we see 
    that $\overline{G_{p(n)}(\pMM)}$ is the $\Z_{p(n)}$-HNN extension of $G_{p(n)}(\pMM)$
    by $\widetilde{\alpha}$.

    To see that $K_{p(n)}$ is amenable, observe that $\ip{ x(q_1 A_0)}^{K_{p(n)}}$ is
    a subgroup of the group $T$ from \cite[Lemma 4.5]{sim_group}, and hence is 
    abelian. As a semidirect product of two abelian groups, $K_{p(n)}$ is solvable,
    and hence amenable.
\end{proof}

Towards proving part (b) of \cref{prop:exist}, we get:
\begin{corollary}
	\label{cor:sofic}
	 The group $\overline{G_{p(n)}(\pMM)}$ is sofic.
\end{corollary}
\begin{proof}
    Since $G(\pMM)$ is solvable and $G_{p(n)}(\pMM)$ is a quotient of
    $G(\pMM)$, $G_{p(n)}(\pMM)$ is solvable.  Since $\overline{G_{p(n)}(\pMM)}$
    is a $\Z_{p(n)}$-$HNN$-extension of $G_{p(n)}(\pMM)$ over the amenable subgroup
    $K_{p(n)}$,
    $\overline{G_{p(n)}(\pMM)}$ is sofic by \cref{lem:ZnHNN2}. 
\end{proof}

For the proof of \cref{prop:exist}, we'll take $x = w(0) w_{accept}$. We already
have the ingredients to show that $x=e$ in $\overline{G_{p(n)}(\pMM)}$ if $n \in X$.
However, we also need to ensure that $x \neq e$ if $n \not\in X$. For this, we
introduce a non-deterministic modification of $\pMM$, denoted by
$\pMM^{(p(n))}$. To construct $\pMM^{(p(n))}$ from $\pMM$, we add $p(n)-1$ additional
states, which we denote by $2', \ldots, p(n)'$. We include all the commands of $\pMM$
in $\pMM^{p(n)}$, and we add $p(n)$ new commands:
\begin{align*}
    	&1; \quad \rightarrow \quad 2';Add(1) \\
    	&i'; \quad \rightarrow \quad (i+1)';Add(1) \text{ for } 2 \leq i < p(n) \\
    	&p(n)'; \quad \rightarrow \quad 1;Add(1).
\end{align*}
In other words, in configuration $(1;k,0,0,0)$ in $\pMM^{p(n)}$, we have two
choices. We can either apply commands from $\pMM$, or add a coin to the first
glass and proceed to state $2'$. After this choice, we are forced to go through
the states $i'$ for $i=2,\ldots,p(n)$, adding a coin to the first glass each
time, until we return to state $1$ in configuration $(1;k+p(n),0,0,0)$. 
Since $\pMM$ was constructed using \cref{lem:oneglassmore}, we can show:
\begin{lemma}\label{lem:pnequiv}
    Let $n \geq 1$. Then
    $(1;0,0,0,0) \equiv_{\pMM^{p(n)}} (0;0,0,0,0)$ if and only if $n \in X$.
\end{lemma}
\begin{proof}
    If $n \in X$, then $(1;p(n),0,0,0)$ is accepted by $\pMM$. So in $\pMM^{p(n)}$
    there is a computation path going from $(1;0,0,0,0)$ to $(1;p(n),0,0,0)$,
    and from there to the accept configuration. Hence $(1;0,0,0,0) \equiv_{\pMM^{p(n)}}
    (0;0,0,0,0)$.
    
    For the other direction, suppose $(1;0,0,0,0) \equiv_{\pMM^{p(n)}}
    (0;0,0,0,0)$.  Let $N$ be the smallest integer such that there is a
    sequence of configurations
    \begin{equation*}
        (1;0,0,0,0) =: C_0 \to C_1 \to \ldots \to C_N := (0;0,0,0,0),
    \end{equation*}
    where for all $i=1,\ldots,N$, either $C_{i-1}$ can be transformed to $C_i$
    by a command of $\pMM^{p(n)}$, or $C_i$ can be transformed to $C_{i-1}$.
    Let $k$ be the largest integer such that $C_0 \to C_1 \to \ldots \to
    C_k$ does not use any commands from $\pMM$. The states of the
    configurations $C_i$ for $0 \leq i \leq k$ must belong to
    $\{1,2',3',\ldots,p(n)'\}$, so in particular $k < N$. In addition, since
    $C_{k} \to C_{k+1}$ must use a command from $\pMM$, $C_{k}$ must be in 
    in state $1$. The sequence $C_0 \to \ldots \to C_{k}$ shows that
    $(1;0,0,0,0) \equiv_{\MM'} C_{k}$ in the Minsky machine $\MM'$ with
    states $\{1,2',\ldots,p(n)'\}$ and commands
\begin{align*}
    	&1; \quad \rightarrow \quad 2';Add(1) \\
    	&i'; \quad \rightarrow \quad (i+1)';Add(1) \text{ for } 2 \leq i < p(n) \\
    	&p(n)'; \quad \rightarrow \quad 1;Add(1).
\end{align*}
    $\MM'$ is deterministic and sends $(1;a,0,0,0)$ to $(1;a+mp(n),0,0,0)$ for $m \geq 0$.
    Thus $(1;a,0,0,0) \equiv_{\MM'} (1;b,0,0,0)$ if and only if $a = b \text{ mod } p(n)$.    We conclude that $C_{k} = (1;mp(n),0,0,0)$ for some $m \geq 0$.

    Let $\ell$ be the largest integer such that $C_k \to \ldots \to C_\ell$ involves
    only commands from $\pMM$. If $\ell = N$, then $(1;mp(n),0,0,0)
    \equiv_{\pMM} (0;0,0,0,0)$. Since $\pMM$ is deterministic, that would mean
    that $mp(n)$ is accepted by $\pMM$, so $m = 1$ and $n \in X$ as desired. 

    Suppose $\ell < N$, so $C_{\ell} \to C_{\ell+1}$ involves a command not in
    $\pMM$. The configurations $C_i$ for $k \leq i \leq \ell$ cannot be in
    states $\{2',\ldots,p(n)'\}$, so $C_{\ell}$ must be in state $1$. The
    sequence $C_{k} \to \ldots \to C_{\ell}$ shows that $C_k \equiv_{\pMM}
    C_{\ell}$. If $n \not\in X$, then $mp(n)$ is not accepted by $\pMM$ for
    all $m \geq 0$, so $C_\ell = C_{k}$ by \cref{lem:oneglassmore}.
    But then $C_0 \to \ldots \to C_k \to C_{\ell+1} \to \ldots C_N$ is a
    sequence showing that $(1;0,0,0,0) \equiv_{\pMM^{p(n)}} (0;0,0,0,0)$.
    Since $k < \ell$ by the definition of $k$, this contradicts the minimality
    of $N$. So again, we conclude that $n \in X$ as desired. 
\end{proof}

We are now ready to show:
\begin{lemma}\label{lem:acceptidentity}
	\label{prop:neq}
	In $G_{p(n)}(\pMM)$, $w(0) = w_{accept}$ if and only if $n \in X$.
\end{lemma}
\begin{proof}	
    Because $x(q_1 A_0) \oast a_1^{\oast p(n)} = x(q_1 A_0)$ in $G_{p(n)}(\pMM)$,
	\begin{equation*}
        		w(0) = x(q_1A_0) \oast A_1 \oast A_2 \oast A_3 \oast A_4 
                           			= x(q_1A_0) \oast a_1^{\oast p(n)} \oast A_1 \oast A_2 \oast A_3 \oast A_4 = w(p(n)).
    	\end{equation*}
    If $n \in X$, then $p(n)$ is accepted by $\pMM$, so $w(p(n)) = w_{accept}$ in $G(\pMM)$, and
    hence the same is true in $G_{p(n)}(\pMM)$. Thus we conclude that $w(0) = w_{accept}$
    in $G_{p(n)}(\pMM)$ when $n \in X$. 

    For the other direction,
    by \cref{lem:hom} there is a homomorphism
    $G(\pMM) \to G(\pMM^{p(n)})$ which is the identity on generators.
    From the states and commands added to $\pMM^{p(n)}$, we have generators 
    $x(q_{i'} A_0)$ for $2 \leq i \leq p(n)$, and relations
	\begin{align*}
    		&x(q_1A_0) = x(q_{2'}A_0) \oast a_1, \\
            &x(q_{i'} A_0) = x(q_{(i+1)'} A_0) \oast a_1 \text{ for all } 2 \leq i < p(n), \\
    		&x(q_{p(n)'}A_0) = x(q_1A_0) \oast a_1. 
	\end{align*}
	Combining these relations, we see that 
	\begin{align*}
		\label{eq:addpn}
    		x(q_1A_0) \oast a_1^{\oast p(n)} = x(q_1A_0)
	\end{align*}
    in $G(\pMM^{(p(n))})$, so the homomorphism $G(\pMM) \to G(\pMM^{p(n)})$
    descends to a homomorphism $G_{p(n)}(\pMM) \to G(\pMM^{p(n)})$. 
    If $n \not\in X$, then by \cref{lem:pnequiv}, $(1;0,0,0,0)
    \not\equiv_{\pMM^{p(n)}} (0;0,0,0,0)$, so $w(0) \neq w_{accept}$ in 
    $G(\pMM^{p(n)})$. Hence $w(0) \neq w_{accept}$ in $G_{p(n)}(\pMM)$.
\end{proof}

The relations between $\overline{G(\pMM)}/\ip{t^{p(n)} = e}$, $G_{p(n)}(\pMM)$, $\overline{G_{p(n)}(\pMM)}$
and $G(\pMM^{(p(n))})$ are summarized in the following equation:
 \begin{equation*}
    G(\pMM^{p(n)}) \longleftarrow G_{p(n)}(\pMM) \lhook\joinrel\longrightarrow \overline{G_{p(n)}(\pMM)} \iso \overline{G(\pMM)} / \ip{ t^{p(n)} = e }.
\end{equation*}

We are finally ready to prove:
\begin{proof}[Proof of \cref{prop:exist}]
    Let
    \begin{equation*}
        H = \ip{ \overline{G(\pMM)}, u \ :\ u^{-1} t u = t^r }.
    \end{equation*}
    We take $x = w(0) w_{accept} \in H$, and $u$ and $t$ to be the generators already defined in $H$.
    Part (a) follows immediately from the definition. For parts (b)-(d), observe that 
    \begin{equation*}
        H / \ip{ t^{p(n)} = e} = \ip{ \overline{G_{p(n)}(\pMM)}, u\ :\ u^{-1} t u = t^r }
    \end{equation*}
    by \cref{cor:iso}. By \cref{lm:hnn-ext}, $\overline{G_{p(n)}(\pMM)}$ is a
    $\Z_{p(n)}$-HNN extension of $G_{p(n)}(\pMM)$. Hence $t$ has order $p(n)$
    in  $\overline{G_{p(n)}(\pMM)}$ by \cref{lem:ZnHNN}. Since $r$ is coprime
    to $p(n)$, $t^i \mapsto t^{ri}$ is an automorphism of the subgroup generated
    by $t$ in $\overline{G_{p(n)}(\pMM)}$, so $H / \ip{t^{p(n)}=e}$ is an HNN-extension
    of $\overline{G_{p(n)}(\pMM)}$ over the subgroup generated by $t$. Since the subgroup
    generated by $t$ is finite (and hence amenable), $H / \ip{t^{p(n)}=e}$ is sofic by
    \cite[Proposition 2.4.1]{capraro2015}, proving part (b).
    We also get that $\overline{G_{p(n)}(\pMM)}$ is a subgroup of $H /
    \ip{t^{p(n)}=e}$, so $t$ has order $p(n)$ in $H / \ip{t^{p(n)}=e}$ as well, proving part (d). 
    And since $G_{p(n)}(\pMM)$ is a subgroup of  $\overline{G_{p(n)}(\pMM)}$ by
    \cref{lem:ZnHNN}, $x = e$ in $H / \ip{t^{p(n)}=e}$ if and only
    if $n \in X$ by \cref{lem:acceptidentity}, proving part (c).

    It remains to show that $H$ is an exended homogeneous linear-plus-conjugacy group,
    with $u$, $t$, and $x$ in the generating set. By \cref{lm:kms_group}, $G(\pMM)$
    has a presentation as an $m \times n' \times \ell'$ extended homogeneous
    linear-plus-conjugacy group $E \Gamma(A,C_0,C_1,L)$, in which $x(q_1 A_0)$, $x = w(0)
    w_{accept}$ and $x(q_1 A_0) \oast a_1$ belong to the generating set
    $\{x_0,\ldots,x_{n'-1}\}$, and $a_1$ and $a_1'$ belong to the generating
    set $\{y_0,\ldots,y_{\ell'-1}\}$. To present $H$ as an
    extended homogeneous linear-plus-conjugacy group, we can add two additional
    generators $y_{\ell'}$ and $y_{\ell'+1}$ for $t$ and $u$ respectively, and add the conjugacy relations
    \begin{equation*}
        y_{\ell'}^{-1} x(q_1 A_0) y_{\ell'} = x(q_1 A_0) \oast a_1,
        y_{\ell'}^{-1} a_1 y_{\ell'} = a_1, 
        y_{\ell'}^{-1} a_1' y_{\ell'} = a_1',
        \text{ and } y_{\ell'+1}^{-1} y_{\ell'} y_{\ell'+1} = y_{\ell'}^r.
    \end{equation*}
    This gives a presentation of $H$ as an $m \times n' \times \ell$ extended
    homogeneous linear-plus-conjugacy group, where $\ell = \ell'+2$.
\end{proof}

\subsection{Embedding KMS groups in solution groups}

\begin{proposition}\label{prop:existsg}
    Let $p(n)$, $n \geq 1$ be an increasing sequence of prime numbers, where the
    function $p : \Z_{\geq 1} \to \Z_{\geq 1}$ is computable, let $X$ be a recursively
    enumerable set of positive integers, and let $r$ be a positive integer which is coprime to
    $p(n)$ for all $n \geq 1$. Then there is an $m' \times n'$ solution group
    $\Gamma(A)$ for some $m'$, $n'$ with generators $x, t_1,t_2,u_1,u_2 \in \{x_0,\ldots,x_{n'-1}\}$ such that
    \begin{enumerate}[(a)]
        \item $u_2 u_1 t_1 t_2 u_1 u_2 = (t_1 t_2)^r$ in $\Gamma(A)$,
        \item $x=e$ in $\Gamma(A) / \ip{(t_1 t_2)^{p(n)} = e}$ if and only if $n \in X$, 
        \item if $n \not\in X$, then
            $x$ is non-trivial in approximate representations of $\Gamma(A) / \ip{(t_1 t_2)^{p(n)} = e}$, 
        \item $t_1 t_2$ has order $p(n)$ in $\Gamma(A) / \ip{(t_1 t_2)^{p(n)}=e}$, 
        \item every $w \in \ip{ t_1, t_2 } \setminus \set{e}$ is nontrivial in approximate representations of 
        $\Gamma(A) / \ip{(t_1 t_2)^{p(n)}=e}$, and 
        \item each row of $A$ only has three nonzero entries.
    \end{enumerate}
\end{proposition}
\begin{proof}
    Let $H$ be the $m \times n'' \times \ell$ extended homogeneous linear-plus-conjugacy group from \cref{prop:exist},
    with elements $x,t,u \in H$. Let $\Gamma(A)$ be the $m' \times n'$ solution group 
    corresponding to $H$ from
    \cref{prop:lpc_to_sol}, and let 
    \begin{equation*}
        \widetilde{\phi} : \calF(x_0,\ldots,x_{n''-1},y_{0},\ldots,y_{\ell-1})
            \to \calF(x_0,\ldots,x_{n'-1}) \text{ and } \phi : H \to \Gamma(A)
    \end{equation*}
     be the homomorphisms from that proposition. By parts (b) and (c) of \cref{prop:lpc_to_sol}, 
    $\widetilde{\phi}(x) = x$, $\widetilde{\phi}(t) = t_1 t_2$, and $\widetilde{\phi}(u)
    = u_1 u_2$ for some generators $x, t_1, t_2, u_1, u_2 \in \{x_{0},\ldots,x_{n'-1}\}$. 
    Since $\phi$ is a homomorphism, 
    \begin{equation*}
        u_2 u_1 t_1 t_2 u_1 u_2 = \phi(u^{-1} t u) = \phi(t^r) = (t_1 t_2)^r,
    \end{equation*}
    proving part (a). 
    Let $\Gamma_n := \Gamma(A) / \langle (t_1 t_2)^{p(n)}=e \rangle$.
    If $n \in X$, then $x = e$ in $H / \langle t^{p(n)}=e
    \rangle$, so $x = e$ in $\Gamma_n$. 
    
    By part (e) of \cref{prop:lpc_to_sol}, there are integers $n_1$, $n_2$, $n_3$ such
    that for any $\ep$-representation $\psi$ of $H$, there is an $O(\ep)$-representation
    $\alpha$ of $\Gamma(A)$ such that 
    \begin{equation*}
        \alpha(\widetilde{\phi}(g)) = \psi(g)^{\oplus n_1} \oplus \overline{\psi(g)}^{\oplus n_2} \oplus \1_{n_3}
        \text{ for all } g \in \calF(x_0,\ldots,x_{n''-1},y_{0},\ldots,y_{\ell-1}).
    \end{equation*}
    If $\psi$ is an $\ep$-representation
    of $H / \langle t^{p(n)}=e \rangle$, then $\norm{\psi(t^{p(n)}) - \1}  \leq \ep$, so
    \begin{equation*}
        \norm{\alpha((t_1 t_2)^{p(n)}) - \1} = \norm{\alpha(\widetilde{\phi}(t^{p(n)})) - \1}
        = \norm{ \psi(t^{p(n)})^{\oplus n_1} \oplus
    \overline{\psi(t^{p(n)})}^{\oplus n_2} \oplus \1_{n_3} - \1 } \leq \ep
    \end{equation*}
    as well, and $\alpha$ is an $O(\ep)$-representation of $\Gamma_n$.
    
    Suppose $w \neq e$ in $H / \langle t^{p(n)}=e\rangle$. Since $H / \langle
    t^{p(n)}=e \rangle$ is sofic, $w$ is non-trivial in approximate
    representations of $H / \langle t^{p(n)}=e \rangle$, which means that there is 
    $\delta > 0$ such that for all $\ep > 0$, there is an $\ep$-representation
    $\psi$ with $\norm{\psi(w) - \1} \geq \delta$. If $\alpha$ is an
    $O(\ep)$-representation of $\Gamma_n$
    as above, then 
    \begin{equation*}
        \norm{ \alpha(\widetilde{\phi}(w)) - \1} \geq \sqrt{\frac{n_1 d + n_2 d}{n_1 d + n_2 d + n_3}} \delta
            \geq \sqrt{\frac{n_1 + n_2}{n_1 + n_2 + n_3}} \delta,
    \end{equation*}
    where $d$ is the dimension of $\psi$.
    Since $n_1$, $n_2$, and $n_3$ are independent of $\ep$, $\widetilde{\phi}(w)$ is non-trivial in approximate
    representations of $\Gamma_n$.

    If $n \not\in X$, then $x \neq e$ in $H / \langle t^{p(n)}=e \rangle$, so $x$ is non-trivial
    in approximate representations of $\Gamma_n$. In particular,
    $x \neq e$ in $\Gamma_n$, so parts (b) and (c) hold.
    Similarly, $t^i \neq e$ in $H / \langle t^{p(n)} = e \rangle$ for all $0 < i < p(n)$, and thus
    $(t_1 t_2)^i \neq e$ in $\Gamma_n$, proving part (d).
    This also shows that $(t_1t_2)^i$ is nontrivial in approximate representations of $\Gamma_n$
    for $0 < i < p(n)$, or 
    in other words that $(t_1t_2)^i \neq e$ in $\Gamma_n^{fa}$.
    This means that $t_1$ and $t_2$ must have order $2$ in $\Gamma_n^{fa}$.
    Thus $\ip{t_1, t_2}_{\Gamma_n^{fa}} \cong D_{p(n)}$.
    Since $\ip{t_1, t_2}_{\Gamma_n} \cong D_{p(n)}$ as well, this proves part (e).
    Finally, part (f) follows immediately from part (a) of \cref{prop:lpc_to_sol}.
\end{proof}

\section{Constant-sized correlations $\fC_p$ for the dihedral groups}
\label{sec:dihedral}

In \cref{S:embeddings}, we showed that any strategy for a perfect correlation of a
solution group must come from a representation of that group. The dihedral
group $D_n$ has a presentation as an extended homogeneous linear-plus-conjugacy
group, and thus can be embedded in solution groups. However, the size of this
solution group will depend on $n$. In this section we write down a constant-sized
correlation $\fC_p$ for the dihedral group $D_p$, $p$ a prime, such that any
commuting-operator strategy (meeting a condition which we will enforce in the
next section using perfect correlations for the solution group from
\cref{prop:existsg}) induces a representation of the dihedral group. The
construction we use comes from \cite{fu2019}, although we modify the
construction slightly so that the correlation comes from the regular
representation, rather than the representation used in \cite{fu2019}. This
modification is necessary for the next section, where we use the correlation
from this section in conjunction with perfect correlations for the solution
group in \cref{prop:existsg}.

To define $\fC_p$, recall that $D_p$ is generated by $t_1$ and $t_2$, and consists
of the elements $(t_1 t_2)^j$ and $t_2 (t_1 t_2)^j$ for $j \in [p]$. As in \cref{sec:prelim},
let $L$ and $R$ denote the left and right regular representations of $D_p$ on 
\begin{align*}
	\ell^2 D_{p} = \spn( \set{ \ket{(t_1t_2)^j}, \ket{t_2(t_1t_2)^j} \;|\; j \in [p]}).
\end{align*}
Define elements
\begin{equation}
\label{eq:idem}
\begin{aligned}
     	&\pi_{0}^{(0)} = \frac{1}{p} \sum_{j \in [p]} (t_1t_2)^j, \\
    	&\pi_{0}^{(1)} = \frac{2}{p} \sum_{j \in [p]} \cos(\frac{2j\pi}{p}) (t_1t_2)^j, \\
    	&\pi_{0}^{(2)} = e - \pi_{0}^{(0)} - \pi_{0}^{(1)},\\
	&\pi_{1}^{(0)} = \frac{1}{2} \pi_{0}^{(1)} + \frac{1}{p} \sum_{j \in [p]} \cos(\frac{(2j+1)\pi}{p}) t_2(t_1t_2)^j, \\
	&\pi_{1}^{(1)} = \pi_{0}^{(1)} - \pi_{1}^{(0)}, \\ 
	&\pi_{1}^{(2)} = e - \pi_{0}^{(1)},\\
	&\pi_{2}^{(0)} = \frac{1}{2} \pi_{0}^{(1)} + \frac{1}{p} \sum_{j \in [p]} \sin(\frac{(2j+1)\pi}{p}) t_2(t_1t_2)^j, \\
	&\pi_{2}^{(1)} =  \pi_{0}^{(1)} - \pi_{2}^{(0)},\text{ and} \\ 
	&\pi_{2}^{(2)} = e - \pi_{0}^{(1)} 
\end{aligned}
\end{equation}
in the group algebra $\C[D_p]$. These elements are all projections. For
instance, $\pi^{(0)}_0$ is the sum of central projections for the trivial and
sign representations.  $\pi^{(1)}_0$ is the central projection for the
$2$-dimensional irreducible representation $V^{(1)}$ sending 
\begin{equation*}
    t_1 \mapsto \begin{pmatrix} 0 & \omega_p \\ \omega_p^{-1} & 0 \end{pmatrix}
    \text{ and } t_2 \mapsto \begin{pmatrix} 0 & 1
\\ 1 & 0 \end{pmatrix}.
\end{equation*} 
The projections $\pi^{(0)}_1$ and $\pi^{(0)}_2$ are more complicated: they correspond to the
rank-one projections
\begin{equation*}
        \frac{1}{2} \begin{pmatrix} 1 & \omega_{2p} \\ \omega_{2p}^{-1} & 1 \end{pmatrix}
        \text{ and } 
        \frac{1}{2} \begin{pmatrix} 1 & -i\omega_{2p} \\ i\omega_{2p}^{-1} & 1 \end{pmatrix}
\end{equation*}
in $V^{(1)}$.

Before we define $\fC_p$, 
we first consider the correlation $\fC_p'$ for the
scenario $([5],[5],[3], [3])$ defined by 
\begin{align*}
	 \fC_{p}' (a,b | x,y ) =\frac{\bra{e}L(e - \pi_0^{(0)})  \tP_x^{(a)} \tQ_y^{(b)} L(e - \pi_0^{(0)}) \ket{e}}{\norm{L(e - \pi_0^{(0)})\ket{e}}^2}, x, y \in [5], a,b \in [3],
\end{align*}
where 
\begin{align*}
	& \tP_{x}^{(a)} = \begin{dcases}
			L(\pi_0^{(a+1)}) &\text{ if } x =0, a \in \set{0,1} \\
			\frac{L(e) + (-1)^{a} L(t_x)}{2} &\text{ if } x \in \{1,2\}, a \in \set{0,1} \\
			L(\pi_1^{(a)}) & \text{ if } x = 3 \\
			L(\pi_2^{(a)}) & \text{ if } x = 4 \\
            0 & \text{otherwise}
        \end{dcases}, \text{ and } \\
		& \tQ_y^{(b)} = \begin{dcases}
			R(\pi_0^{(b+1)}) &\text{ if } y = 0 , b \in \set{0,1} \\
			\frac{R(e) + (-1)^{b} R(t_y)}{2} &\text{ if } y \in \{1,2\}, b\in \set{0,1} \\
			R(\pi_1^{(b)}) & \text{ if } y = 3 \\
			R(\pi_2^{(b)}) & \text{ if } y = 4 \\
            0 & \text{otherwise}
        \end{dcases}.
\end{align*}	
Direct calculation gives us the following.
\begin{lemma}
	The correlation $\fC_p'$ is the correlation $\hat{P}_{-\pi/p}$ defined in \cite[Definition 3.2]{fu2019}.
\end{lemma}
One important property of $\hat{P}_{-\pi/p}$ is summarized in the following proposition.
\begin{proposition}[Proposition 3.3 of \cite{fu2019}]
	\label{prop:phi1}
	Let $(\ket{\psi}, \set{ \set{P_x^{(a)} }}, \set{\set{Q_y^{(b)} }})$ be an inducing strategy 
	of $\hat{P}_{-\pi/p}$, and let 
	\begin{align}
		\label{eq:phi1}
		\ket{\phi_1} = \frac{1}{2}\left( P_3^{(0)} + iP_4P_3^{(1)} -iP_4P_3^{(0)} +P_3^{(1)} \right) \ket{\psi},
	\end{align}
	where $P_4 = P_4^{(0)} - P_4^{(1)}$. 
	Then
	\begin{align*}
		\norm{\ket{\phi_1}}^2 = 1/(p-1),\,
		P_1P_2 \ket{\phi_1} = \omega_p^{-1} \ket{\phi_1},
		\text{ and }
		Q_1Q_2 \ket{\phi_1} = \omega_p \ket{\phi_1},
	\end{align*}
	where $P_x = P_x^{(0)} - P_x^{(1)}$ and $Q_y = Q_y^{(0)} - Q_y^{(1)}$
	for $x,y \in \set{1,2}$.
\end{proposition}
Note that the proof of Proposition 3.3 of \cite{fu2019} also tells us that
\begin{align}
	\label{eq:phi1_q} 
	\ket{\phi_1} = \frac{1}{2}\left( Q_3^{(0)} - iQ_4Q_3^{(1)} + iQ_4Q_3^{(0)} +Q_3^{(1)} \right) \ket{\psi},
\end{align}
where $Q_4 = Q_4^{(0)} - Q_4^{(1)}$.

Now we are ready to define $\fC_p$.
Let $I := \set{0,1,2, t_1, t_2, (0, t_1), (0, t_2)}$. We define $\fC_p$ as a correlation for the
scenario $(I,I,[3] \times [2], [3] \times [2])$ by 
\begin{align*}
    \fC_{p}(a,b | x,y ) = \bra{e} \tM_x^{(a)} \tN_y^{(b)} \ket{e}, x, y \in I, a,b \in [3] \times [2],
\end{align*}
where $\ket{e} \in L_2 D_p$, 
\begin{align*}
    & \tM_{x}^{(a_0,a_1)} = \begin{dcases}
			L(\pi_x^{(a_0)}) &\text{ if } x \in \{0,1,2\}, a_1=0, \\
			\frac{L(e) + (-1)^{a_1} L(x)}{2} &\text{ if } x \in \{t_1,t_2\}, a_0 =0, \\
			\tM_0^{(a_0,0)} \tM_t^{(0,a_1)} & \text{ if } x = (0,t) \text{ for } t \in \set{ t_1, t_2}  \\
            0 & \text{otherwise}
        \end{dcases}, \text{ and } \\
		& \tN_y^{(b_0,b_1)} = \begin{dcases}
			R(\pi_y^{(b_0)}) &\text{ if } y \in \{0,1,2\}, b_1 =0, \\
			\frac{R(e) + (-1)^{b_1} R(y)}{2} &\text{ if } y \in \{t_1,t_2\}, b_0=0, \\
			\tN_0^{(b_0,0)} \tN_t^{(0,b_1)} & \text{ if } y = (0,t) \text{ for } t \in \set{t_1, t_2} \\
            0 & \text{otherwise}
        \end{dcases}.
\end{align*}
It's easy to see that the families $\{\tM_x^{(a_0,a_1)} : (a_0,a_1) \in [3] \times [2]\}$ and
$\{\tN_y^{(b_0,b_1)} : (b_0,b_1) \in [3] \times [2]\}$ are projective measurements for
all $x,y \in \{0,1,2,t_1,t_2\}$. For $x,y \in \{(0,t_1),(0,t_2)\}$, this follows from the
fact that $\pi_0^{(a)}$ is a central projection for all $a \in [3]$. 
So $\fC_p$ is in $C_{qc}(I,I,[3]\times [2], [3] \times [2])$. 
Some of the most important entries of $\fC_p$ are shown in \cref{tb:fcp}.
We summarize some
additional properties in the following lemma:

\begin{lemma} \
    \begin{enumerate}[(a)]
        \item $\fC_p$ is a synchronous correlation in $C_{qc}(I,I,[3]\times
            [2], [3] \times [2])$.
        \item The entries of $\fC_p$ are computable elements of $\overline{\Q}$.
    \end{enumerate}
\end{lemma}
\begin{proof}
    Let $\iota : \C[D_p] \to \C[D_p]$ be the linear involution sending $g
    \mapsto g^{-1}$ for all $g \in D_p$. It is not hard to see that
    $\pi^{(i)}_j = \iota(\pi^{(i)}_j)$ for all $i$ and $j$. In addition,
    $R(\alpha) \ket{e} = L(\iota(\alpha)) \ket{e}$ for all $\alpha \in \C[D_p]$, 
    so $\tN_y^{(b)} \ket{e} = \tM_{y}^{(b)} \ket{e}$ for all $y \in I$ and $b
    \in [3] \times [2]$. So
    \begin{equation*}
        \bra{e} \tM_{x}^{(a)} \tN_x^{(b)} \ket{e} = \bra{e} \tM_{x}^{(a)}
        \tM_{x}^{(b)} \ket{e} = 0
    \end{equation*}
    for all $x \in I$ and $a \neq b \in [3] \times [2]$. We conclude that
    $\fC_p$ is synchronous. Part (b) follows from the fact that the elements
    $\pi^{(i)}_j$ belong to $\overline{\Q}[D_p]$. 
\end{proof}

We now come to the main theorem of this section.
To state this theorem, recall that an integer $r$ is a primitive root of a prime $p$ if all the integers between $1$ and $p-1$ are congruent
modulo $p$ to some power of $r$.
\begin{theorem}
\label{thm:fcp}
Let $S = (\ket{\psi}, \set{ M_x^{(a_0,a_1)} }, \set{ N_y^{(b_0,b_1)} }  )$ be a good strategy for $\fC_{p}$ and
let $r$ be a primitive root of $p$.
Suppose there exist unitaries $U_A$ and $U_B$ such that 
\begin{align*}
	&U_AU_B = U_BU_A, \\
	&U_A N_y^{(b_0, b_1)} = N_y^{(b_0, b_1)} U_A  \text{ for all } y, b_0, b_1, \\
	&U_B M_x^{(a_0, a_1)} = M_x^{(a_0, a_1)} U_B \text{ for all } x, a_0, a_1, \\
	&U_AU_B\ket{\psi}  = \ket{\psi}, \\
	&(N_{t_1}N_{t_2})U_B \ket{\psi} = U_B (N_{t_1}N_{t_2})^{r} \ket{\psi}, \text{ and } \\
	&(M_{t_1}M_{t_2}) U_A \ket{\psi} = U_A (M_{t_1}M_{t_2})^{r} \ket{\psi},
\end{align*}
where $M_{x} = M_{x}^{(0,0)} - M_x^{(0,1)}$ and $N_{y} = N_{y}^{(0,0)} - N_y^{(0,1)}$ for $x,y \in \set{t_1, t_2}$. Then
\begin{align*}
	(M_{t_1}M_{t_2})^{p}\ket{\psi} = \ket{\psi}.
\end{align*}
\end{theorem}
When we use this theorem in the next section, the existence of $U_A$ and $U_B$ will be guaranteed 
by \cref{prop:existsg}.

\begin{table}
\centering
\begin{subtable}[t]{\textwidth}
\centering
\begin{tabular}{|c|c||c|c|c|}
\hline
\multicolumn{2}{|c|}{} &
\multicolumn{3}{|c|}{$y= 0 $}\\
\cline{3-5}
\multicolumn{2}{|c|}{} &$b = (0,0)$ & $b=(1, 0)$ &
$b = (2, 0)$  \\
\hline
\hline
\multirow{3}{*}{$x = 0$} &
 $a=(0,0)$ & $1/p$ & $0$ 
& $0$  \\
\cline{2-5}
&$a=(1,0)$ & $0$ & $2/p$  
&  $0$   \\
\cline{2-5}
&$a=(2,0)$ & $0$ & $0$ &  
$(p-3)/p$ \\
\hline
\end{tabular}
\subcaption{$\fC_p$: the correlation values for $x = y = 0$.}
\label{tb:xy0}
\end{subtable}

\begin{subtable}[t]{\textwidth}
\centering
\begin{tabular}{|c|c||c|c|c|c|}
\hline
\multicolumn{2}{|c|}{} &
\multicolumn{2}{|c|}{$y=1$}&
\multicolumn{2}{|c|}{$y=2$} \\
\cline{3-6}
\multicolumn{2}{|c|}{} &
$b = (0,0)$ & $b=(1,0)$ &
$b = (0,0)$ & $b=(1,0)$ \\
\hline
\hline
\multirow{2}{*}{$x = t_1$} & $a=(0,0)$ & $\frac{\cos^2(\pi/2p)}{p}$ & $\frac{\sin^2(\pi/2p)}{p}$ 
& $\frac{1-\sin(\pi/p)}{2p}$ & $\frac{1+\sin(\pi/p)}{2p}$  \\
\cline{2-6}
&$a=(0,1)$ & $\frac{\sin^2(\pi/2p)}{p}$ & $\frac{\cos^2(\pi/2p)}{p}$  
&  $\frac{1+\sin(\pi/p)}{2p}$ & $\frac{1-\sin(\pi/p)}{2p}$  \\
\hline
\multirow{2}{*}{$x = t_2$} & $a=(0,0)$ & $\frac{\cos^2(\pi/2p)}{p}$ & $\frac{\sin^2(\pi/2p)}{p}$  & 
$ \frac{1+\sin(\pi/p)}{2p}$ & $ \frac{1-\sin(\pi/p)}{2p}$  \\
\cline{2-6}
&$a=(0,1)$ & $\frac{\sin^2(\pi/2p)}{p}$ & $\frac{\cos^2(\pi/2p)}{p}$ &  
$ \frac{1-\sin(\pi/p)}{2p}$ & $ \frac{1+\sin(\pi/p)}{2p}$ \\
\hline
\end{tabular}
\subcaption{$\fC_p$: the correlation values for $x \in \set{t_1, t_2}$ and $y \in \set{1, 2}$.}
\label{tb:achsh2}
\end{subtable}

\begin{subtable}[t]{\textwidth}
\centering
\begin{tabular}{|c|c||c|c|c|c|c|c|c|c|}
\hline
\multicolumn{2}{|c|}{} &
\multicolumn{3}{|c|}{$x=1$}&
\multicolumn{3}{|c|}{$x=2$}&
\multicolumn{2}{|c|}{$x=0$}\\
\cline{3-10}
\multicolumn{2}{|c|}{} &
$a_0 = 0$ & $a_0=1$ & $a_0=2$ &
$a_0 = 0$ & $a_0=1$ & $a_0=2$ &
$a_0 = 1$ & $a_0 \neq 1 $\\
\hline
\hline
\multirow{3}{*}{$y = 1$} & $b_0=0$ & $\frac{1}{p}$ & $0$ & 0 
& $\frac{1}{2p}$ & $\frac{1}{2p}$ & 0 & $\frac{1}{p}$ & 0 \\
\cline{2-10}
&$b_0=1$ & 0 & $\frac{1}{p}$ & $0$ 
&  $\frac{1}{2p}$ & $\frac{1}{2p}$ & 0 &$\frac{1}{p}$ & 0 \\
\cline{2-10}
&$b_0=2$ & 0 & 0 & $\frac{p-2}{p}$ 
&  0 & 0 &  $\frac{p-2}{p} $ &0 & $\frac{p-2}{p}$ \\
\hline
\multirow{3}{*}{$y = 2$} & $b_0=0$ & $\frac{1}{2p}$ & $\frac{1}{2p}$ & 0 
& $\frac{1}{p}$ & $0$ & 0 & $\frac{1}{p}$ & 0 \\
\cline{2-10}
&$b_0=1$ & $\frac{1}{2p}$ & $\frac{1}{2p}$ & $0$ 
&  0 & $\frac{1}{p}$ & $0$ &$\frac{1}{p}$ & 0 \\
\cline{2-10}
&$b_0=2$ & 0 & 0 & \small $\frac{p-2}{p}$ 
&  0 & 0 & \small $\frac{p-2}{p} $ &0 &\small $\frac{p-2}{p}$ \\
\hline
\multirow{2}{*}{$y = 0$} & $b_0=1$ & $\frac{1}{p}$ & $\frac{1}{p}$ & 0 
& $\frac{1}{p}$ & $\frac{1}{p}$ & 0 & $\frac{2}{p}$ & 0 \\
\cline{2-10}
&$b_0\neq 1$ & 0 & 0 & $\frac{p-2}{p}$ 
&  0 & 0 & \small $\frac{p-2}{p} $ &0 & \small $\frac{p-2}{p}$ \\
\hline
\end{tabular}
\subcaption{$\fC_p$: the correlation values for $x ,y\in\{0, 1, 2\} $.}
\label{tb:eqv_tb2}
\end{subtable}

\begin{subtable}[t]{\textwidth}
\centering
\scriptsize
\begin{tabular}{|c|c||c|c|c|c|c|c|}
\hline
\multicolumn{2}{|c|}{} &
\multicolumn{6}{|c|}{$y=(0, t_1)$}\\
\cline{3-8}
\multicolumn{2}{|c|}{} &
$b = (0,0)$ & $b=(0,1)$ & 
$b = (1,0)$ & $b=(1,1)$ &
$b = (2,0)$ & $b=(2,1)$ \\
\hline
\hline
\multirow{3}{*}{$x = 0$} & $a_0=0$ & $\frac{1}{2p}$ & $\frac{1}{2p}$ &  $0$
& $0$ & $0$ & $0$  \\
\cline{2-8}
&$a_0=1$ & $0$ & $0 $ & $\frac{1}{p}$ 
&  $\frac{1}{p}$ & $0$ & $0$  \\
\cline{2-8}
&$a_0=2$ & 0 & 0 & $0$ 
&  0 & $\frac{p-3}{2p}$ & $\frac{p-3}{2p} $ \\
\hline
\multirow{2}{*}{$x = t_1$} & $a_1=0$ & $\frac{1}{2p}$ & $0$ & $\frac{1}{p}$ 
& $0$ & $\frac{p-3}{2p}$ & $0$   \\
\cline{2-8}
&$a_1=1$ & $0$ & $\frac{1}{2p}$ & $0$ 
&  $\frac{1}{p}$ & $0$ & $\frac{p-3}{2p}$ \\
\hline
\end{tabular}
\subcaption{$\fC_p$: the correlation values for the commutation test for Alice's questions $0$ and $t_1$.}
\label{tb:comm2}
\end{subtable}

\begin{subtable}[t]{\textwidth}
\centering
\begin{tabular}{|c|c||c|c|}
\hline
\multicolumn{2}{|c|}{} &
\multicolumn{2}{|c|}{$y= (0,t_2) $}\\
\cline{3-4}
\multicolumn{2}{|c|}{} &$b = (0,0)$ & $b=(0, 1)$  \\
\hline
\hline
\multirow{2}{*}{$x = (0,t_1)$} &
 $a=(0,0)$ & $1/p$ & $0$  \\
\cline{2-4}
&$a=(0,1)$ & $0$ & $1/p$  \\
\hline
\end{tabular}
\subcaption{$\fC_p$: some of the values for $x =(0,t_1)$, $y = (0,t_2)$, $a_0 = b_0 = 0$.}
\label{tb:eqv_test2}
\end{subtable}
\caption{Some important values of $\fC_p$.}
\label{tb:fcp}
\end{table}

The rest of the section is devoted to the proof of \cref{thm:fcp}, so for the remainder of the section, we will assume
that we have a good strategy $S$ and unitaries $U_A$ and $U_B$ satisfying the conditions of the theorem.
Note that in a good strategy, $ M_x^{(a_0, a_1)} = N_x^{(a_0, a_1)} = 0$  if $x \in \set{0,1, 2}$ and $a_1 \neq 0$, or if $x \in \set{t_1, t_2}$
and $a_0 \neq 0$.
In particular, $M_{t_i}$ and $N_{t_i}$ are binary observables.
The basic idea of the proof is to find
a decomposition of $\ket{\psi}$ as $\ket{\psi} = \sum_{j = 0}^{p} \ket{\psi_j}$, where 
$\norm{\ket{\psi_0}}^2 = \norm{\ket{\psi_p}}^2 = 1/2p$, $\norm{\ket{\psi_j}}^2 = 1/p$ for $1 \leq j \leq p-1$,
and $\ket{\psi_j}$ is an eigenvector of $M_{t_1}M_{t_2}$ with eigenvalue $\omega_p^j$.
Intuitively, $\ket{\psi_0}$ and $\ket{\psi_p}$ are in the $1$-dimensional irreducible representation of $D_p$,
and $\ket{\psi_j}$ and $\ket{\psi_{p-j}}$ are in the $2$-dimensional irreducible representation of $D_p$ sending
\begin{align*}
	t_1t_2 \mapsto \begin{pmatrix}
	\omega_p^j & 0 \\
	0 & \omega_p^{-j} 
	\end{pmatrix}
\end{align*}
for $1 \leq j \leq (p-1)/2$. The norms of the vectors are chosen because
the multiplicity of the 2-dimensional irreducible representations in the regular representation of $D_p$ is 2, and the multiplicity of the 1-dimensional 
irreducible representations is 1.

The vectors $\ket{\psi_0}$ and $\ket{\psi_p}$ are defined as
\begin{align}
	\label{eq:def_psi0}
	\ket{\psi_0} = M_{t_1}^{(0,0)} M_0^{(0,0)} \ket{\psi}, \text{ and }
	\ket{\psi_p} = M_{t_1}^{(0,1)} M_0^{(0,0)} \ket{\psi}.
\end{align}
It follows immediately from the definition of $M_{t_1}$ that 
\begin{align}
	\label{eq:t1_0_val} &M_{t_1}\ket{\psi_0} = \ket{\psi_0}, \text{ and }
	 M_{t_1}\ket{\psi_p} = -\ket{\psi_p}.
\end{align}
To see that
\begin{align}	
	\label{eq:pis_0_p_norm} &\norm{\ket{\psi_0}}^2 = \norm{\ket{\psi_p}}^2 = \frac{1}{2p},
\end{align}
we need the following identities.
By \cref{prop:eqv_test},
\begin{align}
	\label{eq:eqv_MN}
	M_{(0, x)}^{(a_0, a_1)} \ket{\psi} = N_{(0, x)}^{(a_0, a_1)} \ket{\psi} 
\end{align} 
for $a_0 \in [3]$, $a_1 \in [2]$ and $x \in \set{t_1,t_2}$.
\Cref{prop:comm_test} applied to \cref{tb:comm2} implies that
\begin{align}
	\label{eq:t1_comm}
	M_x^{(0,a_1)} M_0^{(a_0,0)} \ket{\psi} = N_{(0, x)}^{(a_0, a_1)} \ket{\psi} = M_0^{(a_0,0)} M_x^{(0,a_1)}\ket{\psi}
\end{align}
for $a_0 \in [3]$, $a_1 \in [2]$ and $x \in \set{t_1,t_2}$.
Then
\begin{align*}
	\norm{\ket{\psi_0}}^2 = \bra{\psi} M_0^{(0,0)} M_{t_1}^{(0,0)} M_0^{(0,0)} \ket{\psi} = \bra{\psi} M_0^{(0,0)}
	N_{(0, t_1)}^{(0, 0)} \ket{\psi} = \frac{1}{2p}
\end{align*}
as shown in \cref{tb:comm2}. The derivation of $\norm{\ket{\psi_p}}^2 = \frac{1}{2p}$ is similar.
Next, \cref{lem:eqv_test} applied to \cref{tb:eqv_test2} implies that 
\begin{align*}
	&M_{(0,t_1)}^{(0, a_1)} \ket{\psi} = N_{(0,t_2)}^{(0, a_1)} \ket{\psi}
\end{align*}
for each $a_1 \in [2]$.
Hence, from the equation above, \cref{eq:t1_comm,eq:eqv_MN}
\begin{align*}
	\ket{\psi_0} = M_{t_1}^{(0, 0)} M_0^{(0,0)} \ket{\psi} = N_{(0, t_1)}^{(0, 0)} \ket{\psi} = M_{(0, t_1)}^{(0, 0)} \ket{\psi} =
	 N_{(0, t_2)}^{(0, 0)} \ket{\psi} = M_{t_2}^{(0,0)} M_0^{(0,0)}\ket{\psi},
\end{align*}
and similarly, $\ket{\psi_p} = M_{t_2}^{(0,1)} M_0^{(0,0)}\ket{\psi}$.
Thus
\begin{align}
	\label{eq:t2_0_val} M_{t_2}\ket{\psi_0} = \ket{\psi_0}, \text{ and }  M_{t_2}\ket{\psi_p} = -\ket{\psi_p},
\end{align}
implying that $\ket{\psi_0}$ and $\ket{\psi_p}$ are
$1$-eigenvectors of $M_{t_1}M_{t_2}$.

The vectors $\ket{\psi_j}$ for $2 \leq j \leq p-1$ can be constructed from $\ket{\psi_1}$,
but the construction of $\ket{\psi_1}$ is more complicated and requires \cref{prop:phi1}.
We first show a strategy for $\fC_p'$ can be extracted from any strategy for $\fC_p$.
\begin{proposition}
\label{prop:cpprime}
Suppose $(\ket{\psi}, \set{M_x^{(a)} \mid a \in [3] \times [2]}, x \in I, 
\set{N_y^{(b)} \mid b \in [3] \times [2]}, y \in I)$ is a good strategy for $\fC_p$. Let $\ket{\psi'} = (\1 - M_0^{(0,0)}) \ket{\psi}/ \norm{ (\1 - M_0^{(0,0)}) \ket{\psi} }$,
and define projective measurements $\set{ P_x^{(a)} \mid a \in [3] }$, $x \in [5]$, and
$\set{ Q_y^{(b)} \mid b \in [3] }$, $y \in [5]$ by
\begin{align*}
		&P_0^{(0)} = M_0^{(1,0)}, && Q_0^{(0)} = N_0^{(1,0)}, \\
		&P_0^{(1)} = M_0^{(2,0)}, && Q_0^{(1)} = N_0^{(2,0)}, \\
		&P_0^{(2)} = 0, && Q_0^{(2)} = 0, \\
		\text{for } x = 1, 2:\quad
		&
		P_x^{(a)} =  \begin{cases}
			M_{t_x}^{(0,a)} &\text{ if } a = 0, 1\\
			0 &\text{ otherwise}
		\end{cases},
		&&
		Q_x^{(a)} =  \begin{cases}
			N_{t_x}^{(0,a)} &\text{ if } a = 0, 1\\
			0 &\text{ otherwise}
		\end{cases}, \\
		&P_3^{(a)} =  M_1^{(a,0)},		&&
		Q_3^{(a)} =  N_1^{(a,0)}, \\
		&P_4^{(a)} =  M_2^{(a,0)},	 \text{ and }	&&
		Q_4^{(a)} =  N_2^{(a,0)}.
	\end{align*}
	Then $(\ket{\psi'}$, $\set{ P_x^{(a)} \mid a \in [3] }$, $x \in [5]$,$\set{ Q_y^{(b)} \mid b \in [3] }$, $y \in [5])$
	is a strategy for $\fC_p'$.
\end{proposition}
\begin{proof}
For this proof,
the first key observation is that 
\begin{align*}
	\norm{ ( \1 - M_0^{(0,0)})\ket{\psi}}^2 = (p-1)/p,
\end{align*}
which follows from the fact that $\norm{ M_0^{(0,0)} \ket{\psi}}^2 = 1/p$.
We also need the following identities:
\begin{equation}
	\label{eq:eqv_rel1}
\begin{aligned}
	(M_0^{(0,0)} &+ M_0^{(2,0)}) \ket{\psi} = M_1^{(2,0)} \ket{\psi} = M_2^{(2,0)} \ket{\psi} \\
	&= N_1^{(2,0)} \ket{\psi} = N_2^{(2,0)} \ket{\psi} = (N_0^{(0,0)} + N_0^{(2,0)}) \ket{\psi},
\end{aligned}
\end{equation}
and
\begin{equation}
\label{eq:eqv_rel2}
\begin{aligned}
	(M_1^{(0,0)} &+ M_1^{(1,0)}) \ket{\psi}  = (M_2^{(0,0)} + M_2^{(1,0)}) \ket{\psi} = M_0^{(1,0)} \ket{\psi} \\
	&=  N_0^{(1,0)} \ket{\psi} = (N_1^{(0,0)} + N_1^{(1,0)}) \ket{\psi}  = (N_2^{(0,0)} + N_2^{(1,0)}) \ket{\psi}.
\end{aligned}
\end{equation}
To prove \cref{eq:eqv_rel1}, first observe that 
$\set{M_i^{(0,0)} , M_i^{(1,0)}, M_i^{(2,0)}}$ and
$\set{N_i^{(0,0)} , N_i^{(1,0)}, N_i^{(2,0)}}$ are projective measurements for $i =0,1,2$.
By \cref{tb:eqv_tb2},
the two projective measurements
 $\set{ M_0^{(1,0)},  (M_0^{(0,0)} + M_0^{(2,0)})}$ and $\set{ (N_1^{(0,0)} + N_1^{(1,0)}), N_1^{(2,0)}}$ satisfy the conditions of \cref{lem:eqv_test} with respect to $\ket{\psi}$, and hence we have $(M_0^{(0,0)} + M_0^{(2,0)}) \ket{\psi} = N_1^{(2,0)} \ket{\psi}$. 
For the same reason, $N_1^{(2,0)} \ket{\psi} = M_1^{(2,0)} \ket{\psi} =(N_2^{(0,0)} + N_2^{(1,0)}) \ket{\psi}$ and so on.
\Cref{eq:eqv_rel2} follows from \cref{eq:eqv_rel1} as $\1 - (M_0^{(0,0)} + M_0^{(2,0)}) = M_0^{(1,0)}$, $\1 - N_1^{(2,0)} =  N_1^{(0,0)} + N_1^{(1,0)}$
and so on.

Since $\norm{ M_x^{(a,0)} N_0^{(2,0)} \ket{\psi}}^2 = \bra{\psi} M_x^{(a,0)} N_0^{(2,0)} \ket{\psi} = 0$
from \cref{tb:comm2}, we have $M_x^{(a,0)} N_0^{(2,0)} \ket{\psi} = 0$.
Using the identities in \cref{eq:eqv_rel1,eq:eqv_rel2}, we can prove that, 
for $x = 1, 2$ and $a \in \set{0,1}$,
\begin{align*}
	M_x^{(a,0)} (\1 - M_0^{(0,0)})\ket{\psi} 
	&=  M_x^{(a,0)}(N_0^{(1,0)} + N_0^{(2,0)})\ket{\psi} 
	=  M_x^{(a,0)} N_0^{(1,0)} \ket{\psi} \\
	&=  M_x^{(a,0)} (M_x^{(0,0)} + M_x^{(1,0)} \ket{\psi} 
	= M_x^{(a,0)}\ket{\psi}.
\end{align*}
That is,
\begin{align}
	\label{eq:mxa}
	&M_x^{(a, 0)} \ket{\psi} = M_x^{(a,0)} (\1 - M_0^{(0,0)})\ket{\psi}.
\end{align}
The same argument can also give us that, for $x = 1, 2$ and $a \in \set{0,1}$,
\begin{align}
	\label{eq:nxa}
	&N_x^{(a, 0)} \ket{\psi} = N_x^{(a,0)} (\1 - M_0^{(0,0)})\ket{\psi}.
\end{align}
Following the definitions of $M_0^{(a,0)}$ and $N_0^{(a,0)}$ and using $N_0^{(0,0)}\ket{\psi} = M_0^{(0,0)}\ket{\psi}$, we can also see that
\begin{align}
	\label{eq:m0a}
	M_0^{(a,0)} \ket{\psi} = M_0^{(a,0)} (\1 - M_0^{(0,0)})\ket{\psi}  \text{  and  }
	N_0^{(a,0)} \ket{\psi} = N_0^{(a,0)} (\1 - M_0^{(0,0)})\ket{\psi}
\end{align}
for $a = 1, 2$.

To see the extracted strategy induces $\fC_p'$, we need to determine
$\bra{\psi} (\1 - M_0^{(0,0)}) P_x^{(a)} Q_y^{(b)} (\1 - M_0^{(0,0)}) \ket{\psi}$ 
for $x,y \in [5]$ and $a,b \in [3]$ from the values of $\fC_p$.
If $x,y  \in \set{0,3,4}$, it follows from the definitions of $P_x^{(a)}$ and $Q_y^{(b)}$
and \cref{eq:mxa,eq:nxa,eq:m0a} that 
\begin{align*}
	\bra{\psi} (\1 - M_0^{(0,0)}) P_x^{(a)} Q_y^{(b)} (\1 - M_0^{(0,0)}) \ket{\psi}
	= \bra{\psi} P_x^{(a)} Q_y^{(b)} \ket{\psi}.
\end{align*}
If  $x \in \set{1, 2}$ and $y \in \set{0, 3, 4}$, \cref{eq:t1_comm,eq:m0a,eq:nxa} imply that
\begin{align*}
	\bra{\psi} (\1 - M_0^{(0,0)}) P_x^{(a)} Q_y^{(b)} (\1 - M_0^{(0,0)}) \ket{\psi} 
	=\bra{\psi} P_x^{(a)} Q_y^{(b)} (\1 - M_0^{(0,0)})^2 \ket{\psi} 
	= \bra{\psi} P_x^{(a)} Q_y^{(b)} \ket{\psi}.
\end{align*}
The same identity holds for  $x \in \set{0, 3, 4}$ and $y \in \set{1,2}$ by symmetry.
Finally, when $x, y \in \set{1,2}$,  \cref{eq:t1_comm} implies that 
 \begin{align*}
 	\bra{\psi} (\1 - M_0^{(0,0)}) P_x^{(a)} Q_y^{(b)} (\1 - M_0^{(0,0)}) \ket{\psi} 
	=  \bra{\psi} P_x^{(a)} (\1 - M_0^{(0,0)})^2 Q_y^{(b)} \ket{\psi} \\
	=  \bra{\psi} (\1 - M_0^{(0,0)}) P_x^{(a)} Q_y^{(b)} \ket{\psi}.
 \end{align*}
If $a = 0$,
\begin{align*}
	P_x^{(0)} (\1 - M_0^{(0,0)}) \ket{\psi} 
	= M_{t_x}^{(0,0)} (\1 - M_0^{(0,0)}) \ket{\psi} 
	= M_{t_x}^{(0,0)} \ket{\psi} - \ket{\psi_0},
\end{align*}
where $\ket{\psi_0} = M_{t_x}^{(0,0)} M_0^{(0,0)} \ket{\psi}$. 
Observe that, because $\fC_p$ is synchronous, 
\begin{align*}
\ket{\psi_0} = M_{t_x}^{(0,0)} (M_0^{(0,0)})^2 \ket{\psi} = M_{t_x}^{(0,0)} M_0^{(0,0)} N_0^{(0,0)}\ket{\psi} = N_0^{(0,0)}\ket{\psi_0}.
\end{align*}
Hence,
\begin{align*}
	\bra{\psi} (\1 - M_0^{(0,0)}) P_x^{(a)} Q_y^{(b)} \ket{\psi} 
	&= \bra{\psi} M_{t_x}^{(0,0)}Q_y^{(b)} \ket{\psi} - \bra{\psi_0} Q_y^{(b)} \ket{\psi} \\
	&= \bra{\psi} M_{t_x}^{(0,0)}Q_y^{(b)} \ket{\psi} - \bra{\psi_0} N_0^{(0,0)} Q_y^{(b)} \ket{\psi} \\
	&=\bra{\psi} M_{t_x}^{(0,0)}Q_y^{(b)} \ket{\psi} - \braket{\psi_0}{\psi_j},
\end{align*}
where $j = 0$ if $b = 0$ and $j = p$ otherwise.
Since $\ket{\psi_0}$ and $\ket{\psi_p}$ are $1$ and $-1$-eigenvectors of $M_{t_1}$ by \cref{eq:t1_0_val},
$\braket{\psi_0}{\psi_p} = 0$.
If $a = 1$, the same calculation holds except that $\bra{\psi_0}$ is replaced by $\bra{\psi_p}$ and $M_{t_x}^{(0,0)}$ is replaced 
by $M_{t_x}^{(0,1)}$.
Thus in all cases, we can calculate
\begin{align*}
 \bra{\psi} (\1 - M_0^{(0,0)}) P_x^{(a)} Q_y^{(b)} (\1 - M_0^{(0,0)}) \ket{\psi}
 \end{align*}
 from $\fC_p$,
and direct calculation shows that the induced correlation is $\fC_p'$.
\end{proof}
We can now define $\ket{\psi_1}$ based on \cref{prop:phi1,prop:cpprime}.
\begin{corollary}
	\label{cor:psi1}
	Let $(\ket{\psi}, \set{M_x^{(a)} \mid a \in [3] \times [2]}, x \in I, 
\set{N_y^{(b)} \mid b \in [3] \times [2]}, y \in I)$ be a good strategy for $\fC_p$, 
and let
\begin{align}
	\label{eq:psi1}
	\ket{\psi_1} =\frac{1}{2}(M_{1}^{(0,0)} + iM_{2}M_{1}^{(1,0)} - iM_{2}M_{1}^{(0,0)} +M_{1}^{(1,0)}) \ket{\psi},
\end{align}
where $M_2 = M_2^{(0,0)} - M_2^{(1,0)}$.
Then
\begin{align}
	\label{eq:psi1_n}
	&\ket{\psi_1} = \frac{1}{2} (N_{1}^{(0,0)} - iN_{2}N_{1}^{(1,0)} + iN_{2}N_{1}^{(0,0)} +N_{1}^{(1,0)}) \ket{\psi}, \\
	\label{eq:psi1_norm}
	&\norm{\ket{\psi_1}}^2= \frac{1}{p} = \braket{\psi}{\psi_1},\\
	\label{eq:psi_1_m_val}
	&M_{t_1}M_{t_2}\ket{\psi_1} = \omega_{p}^{-1} \ket{\psi_1}, \text{ and }\\
	\label{eq:psi_1_n_val}
	&N_{t_1}N_{t_2} \ket{\psi_{1}} = \omega_{p} \ket{\psi_1},
\end{align}
where $M_{t_x} = M_{t_x}^{(0,0)} - M_{t_x}^{(0,1)}$ and $N_{t_x} = N_{t_x}^{(0,0)} - N_{t_x}^{(0,1)}$
for $x = 1,2$.
\end{corollary}
\begin{proof}
Suppose $(\ket{\psi}, \set{M_x^{(a)} \mid a \in [3] \times [2]}, x \in I, 
\set{N_y^{(b)} \mid b \in [3] \times [2]}, y \in I)$ is a good strategy for $\fC_p$.
Then
\begin{align*}
\braket{\psi}{\psi_1} &= \frac{1}{2} \bra{\psi}(N_{1}^{(0,0)} + i(M_{2}^{(0,0)} - M_2^{(1,0)}) N_{1}^{(1,0)} - i(M_{2}^{(0,0)} - M_2^{(1,0)})N_{1}^{(0,0)} +N_{1}^{(1,0)})\ket{\psi} 
=\frac{1}{p}
\end{align*}
by \cref{tb:eqv_tb2}.
Let  $(\ket{\psi'}$, $\set{ P_x^{(a)} \mid a \in [3] }$, $x \in [5]$,$\set{ Q_y^{(b)} \mid b \in [3] }$, $y \in [5])$
	be the strategy for $\fC_p'$ from \cref{prop:cpprime},
	and let $\ket{\phi_1} = \frac{1}{2}\left( P_3^{(0)} + iP_4P_3^{(1)} -iP_4P_3^{(0)} +P_3^{(1)} \right) \ket{\psi}$
	as in \cref{prop:phi1}.
Expanding $P_x^{(a)}$, we see that 
\begin{align*}
	\norm{(\1 - M_0^{(0,0)})\ket{\psi}} \cdot  \ket{\phi_1} &=\frac{1}{2}(M_{1}^{(0,0)} - iM_{2}M_{1}^{(1,0)} + iM_{2}M_{1}^{(0,0)} +M_{1}^{(1,0)}) (\1 - M_0^{(0,0)})\ket{\psi} \\
	&= \frac{1}{2} (M_{1}^{(0,0)} - iM_{2}M_{1}^{(1,0)} + iM_{2}M_{1}^{(0,0)} +M_{1}^{(1,0)}) \ket{\psi} \\
	&= \ket{\psi_1},
\end{align*}
where the second equality follows from \cref{eq:mxa}. \Cref{eq:psi1_n} follows from expanding \cref{eq:phi1_q} similarly.
Since $M_{t_i} = P_i^{(0)} - P_i^{(1)}$, $N_{t_i} = Q_i^{(0)} - Q_i^{(1)}$ for $i = 1,2$, 
\cref{eq:psi1_norm,eq:psi_1_m_val,eq:psi_1_n_val} follow immediately from \cref{prop:phi1}
and the fact that $\norm{ (\1 - M_0^{(0,0)})\ket{\psi}}^2 = (p-1)/p$.
\end{proof}

We can now finish the proof of \cref{thm:fcp}.
\begin{proof}[Proof of \cref{thm:fcp}]
We have already defined $\ket{\psi_0}$ and $\ket{\psi_p}$ in \cref{eq:def_psi0},
and $\ket{\psi_1}$ in \cref{cor:psi1}.
Next we define $\ket{\psi_j}$ for $2 \leq j \leq p-1$ as 
\begin{align*}
	\ket{\psi_j} = (U_AU_B)^{\log_r j} \ket{\psi_1}
\end{align*}
where $\log_r j$ is the discrete log of $j$ modulo $p$ (in other words, $\log_r j = a$ where $r^a \equiv j \pmod{p}$).
The discrete log is defined for all $1 \leq j \leq p-1$ because $r$ is a primitive root of $p$.
Since $U_A$ and $U_B$ are unitary, $\norm{\ket{\psi_j}}^2  = 1/p$.
To prove
\begin{align}
	\label{eq:psij_m_val}
	(M_{t_1}M_{t_2}) \ket{\psi_j} = \omega_p^{-j} \ket{\psi_j}, \text{ and }
	(N_{t_1}N_{t_2}) \ket{\psi_{j}} = \omega_p^{j} \ket{\psi_{j}},
\end{align}
observe that since $M_{t_i}\ket{\psi} = N_{t_i}\ket{\psi}$ by \cref{prop:eqv_test}, $(M_{t_1}M_{t_2})^n \ket{\psi} = (N_{t_1}N_{t_2})^{-n} \ket{\psi}$.
Similarly $U_A^n\ket{\psi} = U_B^{-n} \ket{\psi}$ from the hypothesis $U_AU_B\ket{\psi} = \ket{\psi}$.
Thus
\begin{align*}
	(M_{t_1}M_{t_2})^n U_A \ket{\psi} &= (M_{t_1}M_{t_2})^{n-1} U_A (M_{t_1}M_{t_2})^r \ket{\psi} \\
	& = (N_{t_1}N_{t_2})^{-r}(M_{t_1}M_{t_2})^{n-1} U_A \ket{\psi} \\
	& \ldots \\
	& =  (N_{t_1}N_{t_2})^{-nr} U_A \ket{\psi} \\
	& = U_A (M_{t_1}M_{t_2})^{nr} \ket{\psi}.
\end{align*}
Hence
\begin{align*}
	(M_{t_1}M_{t_2}) U_A^n \ket{\psi} &= (M_{t_1}M_{t_2}) U_A (U_B^\dagger)^{n-1} \ket{\psi} \\
	&= (U_B^\dagger)^{n-1} U_A (M_{t_1}M_{t_2})^r \ket{\psi} \\
	&=(U_B^\dagger)^{n-2} U_A (M_{t_1}M_{t_2})^r U_A \ket{\psi}  \\
	&=(U_B^\dagger)^{n-2} U_A^2 (M_{t_1}M_{t_2})^{r^2} \ket{\psi} \\
	&\ldots \\
	& = U_A^n (M_{t_1}M_{t_2})^{r^n} \ket{\psi}.
\end{align*}
Then
\begin{align*}
	(M_{t_1}M_{t_2}) \ket{\psi_j} &= (M_{t_1}M_{t_2}) (U_AU_B)^{\log_r j} \ket{\psi_1} \\
	&= \frac{1}{2} U_B^{\log_r j} (N_{1}^{(0,0)} - iN_{2}N_{1}^{(1,0)} + iN_{2}N_{1}^{(0,0)} +N_{1}^{(1,0)})(M_{t_1}M_{t_2}) U_A^{\log_r j}  \ket{\psi} \\
	&= \frac{1}{2} U_B^{\log_r j}  (N_{1}^{(0,0)} - iN_{2}N_{1}^{(1,0)} + iN_{2}N_{1}^{(0,0)} +N_{1}^{(1,0)})U_A^{\log_r j} (M_{t_1}M_{t_2})^j \ket{\psi} \\
	&=(U_AU_B)^{\log_r j}  (M_{t_1}M_{t_2})^{j} \ket{\psi_1} 
	= \omega_p^{-j} (U_AU_B)^{\log_r j} \ket{\psi_1} 
	= \omega_p^{-j} \ket{\psi_j}.
\end{align*}

Let 
\begin{align*}
	\ket{\psi'} = \sum_{j=0}^{p} \ket{\psi_j}. 
\end{align*}
Since eigenvectors with different eigenvalues are orthogonal,
$\braket{\psi_j}{\psi_k} = 0$
for $0 \leq j \neq k \leq p$.
As a result, $\norm{\ket{\psi'}} = 1$.
If $j = 0$ or $p$, then $\braket{\psi}{\psi_j} = \norm{\ket{\psi_j}}^2 = 1/2p$.
If $1 \leq j \leq p-1$, then $\braket{\psi}{\psi_j} = \bra{\psi} (U_AU_B)^{\log_r j} \ket{\psi_1} = \braket{\psi}{\psi_1} = 1/p$
using \cref{cor:psi1} and the fact that $U_AU_B \ket{\psi} = \ket{\psi}$.
Thus,
\begin{align*}
	\braket{\psi}{\psi'} =& \braket{\psi}{\psi_0} + \braket{\psi}{\psi_p} 
	+ \sum_{j=1}^{p-1} \braket{\psi}{\psi_j}  \\
	= & \frac{1}{2p} + \frac{1}{2p} + (p-1) \frac{1}{p} = 1,
\end{align*}
implying that $\ket{\psi} = \ket{\psi'}$.
We conclude that
\begin{align*}
	(M_{t_1}M_{t_2})^{p} \ket{\psi} 
	= (M_{t_1}M_{t_2})^{p} (\sum_{j=0}^{p} \ket{\psi_j})
	=  \sum_{j=0}^{p} \omega_{p}^{-jp} \ket{\psi_j}
	=\ket{\psi},
\end{align*}
which completes the proof.
\end{proof}

\section{Membership problems}
\label{sec:main_cqa}
Recall that $\Membership(n_A,n_B,m_A,m_B)_{t,\K}$, where $t \in \set{q, qs, qa, qc}$ and $\K$ is a subfield of $\R$,
is defined in \cref{sec:intro} as the problem of deciding if a correlation $P \in \K^{n_An_Bm_Am_B}$ is in the correlation set $C_t(n_A,n_B,m_A,m_B)$.
In this section, we let $\K = \K_0 \cap \R$, where $\K_0$ is the subfield of $\C$ generated by
the roots of unity $\omega_n^k$ for $k,n \in \Z$. We then drop the subscript $\K$ when referring to membership problems.
The hardness of $\Membership(n_A,n_B,m_A,m_B)_{t}$ is related to the hardness of a more general problem:

 \begin{intersection}
	Given a finite set of correlations $F \subset \K^{n_An_Bm_Am_B}$
	with constants $n_A$, $n_B$, $m_A$ and $m_B$,
	is $F \cap C_{t}(n_A$, $n_B$, $m_A$, $m_B) \neq \emptyset$?
\end{intersection}

 \begin{proposition}
 	\label{prop:equal_hard}
 	For fixed $n_A$, $n_B$, $m_A$, $m_B \in \N$ and $t \in \set{q, qs, qa, qc}$, 
	\begin{align*}
		\Intersection(n_A, n_B, m_A, m_B)_{t} \text{ and } \Membership(n_A, n_B, m_A, m_B)_{t}
	\end{align*}
	 are equivalent under Cook reduction.
\end{proposition}
\begin{proof}
	If $D_M$ is a decider for $\Membership(n_A$, $n_B$, $m_A$, $m_B)_{t}$, we can decide
	if $F \cap C_t(n_A$, $n_B$, $m_A$, $m_B) = \emptyset$ by
	running $D_M$ on all the elements of $F$.
	If $D_I$ is a decider for $\Intersection(n_A$, $n_B$, $m_A$, $m_B)_{t}$, we can decide
	if a correlation $P \in C_t(n_A$, $n_B$, $m_A$, $m_B)$ by running $D_I$ on $\set{P}$. 	
\end{proof}

The main result of this section is the following.
\begin{theorem}
	\label{thm:cqa}
	For every recursively enumerable set $X$ of positive integers,
	 there exists $N \in \N$ and
    	a computable family of finite sets of correlations $\set{F_n \mid n \in \N}$, where $F_n \subset \K^{N^2 \times 8^2}$, 
    	such that 
	\begin{align*}
	&F_n \cap C_{qc}(N,N, 8,8) = \emptyset \text{ if } n \in X, \text{ and }\\
	&F_n \cap C_{qa}(N,N, 8,8) \neq \emptyset \text{ if } n \notin X.
	\end{align*}
\end{theorem}

Before proving \cref{thm:cqa}, we first observe that
\Cref{thm:mainintro} follows directly from \cref{thm:cqa}.

\begin{proof}[Proof of \cref{thm:mainintro}]
	Let $X$ be a $\RE$-complete set of positive integers, and let $N$ and
 $F_n  \subset \K^{N^2 \times 8^2}$ be as in \cref{thm:cqa}.
 Set $\alpha = \max(N, 8)$, and suppose $n_A$, $n_B$, $m_A$, $m_B \geq \alpha$.
	For any $n \in \N$ and $C \in F_n$, define $C' \in \K^{n_An_Bm_Am_B}$ by
	\begin{align*}
		C'(a,b\mid x, y) = 
		\begin{dcases}
			C(a,b \mid \min(x,N), \min(y,N)) &\text{ if } a,b < 8 \\
			0 &\text{ otherwise}
		\end{dcases}.
	\end{align*}
	It follows easily from the definitions that
	\begin{align*}
		C'  \in 
		C_{t}(n_A, n_B, m_A, m_B) 
		\text{ if and only if } C \in  C_{t}(N,N, 8,8).
	\end{align*}
	Hence if $F_n' = \set{ C' \mid C \in F_n}$, then
	\begin{align*}
		F_n' \cap C_t(n_A, n_B, m_A, m_B) \neq \emptyset
		\text{ if and only if } 
		F_n \cap C_t(N, N, 8, 8) \neq \emptyset.
	\end{align*}
	Since $C_{qa}(n_A$, $n_B$, $m_A$, $m_B) \subseteq C_{qc}(n_A$, $n_B$, $m_A$, $m_B)$,
	\cref{thm:cqa} implies that
	\begin{align*}
		F_n' \cap C_t(n_A, n_B, m_A, m_B) \neq \emptyset
		\text{ if and only if } 
		n \notin X
	\end{align*}
	for both $t = qa$ and $t= qc$.
	Thus $\Intersection(n_A$, $n_B$, $m_A$, $m_B)_{t}$ is $\coRE$-hard for $t = qa, qc$.
	By \cref{prop:equal_hard}, $\Membership(n_A$, $n_B$, $m_A$, $m_B)_{t}$
	is also $\coRE$-hard.
	\end{proof}

Although we take $m_A, m_B \geq \alpha$ in the proof, note that it is sufficient to choose $m_A$, $m_B \geq 8$.
Also, it has been shown that $\Membership(n_A$, $n_B$, $m_A$, $m_B)_{qc}$ is in $\coRE$ \cite{npa2008}. 
	Hence, $\Membership(n_A$, $n_B$, $m_A$, $m_B)_{qc}$ 
	is $\coRE$-complete for $n_A, n_B \geq N$ and $m_A, m_B \geq 8$.
	
To prove \cref{thm:cqa}, we first construct $F_n$ from $X$.
Recall that an integer $r$ is a primitive root of a prime $p$ if all the integers between $1$ and $p-1$ are congruent
modulo $p$ to some power of $r$.
By a result of Gupta and Murty \cite{gupta-murty}, there are integers $r$ which are primitive roots of infinitely many primes.
We use a version of this result due to Heath-Brown.
\begin{lemma}[\cite{heath-brown}, see also \cite{murty1988}]
    \label{lm:prim_rt}
    There exists $r \in \set{2, 3, 5}$ such that $r$ is a primitive root of infinitely many primes.
\end{lemma}
For the construction of $F_n$, fix $r \in \set{2, 3, 5}$ such that $r$ is a 
primitive root of infinitely many primes.
Let $p(n)$ be the $n$-th prime greater than $r$ for which $r$ is a primitive root.
Since we can decide if $r$ is a primitive root of a given prime,
the sequence of primes $p(1) < p(2) < \ldots$ is computable.
Let $\Gamma(A)$ be the solution group from \cref{prop:existsg} for the function $p$, set $X$ and integer $r$.
Let $m$ and $\ell$ be the number of rows and columns of $A$ respectively,
and note that each row of $A$ has three nonzero entries.
Recall that the generating set $\set{ x_i \mid i \in [\ell]}$ has special generators $x, t_1, t_2, u_1$ and $u_2$.
By reordering the generators, we can take $x =x_0$, which lets us use $x$ for other things.
Recall that for perfect correlations associated with $Ax = 0$, we use question set $ [m] \cup \calX_{var}$,
where $\calX_{var} = \set{x_i \mid i \in [\ell]}$ is the set of variables in the system, and answer set $\Z_2^3$.
For correlations in $F_n$,
the question set is
$\calX= \calX_{var} \cup [m] \cup \set{m, m+1, m+2, (m, t_1), (m, t_2)}$,
and the answer set is $\calA = \Z_2^3$.
The questions $m$, $m+1$, $m+2$, $(m, t_1)$ and $(m, t_2)$ will correspond to questions $0$, $1$, $2$, 
$(0, t_1)$ and $(0,t_2)$
from the correlation
$\fC_{p(n)}$.

To define the entries of the correlations in $F_{n}$, 
we use the notations from \cref{def:perfectcorrelations}.
In particular, $I_j = \set{k \in [\ell] \mid A(j,k) \neq 0}$ for $j \in [m]$.
Let 
\begin{align*}
	G_n = \ip{ x_0, x_1,\ldots x_{\ell-1} : x_j^2 = e \text{ for all } j \in [\ell], [x_j, x_k] = e \text{ if } j,k \in I_i \text{ for some } i,  (t_1t_2)^{p(n)} = e },
\end{align*}
so that $\Gamma(A)/\ip{(t_1t_2)^{p(n)}=e}$ is a quotient of $G_n$ (note again that $t_1$ and $t_2$ are generators).
$G_n$ is a Coxeter group, so its word problem is decidable \cite[Chapter 5]{coxeter}.
Specifically, two words $w_0$ and $w_1$ over the generators of $G_n$ are equal in $G_n$ if they can both be transformed 
into a third word using the transformations
\begin{align*}
	&x_j^2 \to e, \\
	&x_j x_k \to x_k x_j \text{ if } j,k \in I_i \text{ for some } i, \\
	&\underbrace{t_1t_2 \ldots t_1}_{\text{length } p(n)} \to \underbrace{t_2 t_1 \ldots t_2}_{\text{length } p(n)}, \text{ and } \\
	&\underbrace{t_2t_1 \ldots t_2}_{\text{length } p(n)} \to \underbrace{t_1 t_2 \ldots t_1}_{\text{length } p(n)}.
\end{align*}
Since the transformations never increase the length of a word,
determining if two words are equal is a finite problem.

Recall that a group algebra like $\C[G_n]$ is a $\ast$-algebra under the operation
$ (\sum_g \alpha_g g)^\ast = \sum_g \overline{\alpha_g} g^{-1}$.
We define a mapping $\sigma : \calX \times \calA \to \C[G_n]$ from question-answer pairs to self-adjoint projections in $\C[G_n]$ as follows.
For $(a_0, a_1) \in \Z_2^2$, let $\#(a_0,a_1)$ be the element of $[4]$ with binary representation $(a_0, a_1)$.
\begin{itemize}
 	\item When $x \in \calX_{var}$,
	\begin{align*}
		\sigma(x, a) = 
		\begin{cases}
			\frac{e + (-1)^{a_2} x}{2} & \text{ if } (a_0, a_1) = (0,0) \\
			0 & \text{ otherwise}
		\end{cases}.
	\end{align*}
	\item When $x = i \in [m]$,
	\begin{align*}
		\sigma(i, a) = \prod_{k \in I_i} \frac{e + (-1)^{a_{\phi_i(k)}} x_{k}}{2}.
	\end{align*}
 	\item When $x \in \set{m, m+1, m+2}$,
	\begin{align*}
		\sigma(x, a) = \begin{cases}
			\pi_0^{(\#(a_0,a_1))} & \text{ if } x = m  \text{ and }  \#(a_0,a_1) \leq 2, a_2 = 0\\
			\pi_1^{(\#(a_0,a_1))} & \text{ if } x = m+1 \text{ and }  \#(a_0,a_1) \leq 2, a_2 = 0\\
			\pi_2^{(\#(a_0,a_1))} & \text{ if } x = m + 2 \text{ and }  \#(a_0,a_1) \leq 2, a_2 = 0\\
			0 &\text{ otherwise} \\
		\end{cases},
	\end{align*}
	where $\pi_i^{(a)} \in \C[\ip{t_1, t_2}] \cong \C[D_{p(n)}]$ is defined in \cref{sec:dihedral}, \cref{eq:idem}.
	\item When $x = (m, t_i)$ for $i = 1,2$,
	\begin{align*}
		&\sigma((m, t_1), (a_0,a_1,a_2))
		=\begin{cases}
			\pi_0^{(\#(a_0,a_1))}  \left( \frac{e + (-1)^{a_2} t_1}{2} \right) &\text{ if } \#(a_0,a_1) < 3 \\
			0 &\text{ otherwise}
		\end{cases} \text{ and } \\
		&\sigma((m, t_2), (a_0,a_1,a_2))
		=\begin{cases}
			\pi_0^{(\#(a_0,a_1))}  \left(\frac{e + (-1)^{a_2} t_2}{2} \right)&\text{ if }  \#(a_0,a_1) < 3 \\
			0 &\text{ otherwise }
		\end{cases},
	\end{align*}
	where again $\pi_0^{(a)} \in \C[\ip{t_1, t_2}] \cong \C[D_{p(n)}]$ is defined in \cref{sec:dihedral}, \cref{eq:idem}.
 \end{itemize}
If $z = \sum_g \alpha_g g$, let
\begin{align*}
\supp(z) = \set{ g \in G_n \mid \alpha_g \neq 0 }. 
\end{align*}
Define
\begin{align*}
	&W_n =  \bigcup_{x, y\in \calX, a, b \in \calA} \supp( \sigma(x,a)\sigma(y,b) ),
\end{align*}
and let 
\begin{align*}
	\calF_n = \set{f : W_n \to \set{0,1} \mid f(e) = 1, f(x_0) =0, f(g) = 0 \text{ for } g \in \ip{t_1, t_2} \setminus \set{e}}.
\end{align*}
Functions $f: W_n \to \set{0,1}$ can be regarded as linear functions $\spn_{\C}(W_n) \to \C$
by extending linearly.
Hence, given a function $f \in \calF_n$, we can define a bipartite correlation $C_f$ for the scenario $(\calX, \calX, \calA , \calA )$
by $C_f(a , b| x, y) = f( \sigma(x,a) \sigma(y,b))$ (we stretch the term ``correlation" here since $C_f$ may have negative entries).
These correlations contain a copy of the correlation $\fC_{p(n)}$ from \cref{sec:dihedral}.
\begin{proposition}
	\label{prop:fcp_sub}
	
	Let $Q = \{t_1$, $t_2$, $m$, $m+1$, $m+2$, $(m, t_1)$, $(m ,t_2)\} \subseteq \calX$
	and let $I = \{t_1$, $t_2$, $0$, $1$, $2$, $(0,t_1), (0, t_2)\}$ as in \cref{sec:dihedral}.
	Let $\alpha:  Q \to I$ be the bijection 
	$\alpha(t_i) = t_i$, $\alpha((m,t_i)) = (0,t_i)$ for $i = 1, 2$,
	and $\alpha(m+j) =  j$ for $j \in [3]$.
	If $x,y \in Q$ and $a,b \in \calA$ such that $\#(a_0,a_1),\#(b_0,b_1) < 3$, then
	\begin{align}
		\label{eq:fcp_sub}
		C_f (a,b | x , y) = \fC_{p(n)}( (\#(a_0,a_1), a_2), (\#(b_0,b_1), b_2) | \alpha(x), \alpha(y) ),
	\end{align}
	for all $f \in \calF_n$.
	Furthermore, 
	if $S = (\ket{\psi} \in \calH, \set{M_x^{(a)}, N_x^{(a)} \mid a \in \calA}, x \in \calX)$ is a good strategy for $C_f$,
	then $S' = (\ket{\psi} \in \calH, \set{\tM_x^{(a)}, \tN_x^{(a)} \mid a \in [3] \times  [2]}, x \in I)$ is a good strategy for 
	$\fC_{p(n)}$, where 
	\begin{align*}
		\tM_x^{(\#(a_0, a_1), a_2)}  = M_{\alpha^{-1}(x)}^{(a)} \text{  and  } \tN_x^{(\#(a_0, a_1), a_2)}  = N_{\alpha^{-1}(x)}^{(a)}  
		\, \text{  for  }x \in I,  a \in \calA, \#(a_0, a_1)  < 3.
	\end{align*}
\end{proposition}
\begin{proof}
	Since the subgroup generated by $t_1$ and $t_2$ in $G_n$ is a parabolic subgroup,
	$\ip{t_1, t_2} \subseteq G_n$ is isomorphic to $D_{p(n)}$.
	By construction, if $f \in \calF_n$ and $g \in \ip{t_1, t_2}$ then $f(g) = 1$ if and only if $g = e$.
	Hence
	when $x,y \in Q$ and $\#(a_0,a_1),\#(b_0,b_1) < 3$,
	\begin{align*}
		f(\sigma(x, a) \sigma(y, b)) =\bra{e} L(\sigma(x, a)) R(\sigma(y,b)) \ket{e},
	\end{align*}
	where $\ket{e} \in \ell^2 D_{p(n)}$ and $L: \C[D_{p(n)}] \to \calU(\ell^2 D_{p(n)})$ and 
	$R: \C[D_{p(n)}] \to \calU(\ell^2 D_{p(n)})$ are the left and right regular representations 
	of $\C[D_{p(n)}]$. Since $L(\sigma(x,a)) = \tM_{\alpha(x)}^{(\#(a_0,a_1),a_2)}$ and $R(\sigma(y,b)) =
	\tN_{\alpha(y)}^{(\#(b_0,b_1),b_2)}$ from the definition of $\fC_p$ in
	\cref{sec:dihedral}, \cref{eq:fcp_sub} follows.
	
	If $a \in \calA$ with $\#(a_0, a_1) = 3$, then
	$C_f(a,b|x,y) = C_f(b, a| y, x) = 0$ for all $x \in Q$, $y \in \calX$ and 
	$b\in \calA$.
	Hence $\bra{\psi} M_x^{(a)} \ket{\psi} = \bra{\psi} N_x^{(a)} \ket{\psi} = 0$ for all $x \in Q$.
	Since $S$ is a good strategy, $M_x^{(a)} = N_x^{(a)} = 0$.
	We conclude that $\set{ \tM_x^{(a)} \mid a \in [3] \times [2]}$ and $\set{ \tN_x^{(a)} \mid a \in [3] \times [2]}$
	are projective measurements for all $x \in I$, and thus $S'$ is a good strategy for $\fC_{p(n)}$.
\end{proof}

Finally, we are ready to define $F_n$: 
\begin{align*}
F_n = \set{ C_f \mid f \in \calF_n \text{ such that } C_f|_{\calA \times \calA \times (\calX_{var} \cup [m]) \times (\calX_{var} \cup [m])} 
\text{ is a perfect correlation for } Ax =0},
\end{align*}
where 
$C_f|_{\calA \times \calA \times (\calX_{var} \cup [m]) \times (\calX_{var} \cup [m])}$ is the restriction of $C_f$
 to the question set $\calX_{var} \cup [m]$,
and $A$ is the matrix fixed above.
Since support sets are finite, $W_n$ and $\calF_n$ are finite, and hence $F_n$ is finite.
Since the word problem of $G_n$ is decidable, the sets $W_n$ and $\calF_n$ are computable from $n$.
Whether
$C_f|_{\calA \times \calA \times (\calX_{var} \cup [m]) \times (\calX_{var} \cup [m])}$ is a perfect correlation for $Ax =0$ is also decidable, 
and therefore $F_n$ is computable from $n$.

\begin{proposition}
	\label{prop:not_in_cqc}
	If $n \in X$, $F_n \cap C_{qc}(\calX, \calX, \calA, \calA) = \emptyset$.
\end{proposition}
\begin{proof}
Assume $C_{f} \in F_n \cap C_{qc}(\calX,\calX, \calA, \calA)$ for some $f$.
By \cref{prop:good_strat}, there is a good commuting-operator strategy
\begin{align*}
	S = (\ket{\psi} \in \calH, \set{M_x^{(a)}, N_x^{(a)} \mid a \in \calA}, x \in \calX)
\end{align*}
for $C_f$.
Since $C_f|_{\calA \times \calA \times (\calX_{var} \cup [m]) \times (\calX_{var} \cup [m])}$ is a perfect correlation
for $Ax = 0$, and $S' =  (\ket{\psi}, \set{M_x^{(a)}, N_x^{(a)} \mid a \in \calA}, x \in \calX_{var} \cup [m])$ is a strategy for
this perfect correlation, \cref{prop:correlationrep}
states that there exists a  subspace $\calH_0$ of $\calH$ containing $\ket{\psi}$ and unitary operators
$M(x_i)$ and $N(x_i)$ for $x_i \in \calX_{var}$ on $\calH$ inducing commuting
 representations $\Phi_M$ and $\Phi_N$ of $\Gamma(A)$ on $\calH_0$.
 The operators $M(x_i)$ and $N(x_i)$ are defined as 
 \begin{align*}
 	M(x_i) = M_{x_i}^{(0,0,0)} - M_{x_i}^{(0,0,1)} \text{ and } N(x_i) = N_{x_i}^{(0,0,0)} - N_{x_i}^{(0,0,1)}. 
 \end{align*}
By \cref{prop:existsg}, the generators $t_1, t_2, u_1$ and $u_2$ of $\Gamma(A)$ satisfy the relation $u_2u_1 t_1t_2 u_1 u_2 = (t_1t_2)^r$,
so
\begin{align*}
	(M(t_1)M(t_2)) (M(u_1)M(u_2))\ket{\psi}  &= (M(u_1)M(u_2)) (M(t_1)M(t_2))^{r} \ket{\psi} \text{ and }\\
	(N(t_1)N(t_2)) (N(u_1)N(u_2))\ket{\psi}  &= (N(u_1)N(u_2)) (N(t_1)N(t_2))^{r} \ket{\psi}.
\end{align*}

Let $S' = (\ket{\psi} \in \calH, \set{\tM_x^{(a)}, \tN_x^{(a)} \mid a \in [3] \times [2]}, x \in I)$ be the strategy for $\fC_{p(n)}$
from \cref{prop:fcp_sub}.
Notice that $\tM_{t_i} := \tM_{t_i}^{(0,0)} - \tM_{t_i}^{(0,1)} = M_{t_i}^{(0,0,0)} - M_{t_i}^{(0,0,1)} = M(t_i)$ for $i = 1,2$.
Similarly $\tN_{t_i} := \tN_{t_i}^{(0,0)} - \tN_{t_i}^{(0,1)}  = N(t_i)$ for $i = 1, 2$.
Let $U_A = M(u_1)M(u_2) $ and $U_B = N(u_1)N(u_2)$.
By \cref{prop:correlationrep}, $M(x_i) N(x_i) \ket{\psi} = \ket{\psi}$, so
  $U_A U_B \ket{\psi} = \ket{\psi}$.
Hence 
$U_A$ and $U_B$
satisfy the conditions of \cref{thm:fcp} with the strategy $S'$.
We conclude that
\begin{align*}
		\Phi_M( (t_1t_2)^{p(n)}) \ket{\psi} =  (M(t_1) M(t_2))^{p(n)} \ket{\psi} = \ket{\psi}.
\end{align*}
By part (3) of \cref{prop:correlationrep}, 
$\Phi_M$ descends to a representation of $\Gamma(A)/\ip{(t_1t_2)^{p(n)} =e }$ on $\calH_0$.
On the other hand, 
\begin{align*}
	\bra{\psi} M(x_0) \ket{\psi} = \bra{\psi} M_{x_0}^{(0,0,0)} -M_{x_0}^{(0,0,1)} \ket{\psi}
	= f(\sigma(x_0, (0,0,0)) - \sigma(x_0, (0,0,1))) = f(x_0) = 0
\end{align*}
by the definition of $\calF_n$.
Hence $\Phi_M(x_0) \neq \1_{\calH_0}$,
which implies that 
$x \neq e$ in $\Gamma(A)/\ip{(t_1t_2)^{p(n)} =e }$. 
By part (b) of \cref{prop:existsg}, $n \notin X$.
\end{proof}

\begin{proposition}
	\label{prop:in_cqa}
	If $n \notin X$, $F_n \cap C_{qa}(\calX, \calX, \calA, \calA) \neq \emptyset$.
\end{proposition}
\begin{proof}
	For this proof, let $\Gamma_n := \Gamma(A)/\ip{(t_1t_2)^{p(n)} = e}$ and 
	let $W_n^{fa} = \set{ g \in W_n \mid g \neq e \text{ in } \Gamma_n^{fa}}$.
	Suppose $n \notin X$,
	and define $f : W_n \to \C$ by $f(g) = 0$ if $g \in W_n^{fa}$ and $f(g) = 1$
	if $g \in W_n \setminus W_n^{fa}$. 
	(The function $f$ is the pullback of the canonical trace on $\Gamma_n^{fa}$ to $G_n$ and restricted to $W_n$.)
	Since $n \notin X$, $x_0\in W_n^{fa}$ by part (c) of \cref{prop:existsg}, so $f(x_0) = 0$.
	By parts (d) and (e) of  \cref{prop:existsg}, 
	the set $\ip{t_1, t_2} \setminus \set{e} \subseteq W_n^{fa}$ as well.
	Thus $f \in \calF_n$. 
	
	\newcommand{\tsigma}{\tilde{\sigma}}
	In the rest of proof, we show that $C_f \in C_{qa}(\calX, \calX, \calA, \calA)$.
	Let $\pi: \calF(\calX_{var}) \to G_n$ be the quotient homomorphism, and let $\omega: G_n \to \calF(\calX_{var})$
	be a right inverse for $\pi$, i.e.\ some function such that $\pi \circ \omega$ is the identity for $G_n$.
	For each $x \in \calX$ and $a \in \calA$, let $\tsigma(x,a) = \omega(\sigma(x,a)) \in \C[\calF(\calX_{var})]$, 
	where $\omega$ is extended linearly to a function $\C[G_n] \to \C[\calF(\calX_{var})]$, so that
	$\pi(\tsigma(x,a)) = \sigma(x,a)$.
	Let 
	\begin{align*}
		\tilde{W}_n = \bigcup_{x,y \in \calX a,b \in \calA}  \supp( \tsigma(x,a) \tsigma(y,b))
	\end{align*}
	and 
	\begin{align*}
	\tilde{W}_n^{fa} = \set{ g \in \tilde{W}_n \mid g \neq e \text{ in } \Gamma_n^{fa}}.
	\end{align*}
	Note that if $w \in \tilde{W}_n$, then $\pi(w) \in W_n$, and $w \in \tilde{W}_n^{fa}$
	if and only if $\pi(w) \in W_n^{fa}$.
	The polynomial $\tsigma(x,a)^2 - \tsigma(x,a)$ is not necessarily $0$ in $\C[\calF(\calX_{var})]$,
	but is $0$ in $\C[\Gamma_n]$, and the same is true for the polynomials 
	$\tsigma(x,a)^\ast - \tsigma(x,a)$, $a \in \calA$, $x \in \calX$,
	$\tsigma(x,a)\tsigma(x,b)$, $a \neq b$, $x \in \calX$,
	and $\sum_{a \in \calA} \tsigma(x,a) - e$, $x \in \calX$.
	By \cref{lm:norms}, there is a constant $c$ such that for any $\ep$-approximate representation $\rho: \calF(\calX_{var}) \to \calU(\C^d)$ of $\Gamma_n$,
	we have 
	\begin{equation*}
		\norm{ \rho(\tsigma(x,a))}_{op} \leq c, \quad
		\norm{ \rho(\tsigma(x,a)^2 - \tsigma(x,a)) } \leq c\ep,  \quad
		 \norm{ \rho(\tsigma(x,a)^\ast - \tsigma(x,a)) } \leq c\ep,
	\end{equation*}
	\begin{equation*}
		\norm{ \rho( \tsigma(x,a) \tsigma(x,b) )} \leq c\ep ,  \text{ and }
		  \norm{ \rho( \sum_{a' \in \calA} \tsigma(x,a') - e )}\leq c\ep
	\end{equation*}
	for all $x \in \calX$ and $a \neq b \in \calA$.
	By \cref{prop:tensor_trick},
	for any $\ep, \zeta >0$
	there is an $\ep$-approximate representation $\rho : \calF(\calX_{var}) \to \calU(\C^d)$ of $\Gamma_n$,
	where $d$ depends on $\ep$ and $\zeta$,
	such that $0 \leq \tTr(\rho(w)) \leq \zeta$ for each $w \in W_n^{fa}$,
	and $1 - \zeta \leq \tTr(\rho(w)) \leq 1$ for each $w \in W_n \setminus W_n^{fa}$.
	For $x \in \calX$ and $a \in \calA$, let 
	\begin{align*}
		\tM_x^a = \rho(\tsigma(x,a)).
	\end{align*}
	Let $\norm{\cdot}_1$ denote the $1$-norm on $\C[\calF(\calX_{var})]$ and $\C[G_n]$, so if $\alpha = \sum_g u_g g$ then
	$\norm{\alpha}_1 = \sum_g \abs{u_g}$.
	It is not hard to see that $\norm{\sigma(x,a)}_1 \leq 4$ for all $x \in \calX$, $a \in \calA$, and hence $\norm{\tsigma(x,a)}_1 \leq 4$ for all $x \in \calX$, $a \in \calA$ as well.
	Since the $1$-norm is submultiplicative in group algebras, we see that $\norm{\tsigma(x,a)\tsigma(y,b)}_1 \leq 16$ for all $x,y \in \calX$ and $a,b \in \calA$.
	Hence if we write $\tsigma(x,a)\tsigma(y,b) = \sum_{g \in \tilde{W}_n} u_g g$
	for some $u_g \in \R$, then
	\begin{align*}
		\abs{ C_f(a,b \mid x,y) - \tTr( \tM_x^a \tM_y^b)} &\leq 
		\sum_{g \in \tilde{W}_n} \abs{u_g} \abs{ f(\pi(g)) - \tTr(\rho(g))} \\
		&=
		\sum_{g \in \tilde{W}_n^{fa}} \abs{u_g} \abs{0 -  \tTr(\rho(g))} 
		+ \sum_{g \in \tilde{W}_n \setminus \tilde{W}_n^{fa}} \abs{u_g} \abs{1 -  \tTr(\rho(g))} \\
		&\leq  \norm{\tsigma(x,a)\tsigma(y,b)}_1 \zeta \leq 16\zeta.
	\end{align*}
	Unfortunately, $\set{ \tM_x^a \mid a \in \calA}$ may not be a measurement.
	However, 
	by \cref{lm:round_prj} there are projective measurements $\set{ M_x^a \mid a \in \calA }$, $x \in \calX$,
	such that 
	\begin{align*}
		\norm{ M_x^a - \tM_x^a } \leq \Delta(c, 8) c\ep
	\end{align*}
	for all $x \in \calX$ and $a \in \calA$ (where $8$ comes from the size of $\calA$).
	Then
	\begin{align*}
		\abs{ \tTr( M_x^a M_y^b) - \tTr( \tM_x^a \tM_y^b) }
		&\leq \norm{ M_x^a M_y^b -  \tM_x^a \tM_y^b} \\
		&\leq \norm{ M_x^a}_{op} \norm{M_y^b - \tM_y^b} + \norm{ \tM_y^b}_{op} \norm{M_x^a - \tM_x^a} \\
		&\leq (1 + c) \cdot \Delta(c,8) c\ep.
	\end{align*}
	
	Let $\ket{\psi} = 1/ \sqrt{d} \sum_{i = 1}^d \ket{i}\x\ket{i} \in \C^d \x \C^d $, where $\ket{i}$ is the i-th standard basis of $\C^d$,
	and let $N_x^a = (M_x^a)^T$.
	Since $ \bra{\psi} M_x^a \x N_y^b \ket{\psi} = \tTr(M_x^a M_y^b)$, we conclude that  
	\begin{align*}
		\abs{ \bra{\psi} M_x^a \x N_y^b \ket{\psi} - C_f(a,b \mid x,y)} &\leq 
		\abs{ \tTr( M_x^a M_y^b) - \tTr( \tM_x^a \tM_y^b) } + \abs{ \tTr( \tM_x^a \tM_y^b) - C_f(a,b\mid x,y)} \\
		&\leq (1 + c) \cdot \Delta(c,8) c\ep + 16 \zeta.
	\end{align*}
	The correlation defined by $(\ket{\psi}, \set{ M_x^a \mid a \in \calA},\set{N_x^a \mid a \in \calA}, x \in \calX)$ belongs to $C_q(\calX,\calX, \calA,\calA)$.
	Since $\ep$ and $\zeta$ can be arbitrarily small, we see that $C_f \in  \overline{ C_q(\calX,\calX, \calA,\calA)} = C_{qa}(\calX, \calX, \calA, \calA) $.
	
\end{proof}
The proof of \cref{thm:cqa} follows immediately from \cref{prop:not_in_cqc,prop:in_cqa}.
 
\bibliographystyle{alphaurl}
\bibliography{dissertation}

\end{document}